%% file: main.tex
\documentclass[11pt]{article}
\usepackage[thmnum]{preamble}
\usepackage{authblk}
\usepackage{fullpage}

\theoremstyle{plain}

\newcommand{\localCharacterization}{\hyperref[thm:main_poly_characterization]{Local Characterization Theorem}\xspace}

\newcommand{\discreteLocalCharacterization}{\hyperref[thm:localchar_discrete]{Discrete Local Characterization Theorem}\xspace}

\begin{title}
{\LARGE Low Degree Testing over the Reals}
\end{title}

\author[1]{Vipul Arora}
\affil{\footnotesize National University of Singapore. $\{$\texttt{vipul, arnab}$\}$\texttt{@comp.nus.edu.sg}.}
\author[1]{Arnab Bhattacharyya\thanks{Supported in part by an NRF Tier 2 grant (MOE2019-T2-1-152). Research partly conducted while visiting the Simons Institute for the Theory of Computing.}}
\author[2]{Noah Fleming\thanks{Supported by NSERC.}}
\affil{\footnotesize University of California, San Diego, and Memorial University. \texttt{nfleming@mun.ca}}
\author[3]{Esty Kelman\thanks{Supported in part by an Amazon Faculty Research Award to AB, and in part by ERC grant 834735. Research partly conducted while visiting the Simons Institute for the Theory of Computing.}}
\affil{\footnotesize Hebrew University, and Reichman University. \texttt{esther.kelman@post.idc.ac.il}.}
\author[4]{Yuichi Yoshida\thanks{Supported in part by JSPS KAKENHI Grant Number JP17H04676, and 20H05965.}}
\affil{\footnotesize National Institute of Informatics. \texttt{yyoshida@nii.ac.jp}.}

\ignore{
\author{Vipul Arora\thanks{National University of Singapore. \texttt{vipul@comp.nus.edu.sg}.}
\and
Arnab Bhattacharyya\thanks{National University of Singapore.
\texttt{arnabb@nus.edu.sg}.}
\and
Noah Fleming\thanks{University of California, San Diego, and Memorial University. \texttt{nfleming@eng.ucsd.edu}.}
\and
Esty Kelman\thanks{Hebrew University, and Reichman University. \texttt{esther.kelman@post.idc.ac.il}.}
\and Yuichi Yoshida\thanks{National Institute of Informatics. \texttt{yyoshida@nii.ac.jp}.}
}
}
\date{\today}

\begin{document}
\maketitle

\ignore{
\setlength\tabcolsep{0em}
\newcommand{\myPad}{\hspace{2.em}}
 \begin{tabular}{c@{\myPad}c@{\myPad}c}
   Vipul Arora &
 	Arnab Bhattacharyya  &
 	Noah Fleming$^\dagger$ \\[-.5mm]
   	\small\slshape National University of Singapore &
 	\small\slshape National University of Singapore &
 	\small\slshape UC San Diego \\[-1mm]
 	 & & \small\slshape Memorial University
 \end{tabular}\vspace{1em}
 \begin{tabular}{c@{\myPad}c}
   Esty Kelman &
 	Yuichi Yoshida
 	\\[-.5mm]
   	\small\slshape Hebrew University & \small\slshape National Institute of Informatics \\[-1mm]
 	\small\slshape Reichman University &
 \end{tabular}

\vspace{9mm}
}

\begin{abstract}
We study the problem of testing whether a function $f\colon \reals^n \to \reals$ is a polynomial of degree at most $d$ in the \emph{distribution-free} testing model. Here, the distance between functions is measured with respect to an unknown distribution $\mathcal{D}$ over $\reals^n$ from which we can draw samples. In contrast to previous work, we do not assume that $\mathcal{D}$ has finite support. 
    
We design a tester that given query access to $f$, and sample access to $\mathcal{D}$, makes $\poly(d/\eps)$ queries to $f$, accepts with probability $1$ if $f$ is a polynomial of degree $d$, and rejects with probability at least $2/3$ if every degree-$d$ polynomial $P$ disagrees with $f$ on a set of mass at least $\eps$ with respect to $\mathcal{D}$. 
Our result also holds under mild assumptions when we receive only a polynomial number of bits of precision for each query to $f$, or when $f$ can only be queried on rational points representable using a logarithmic number of bits. Along the way, we prove a new stability theorem for multivariate polynomials that may be of independent interest. 
\end{abstract}

    

\input{intro}
\input{prelims.tex}

\input{exact}

\input{local_to_global_linear_fucntions}
\input{exact-discrete}

\bibliographystyle{alpha}
\bibliography{biblio}

\begin{appendices}
\input{appendix}

\input{tolerant_addativity}
\end{appendices}
\end{document}

%% file: intro.tex

\section{Introduction}

Traditionally, {program testing} involves running a suspect program on a curated test set and checking the validity of the results. 
To formalize and quantitatively study this problem, Blum, Luby, and Rubinfeld~\cite{BlumLR93} initiated research on {\em self-testers}, which check a particular property of the
given program by verifying whether the program's output on a random input is consistent with its outputs on other related inputs. Soon afterwards, spurred by connections to 
the newly emerging areas of interactive proof systems and probabilistically checkable  proofs, self-testing blossomed into the general area of \emph{property testing}; see the textbooks~\cite{Goldreich17,ArnabYoshida22} for detailed introductions.
Perhaps because of these early connections to complexity theory and coding theory, the standard setup in property testing is to assume that both the domain and range of the
function being tested are finite sets. 

In this work, we return to the roots of property testing and consider testing properties of real-valued functions $f: \R^n \to \R$ with real-valued inputs. Specifically, we 
focus on the fundamental problem of \emph{low-degree testing} which has been widely and intensely studied in the standard setup. Recall that in the traditional setup we
are given query access to a function $f\colon \mathbb{F}^n \to \mathbb{F}$ defined over some finite field $\mathbb{F}$ and a parameter $\varepsilon >0$. 
The aim of a property tester for these parameters is to distinguish with probability at least $2/3$ between the case when $f$ is a polynomial of degree at most $d$,
 and the case when $f$ is $\eps$-far, i.e., it disagrees with any polynomial $P$ of degree at most $ d$ on at least an $\eps$ fraction of the domain $\mathbb{F}^n$.
 
 To extend the notion of testing to functions defined over $\R^n$, we need a notion of ``$\eps$-farness'' in this setting. One approach
 is to fix a specific distribution $\caD$ over $\R^n$, and define $f$ to be $\eps$-far from a property $\mathcal{P}$ if:
 \begin{equation}\label{eq:far}
 \inf_{g \in \mathcal{P}}\Pr_{\bm x \sim \caD}[f(x) \neq g(x)] > \eps.
 \end{equation}
 Indeed, for $\caD$ being the standard Gaussian distribution, this is the approach used by most prior work on testing properties of functions over the reals, e.g.,
 testing halfspaces~\cite{MatulefORS10, MatulefORS102, matulef2009testing, Harms19}, surface area~\cite{Neeman14, KothariNOW14}, high-dimensional convexity~\cite{FSS17},  linear separators~\cite{BalcanBBY12}, and linear $k$-juntas~\cite{DeMN19}. However, this approach is not entirely satisfactory, as the assumed $\caD$ may not be the relevant underlying input distribution. 

A different approach is to use the framework of {\em distribution-free testing}, studied first by Halevy and Kushilevitz~\cite{HalevyK07}, that does not assume knowledge of $\caD$. Instead,
it is only assumed that the tester receives sample access to the underlying distribution $\caD$, and the goal is to reject when~\eqref{eq:far} holds. Distribution-free testing has been widely studied for a variety of properties of boolean functions, e.g., monomials~\cite{GlasnerS07, DolevR11}, juntas~\cite{LiuCSSX19}, halfspaces~\cite{ChenX16, CP22}, and monotonicity~\cite{BCS20}. Distribution-free property testing over $\R^n$ is an emerging trend in the field, that has been studied, e.g., for monotonicity~\cite{BCS20, HY20}, halfspaces~\cite{Harms19} and polynomial threshold functions~\cite{BFH21}.  Most directly relevant here is the work of Fleming and Yoshida~\cite{FlemingY20} where they studied distribution-free testing of linearity of functions $f: \R^n \to \R$. 
 
To further discuss testing real functions, we first formally define distribution-free testing of real functions.
For a property $\mathcal P$ over real functions, we say that an algorithm is a \emph{tester} for $\mathcal P$ if, given query access to a function $f: \mathbb{R}^n \to \mathbb{R}$, and sampling access to an unknown distribution $\mathcal{D}$, and $\eps > 0$, it distinguishes the case that $f$ satisfies $\mathcal P$, from the case that $f$ is {\em $\eps$-far from $\mathcal P$ over $\mathcal{D}$}, i.e., for any function  $g:\mathbb{R}^n \to \mathbb{R}$ satisfying $\mathcal P$, \[\Pr_{\bm x \sim \mathcal{D}}[f(\bm x) \neq g(\bm x)]>\eps\] holds.
We say that a tester is a \emph{one-sided error tester}, if it always accepts functions satisfying $\mathcal{P}$.
We also explore testing in the presence of errors. In this context, the early works~\cite{ABCG93, GLRSW91} introduced the notion of {\em approximate testing}, which
was made more formal by the work of Ergun, Kumar and Rubinfeld~\cite{EKR01}. Given two parameters $\alpha<\beta$, in addition to $\eps>0$, query access to $f: \R^n \to \R$ and sample
access to a distribution $\caD$, the goal of an approximate tester for a property $\mathcal{P}$ is to distinguish between the following two cases:
\begin{itemize}
\item \textbf{YES:}
There exists $h \in \mathcal{P}$, such that $|f(\bm x)-h(\bm x)| < \alpha$ for all $\bm x\in\R^n$.
\item \textbf{NO:}
For every $h \in \mathcal{P}$, $\Pr_{\bm x \sim \caD}[|f(\bm x)-h(\bm x)| >\beta]>\eps$. 
\end{itemize}

In the \textbf{YES} case, we say that $f$ is \emph{pointwise $\alpha$-close to $h$}. Here, $\alpha$ should be thought of as a representational limitation, or a round-off/truncation error. For example, $\alpha = 1/\exp(\mathrm{poly}(n))$ can be achieved by storing $\mathrm{poly}(n)$ bits of precision.


\subsection{Our Contributions}

Our first result gives an exact tester for low-degree that generalizes the result of~\cite{FlemingY20}. Note that there is a trivial $\Omega(\max\{d, 1/\varepsilon\})$ lower bound on the complexity of testing degree-$d$ polynomials. 

\begin{restatable}{thm}{exact}\label{thm:low_degree_test}
Let $d\in\mathbb{N}$, and for $L>0$, suppose $f: \R^n \to \R$ is a function that is bounded in the ball $\ball(\bm{0}, L)$. 
Given $\eps >0$, query access to $f$, and sampling access to an unknown distribution $\mathcal{D}$, there exists a one-sided error, distribution-free,  $O (d^5+\frac{d^2}{\varepsilon}\log\frac{1}{\varepsilon} )$-query tester for testing whether $f$ is a degree-$d$ polynomial, or is $\eps$-far from degree-$d$ polynomials over $\mathcal{D}$.
\end{restatable}
\noindent 
Some form of the boundedness condition 
is necessary to test low degree using standard functional equation characterizations. Even for linearity, Hamel \cite{Ham05} showed the existence of functions $f: \R \to \R$ that satisfy the {\em Cauchy functional equation} $f(x+y)=f(x)+f(y)$ everywhere but are unbounded on {\em any} measurable set\footnote{In fact, Hamel showed that if $f$ is a non-linear solution, the set $\{(x,f(x))\}$ intersects every neighborhood of every point in $\R \times \R$, and so is clearly unbounded on any measurable set.}. On the other hand, Cauchy \cite{Cau21} showed that the only continuous solutions to $f(x+y)=f(x)+f(y)$ are the linear maps $f(x) = cx$. Darboux \cite{Dar75} later showed that boundedness on any interval is a weaker condition than continuity that also implies the result of Cauchy. The latter two results were generalized to low-degree polynomials by Fr\'echet \cite{Fre09} and Ciesielski \cite{Ciesielski1959SomePO} respectively. \anote{This is not a formal argument. For example, Hamel's counterexample would be rejected if the tester checked $f(ax)=af(x)$ for real values $a$.}


\ignore{
 \spcnoindent \emph{A note on the analyticity assumption.} \nnote{move this later and put a forward pointer to it} Our \localCharacterization assumes that the given function is analytic. This circumvents discontinuity issues which arise when moving from countable domains, such as the rationals. A simple example of such an issue is that there are functions $f$ which are additive ($f(x+y) = f(x)+f(y)$) but are not homogeneous (there exist $\alpha$ and $x$ such that $f(\alpha x) \neq \alpha f(x)$). Furthermore, for any linear function one can construct (assuming the axiom of choice) dense families of additive but not-homogeneous functions which are close and far from that linear function. The same assumption was used in~\cite{FlemingY20} in order to obtain a linearity tester from their tester for additivity. For us, the proof of the \localCharacterization relies on the Taylor series expansion of $g$ being convergent locally, which is guaranteed by analyticity. \\}

\autoref{thm:low_degree_test} serves as a starting point for our investigation into approximate low-degree testing. In this setting, we give an approximate tester for low-degree polynomials, where the unknown underlying distribution $\mathcal{D}$ is required to be {\em $(\eps, R)$-concentrated}. We say that a distribution is $(\eps, R)$-concentrated if most of its mass is concentrated in a ball of radius $R$, that is, \[
\Pr_{\bm p \sim \mathcal{D}}[\bm p\in\ball(\bm 0,R)] \geq 1- \eps.
\]
Note that the standard Gaussian distribution is $(0.01,2\sqrt{n})$-concentrated.
\begin{restatable}{thm}{approximateLowDegree}\label{thm:approximate_low_degree_tester}
Let $d\in\mathbb{N}$, $f: \R^n \to \R$ be a function that is bounded in  $\ball(\bm{0}, 2d\sqrt{n})$, and for $\eps \in (0,1), R>0$, let $\caD$ be an $(\eps/4, R)$-concentrated distribution. 
Given $\alpha>0, \beta\geq 2^{(2n)^{O(d)}}R^d\alpha$, query access to $f$, and sampling access to $\mathcal{D}$, there is a one-sided error, $O (d^5+\frac{d^2}{\varepsilon} \log\frac{1}{\varepsilon} )$-query tester which,
 distinguishes between the case when  $f$ is pointwise $\alpha$-close to some degree-$d$ polynomial and the case when, for every degree-$d$ polynomial $h\colon\mathbb{R}^n\to\mathbb{R}$, $\Pr_{\bm p \sim \caD}[|f(\bm p)-h(\bm p)| > \beta ] > \eps$.
\end{restatable}

Thus, if $d$ is constant, $R$ is polynomial in $n$, and the tester receives $\poly(n)$ most significant bits of $f({\bm p})$ for any query point $\bm{p}$, the tester accepts when $f$ is a degree-$d$ polynomial, and rejects when $f$ is not pointwise $1$-close to a degree-$d$ polynomial on at least an $\eps$ fraction of $\mathcal{D}$.
In \autoref{sec:additivity}, we consider the  special case of testing additivity. Here, we give a tester which requires only $O(\log n)$ bits of precision.

The above results assume that the function can be queried on arbitrary points in $\mathbb{R}^n$ which is unrealistic in view of finite precision issues. We also analyze the setting where the tester can evaluate $f$ only on points with finite number of bits of precision and also, the unknown distribution $\caD$ is promised to be supported on points with finite number of bits of precision. More precisely, $\caD$ is given to be supported on points of the lattice  $\caL\triangleq\frac1B\Z^n$, for some  parameter $B$ controlling the density of the lattice, and
also, $f$ can be queried only on a lattice $\caL' \triangleq \frac{1}{B'}\Z^n$ for a bounded $B'$.
This setting models the situation where we only care about the function's behavior on finitely representable inputs, and on such  inputs, the function can be evaluated exactly. The goal is to obtain a tester that does not require $B'$ to be very large but still allows $B$ to be large. 

\begin{restatable}{thm}{discretesampling}\label{thm:low_degree_test_discrete}
For $d, B, R>0$, let $B'\geq 16\max\{n^{5/2+2d} d^{2d},B^2R^2/\sqrt{n}\}$ be a multiple of $B$. Let $\caL = \frac{1}{B}\mathbb{Z}^n$ and $\caL' = \frac{1}{B'} \mathbb{Z}^n$.
Given $\eps>0$, query access to a function $f: \R^n \to \R$, and sample access to an unknown $(\eps/4, R)$-concentrated distribution $\caD$ supported on $\caL$, there is a one-sided error,  $O(d^5+\frac{d^2}{\varepsilon}\log\frac{1}{\varepsilon})$-query tester for testing whether $f$ agrees with a degree-$d$ polynomial on $\mathcal{L}$, or is $\eps$-far from degree-$d$ polynomials over $\mathcal{D}$. The tester queries $f$ on points in $\caL'$. 
\end{restatable}

Note that unlike \autoref{thm:approximate_low_degree_tester} and \autoref{thm:low_degree_test}, our tester for lattices does not make any assumptions about the function $f$. 

\ignore{
\nnote{define $(\zeta,R)$-concentrated distribution as a distribution such that at least $(1-\zeta$ fraction of its prob mass is in $\ball(\bm c, R)$ for some $\bm c$, then say we prove it for $\bm c = 0$ as this makes things notationally simpler, but proof generalizes in straightforward way.}
}

\ignore{
Next, we explore the testing in the presence of errors. In this context, the early works~\cite{ABCG93, GLRSW91} introduced the notion of {\em approximate testing}, which
was made more formal by the work of Ergun, Kumar and Rubinfeld~\cite{EKR01}. Given two parameters $\alpha<\beta$, in addition to $\eps>0$, query access to $f: \R^n \to \R$ and sample
access to a distribution $\caD$, the goal of an approximate tester for a property $\mathcal{P}$ is to distinguish between the following two cases:
\begin{itemize}
\item
There exists $h \in \mathcal{P}$, such that $|f(\bm x)-h(\bm x)| < \alpha$ for all $\bm x\in\R^n$.
\item
For every $h \in \mathcal{P}$, $\Pr_{\bm x \sim \caD}[|f(\bm x)-h(\bm x)| >\beta]>\eps$. 
\end{itemize}
In the former case, we say that $f$ is \emph{pointwise $\alpha$-close to $h$}.
Here, $\alpha$ should be thought of as a representational limitation or round-off/truncation error. Thus, it is 
reasonable to assume that $\alpha = 1/\mathrm{poly}(n)$, as this can be achieved by storing $O(\log n)$ bits of precision.



In this setting, we analyze the property of additivity: a function $h\colon\R^n \to \R$ is said to be additive, if $h(\bm p + \bm q) = h(\bm p) + h(\bm q)$ for all $\bm p, \bm q\in\R^n$.\footnote{Recall that over $\mathbb{R}^n$ a function must be both \emph{additive} and \emph{homogeneous} (i.e., $f(c \bm x) = c f(\bm x)$ for every $\bm x \in \mathbb{R}^n$, $c \in \mathbb{R}$) to be linear.}  We give a \emph{partially} \nnote{no longer "partially"} distribution-free tester for additivity, which works for any distribution $\mathcal{D}$. Let $\lrad\in\R$ be the radius of the smallest Euclidean ball that contains most of the  mass of $\caD$ --- formally, we require $\Pr_{\bm p \sim \mathcal{D}} [\bm p \in \ball(\bm 0,\lrad)] \geq 1-\varepsilon/4$. Note that the standard Gaussian distribution satisfies this requirement for $\lrad=O(\sqrt{n})$.

\begin{restatable}{thm}{noisy}\label{thm:central-gausssian}
  There exists a one-sided error partially distribution-free  $O (\frac{1}{\varepsilon} \log\frac{1}{\varepsilon} )$-query approximate tester for testing additivity. The tester distinguishes between the case when $f$ is pointwise $\alpha$-close to some additive function and the case when, for every additive function $h$, $\Pr_{\bm p \sim \caD}[|f(\bm p)-h(\bm p)| > O(\lrad n^{1.5}\alpha) ] > \eps$.
  
\end{restatable}

Finally, if we allow for our approximation factor $\beta$ to depend on $\|p\|_2$, then we show a fully distribution-free tester for additivity. We call it \emph{multiplicative-approximate tester}.

\begin{restatable}{thm}{multiApprox}
 There exists a one-sided error  distribution-free, $O(\frac{1}{\varepsilon}\log\frac{1}{\varepsilon})$-queries, multiplicatively-approximate tester for testing additivity. It distinguishes between the case when $f$ is pointwise $\alpha$-close to some additive function, and the case when, for every additive function $h$, $\Pr_{\bm p \sim \caD}[|f(\bm p)-h(\bm p)| > O((1+\|\bm p\|_2)n^{1.5}\alpha)] > \eps$.
\end{restatable}

An important open question is to show approximate testability for general low-degree polynomials.
}

\subsection{Related Work}

Although distribution-free testing (for graph properties) was already defined in an early work on property testing~\cite{GoldreichGR98}, the first distribution-free testers for non-trivial properties appeared much later in the work of Halevy and Kushilevitz~\cite{HalevyK07}.
Since then, distribution-free testers have been considered for a variety of Boolean functions including low-degree polynomials, dictators, and monotone functions~\cite{HalevyK07}, $k$-juntas~\cite{HalevyK07,LiuCSSX19,Bshouty19,Belovs19}, conjunctions, decision lists, and linear threshold functions~\cite{GlasnerS07}, monotone and non-monotone monomials~\cite{DolevR11}, and monotone conjunctions~\cite{GlasnerS07, ChenX16}. 
The first (partial) distribution-free testing result for functions on the Euclidean space was due to Harms~\cite{Harms19}: He gave an efficient tester for half spaces over any rotationally invariant distribution.
Then, as we mentioned above, Fleming and Yoshida~\cite{FlemingY20} gave a tester for linearity of functions over the Euclidean space.

Property testing originated (implicitly, under the name of self-testing) in the work of Blum, Luby, and Rubinfeld~\cite{BlumLR93}, who exhibited the famous BLR tester for linearity over $\mathbb{F}_2$. Since then, testers have been developed for higher degree polynomials, such as the famous Rubinfeld Sudan~\cite{RubinfeldS96} and Raz and Safra~\cite{RazS97} tests for degree-$d$ polynomials over sufficiently large finite fields. One line of work, closely related to ours extended the domain over which these testers worked, culminating in the work of Lipton~\cite{Lipton89} and Rubinfeld and Sudan~\cite{RubinfeldS92}, who gave testers for degree-$d$ polynomials over any finite subset of rationals, where the distance is measured according to the uniform distribution; see~\cite{KiwiMS01} for an excellent survey.
The main distinguishing features between this paper and the works of~\cite{Lipton89, RubinfeldS92} is that (i) we work in the distribution-free setting, (ii) we do not assume that the domain is finite, and (iii) the input function is multivariate.

\subsection{Proof Overview}\label{sec:proof_technique}


This work significantly extends the framework of Fleming and Yoshida~\cite{FlemingY20}, who exhibited a constant-query algorithm for testing the linearity of functions over $\mathbb{R}^n$ in the distribution-free setting (when distance is measured according to an arbitrary distribution $\mathcal{D}$); thus, we briefly describe their proof first. 

\paragraph{Testing Linearity over the Reals.} The tester follows the high-level ``self-correct and test'' approach of Halevy and Kushilevitz~\cite{HalevyK07}. To test whether a given function $f\colon \mathbb{R}^n \to \mathbb{R}$ is linear, it suffices to construct a \emph{linear} function $g_{\mathsf{lin}}:\mathbb{R}^n \to \mathbb{R}$ such that:
\begin{enumerate}
    \item If $f$ is indeed a linear function, then $f=g_{\mathsf{lin}}$.
    \item For any $\bm p \in \mathbb{R}^n$, we can efficiently query the value of $g_{\mathsf{lin}}(\bm p)$ using queries to $f$. 
\end{enumerate}
Indeed, by (1), to test if $f$ is linear, it suffices to estimate the distance between $f$ and $g_{\mathsf{lin}}$ (measured according to $\mathcal{D}$), which can be done efficiently by (2). 

To construct $g_{\mathsf{lin}}$ they use the standard self-correcting approach pioneered in the Blum, Luby, and Rubinfeld (BLR) test for linearity over $\mathsf{GF(2)}$~\cite{BlumLR93}. However, this has to be significantly modified. Standard self-correction arguments require that every point in the distribution has equal probability mass, and there is no natural analogue to the uniform distribution over $\mathbb{R}^n$. Instead, they modify the self-correcting argument to work for the standard Gaussian distribution --- that is, by evaluating $f$ on points sampled from $\mathcal{N}(\bm 0,I)$, they are able to construct the desired function $g_{\mathsf{lin}}$. Note that even though $g_{\mathsf{lin}}$ is constructed using samples from $\mathcal{N}(\bm 0, I)$, in order to test whether $f$ is close to a linear function over the given distribution $\mathcal{D}$, by (1) it suffices to estimate the distance between $f$ and $g_{\mathsf{lin}}$ over $\mathcal{D}$. This can be done by sampling sufficiently many points $\bm p \sim \mathcal{D}$ and checking whether $f(\bm p) = g_{\mathsf{lin}}(\bm p)$, using (2) in order to evaluate $g_{\mathsf{lin}}(\bm p)$.

To circumvent the issue that points have differing probability mass under $\mathcal{N}(\bm 0, I)$, they project every point into a Euclidean ball $\ball(\bm 0, \srad)$ of small radius $\srad$ at the center of the Gaussian (see \autoref{fig:small_ball_gaussian}). Within this ball, every point has approximately the same mass and they are able to perform the self-correction argument. In particular, they define
\[ g_{\mathsf{lin}}(\bm p) \triangleq \gamma_{\bm p} \cdot \maj_{\bm q \sim \mathcal{N}(\bm 0,I)} \left[f \left( \frac{ \bm p}{\gamma_{\bm p}} -\bm q \right)  + f(\bm q)  \right], \] 
where $\gamma_{\bm p} \in \mathbb{R}$ is such that $\bm p/ \gamma_{\bm p} \in \ball(\bm 0,\srad)$. 
That is, $g_\mathsf{lin}(\bm p)$ is the majority value weighted according to the standard Gaussian distribution.
This is essentially the same self-corrected function used in the BLR test, except that each point $\bm q$ is first projected into $\ball(\bm 0, \srad)$. 

Finally they argue that, if their tests pass with a sufficiently high probability, then $g_{\mathsf{lin}}$ is a linear function, and furthermore, for any $\bm p \in \mathbb{R}^n$, the value of $g_{\mathsf{lin}}(\bm p)$ can be recovered with a small number of queries to $f$.

\begin{figure}
    \centering
    \begin{tikzpicture}[scale=0.8]
        \filldraw[color=white!0, even odd rule,inner color=red!50,outer color=white] (0,0) circle (3);
        \filldraw[color=black!60, fill=green!10, very thick](0,0) circle (0.8);
        
        \filldraw[color=black!60, fill=betterYellow!30, thick](1,1.5) circle (0.07);
        \filldraw[color=black!60, fill=betterYellow!30, thick](0.16,0.3) circle (0.07);
        \node[] at (0,0) {\small $\ball(\bm 0,\srad)$};
        \node[] at (0.1,-2.8) {$\mathcal{N}(\bm 0,I)$};
        \draw[very thick, color=cyan, ->](0.95,1.45) -- (0.2,0.34);
    \end{tikzpicture}
    \caption{Each point in the Gaussian distribution is projected into the small ball $\ball(\bm 0,\srad)$ at the origin.}\label{fig:small_ball_gaussian}
\end{figure}
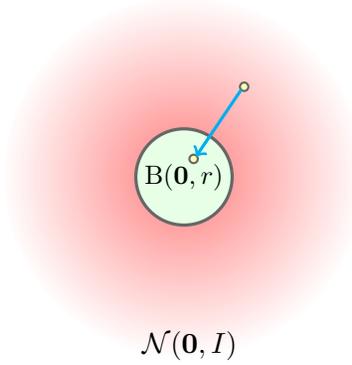

\paragraph{Exactly Testing Polynomials over the Reals.}

Our work is a significant generalization of the ideas used in the linearity test so that they may be applied to degree-$d$ polynomials.
Given a function $f\colon \mathbb{R}^n \to \mathbb{R}$, we construct a degree-$d$ polynomial $g\colon\mathbb{R}^n \to \mathbb{R}$ such that 
\begin{enumerate} 
    \item If $f$ is a degree-$d$ polynomial, then $f=g$.
    \item For any $\bm p \in \mathbb{R}^n$, we can efficiently query the value $g(\bm p)$ using queries to $f$.
\end{enumerate}
As in the the case of linear functions, we construct $g$ using samples from the Gaussian distribution. We mitigate the fact that points are weighted non-uniformly, by restricting attention to a small (open) ball $\ball(\bm 0, \srad)$, defining $g$ within that ball, and then extending outwards. Formally, let $\alpha_i \triangleq {(-1)}^{i+1} \binom{d+1}{i}$, and for any $\bm p \in \ball(\bm 0, \srad)$ and $\bm q \in \mathbb{R}^n$, let $g_{\bm q}(\bm p) \triangleq \sum_{i=1}^{d+1} \alpha_i \cdot f(\bm p+i \bm q)$.
For points $\bm p \in \ball(\bm 0,\srad)$, we define $g$ to be 
\[ g(\bm p) \triangleq \maj_{\bm q \sim \mathcal{N}(\bm 0, I)} \left[ g_{\bm q}(\bm p) \right],\] 
where the majority is weighted according to $\cN(\bm 0,I)$. For points $\bm p \not \in \ball(\bm 0,\srad)$, the value of $g(\bm p)$ is defined by interpolating from evaluations of $g$, on $d+1$ distinct points within $\ball(\bm 0,\srad)$ (this is defined formally in \autoref{sec:exact}).


Having thus defined $g$, we would like to argue that if a certain set of tests pass with sufficiently high probability, then $g$ is a degree-$d$ polynomial.
\ignore{However, previous characterizations of degree-$d$ polynomials, such as the straightforward generalization of the characterization used by Rubinfeld and Sudan~\cite{RubinfeldS96} to polynomials over $\mathbb{R}^n$, are highly non-local. To circumvent this, we develop a new \emph{local characterization} of degree-$d$ polynomials over the reals, which is more amenable to being tested (the \localCharacterization).
With this characterization in hand, we argue that $g$ is a degree-$d$ polynomial}
We make this argument in three steps, where each step extends the domain over which we guarantee that $g$ is a polynomial. 
\begin{enumerate} 
    \item We show that $g$ is consistent with a degree-$d$ univariate polynomial on any line segment $\pline_{\bm a, \bm b}^{\ball} \triangleq \{\bm a+ x \bm b \in \ball(\bm 0,\srad): x \in \mathbb{R}\}$ within the ball $\ball(\bm 0,\srad)$. To prove this, we generalize the self-correction argument from~\cite{RubinfeldS96} 
    to hold over the ball $\ball(\bm 0,\srad)$ of reals.
    \item We show how to stitch together these ``local'' representations of $g$ on lines into a degree-$d$ \emph{multivariate} polynomial, which is consistent with $g$ within a hypercube contained within $\ball(\bm 0, \srad)$. We describe this step in more detail below.
    \item We extend this representation of $g$ within the hypercube to a consistent representation of $g$ as a degree-$d$ polynomial everywhere. This follows by extrapolating $g$ from the small ball $\ball(\bm{0},r)$ to all of $\R^n$. 
\end{enumerate}
The main innovation is step (2), and therefore we will describe it in more detail. Step (2) is proved in two parts: first, we argue that $g$ can be represented as a polynomial of degree $dn$; second, we reduce the degree to $d$.

\begin{figure}[h]
    \centering
    \begin{tikzpicture}[scale=0.7]
        \filldraw[color=black!60, fill=green!9, very thick](3,0.25) circle (3.9);
        \filldraw[color=black!60, fill=red!5, very thick](0,-1) -- (2,-3) -- (6,-2) -- (6,1.5) -- (4,3.5) -- (0,2.5);
        \draw[color=blue!50, very thick] (4.38,.445) -- (8,1.33);
        \draw[color=black!60, very thick] (4,0)-- (4,3.5) -- (6,1.5) -- (6,-2) -- cycle;
        \draw[color=purple!20, very thick](3.2,-0.2)-- (3.2,3.3) -- (5.2,1.3) -- (5.2,-2.2) -- cycle;
        \filldraw[color=purple!60, fill=betterYellow!25, very thick](2.4,-0.4)-- (2.4,3.1) -- (4.4,1.1) -- (4.4,-2.4) -- cycle;
        \draw[color=purple!20, very thick] (0.8,-0.8)-- (0.8,2.7) -- (2.8,0.7) -- (2.8,-2.8) -- cycle;
        \draw[color=purple!20, very thick] (1.6,-0.6)-- (1.6,2.9) -- (3.6,0.9) -- (3.6,-2.6) -- cycle;
        \draw[color=black!60, very thick] (0,-1)-- (0,2.5) -- (2,0.5) -- (2,-3) -- cycle;
        \node[text width=1.2cm] at (3,-4) {$\ball(\bm 0,\srad)$};
        \filldraw[fill=green!9,  color =green!9] (5.4,2.6) circle (6pt);
		\draw [cyan, very thick, xshift=4cm] plot [smooth, tension=1] coordinates { (-1.6,0.7) (-1.1,1.5) (-0.7,-0.5) (0,0.1) (0.4,-2.4)};
		\draw[color=purple!60, very thick](2.4,-0.4)-- (2.4,3.1) -- (4.4,1.1) -- (4.4,-2.4) -- cycle;
		\draw[color=blue!50, very thick] (-2,-1.13) -- (3.15,.14);
		\draw[color=black!60, very thick] (2,-3) -- (6,-2);
        \draw[color=black!60,  very thick] (0,-1) -- (2.4,-.4);
        \draw[color=black!60,  very thick] (0,2.5) -- (4,3.5);
        \draw[color=black!60, very thick] (2,0.5) -- (6,1.5);
		
		\draw[color=black!60, very thick,->] (2.5,-3.15) -- (5.8,-2.35);
	    \filldraw[fill=green!9,  color =green!9] (4.08,-2.8) circle (6.8pt);
	    
		\node[text width=1cm] at (4.5,-2.8) {$x_n$};
		\node[text width=1cm] at (-1.9,-1.3) {$L$};
		\node[text width=1.5cm] at (3.4,-0.7) {\small $n(d-1)$};
    \end{tikzpicture}
    \caption{The construction of the degree $nd$ representation of $g$. $d+1$ slices, parallel to the $x_n$ axis, of the cube are chosen. On each slice $g$ is a degree $n(d-1)$ polynomial (cyan). These degree $n(d-1)$ representations on slices are stitched together along a line $L$.}\label{fig:local_to_global_simple}
\end{figure}
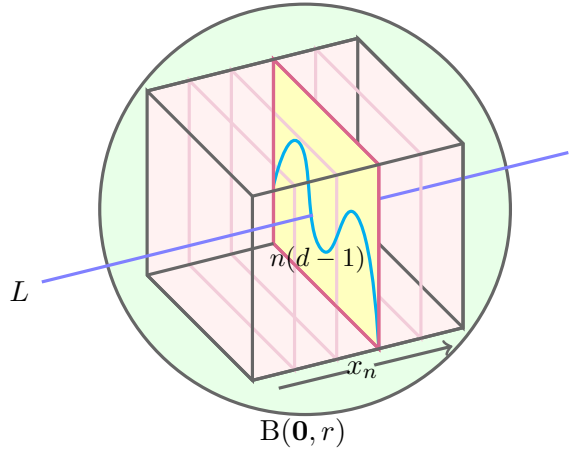

To prove the first part, we consider the largest $n$-dimensional cube that can be inscribed in the ball $\ball(\bm 0,\srad)$. 
We then discretize the cube by picking $d+1$ slices perpendicular to the $x_n$ axis ($(n-1)$-dimensional sub-cubes), and argue by induction that $g$ can be written as a degree $(n(d-1))$ polynomial on each slice\footnote{The reason that we use a hypercube embedded within the ball --- rather than using $d+1$ slices of $\ball(\bm 0,\srad)$ --- is that we require each  of the $(n-1)$-dimensional polynomials have the same domain. If we took $d+1$ slices of $\ball(\bm 0, \srad)$, this would not be true.} (see \autoref{fig:local_to_global_simple}). To combine these $n(d-1)$-degree polynomials into a degree-$dn$ polynomial, we consider any $L$ line parallel to the $x_n$ axis. This line has exactly one intersection point with each of the $d+1$ slices. By step (1), $g$ restricted to this line is a degree-$d$ \emph{univariate} polynomial. Using this univariate representation of $g$ on the line to interpolate between the $n(d-1)$-degree representations of $g$ on the slices allows us to obtain a representation of $g$ as a degree-$dn$ polynomial within the hypercube.


To reduce the degree of this representation of $g$ from $dn$ to $d$, we use the fact that (by step (1)) $g$ can be represented as a degree-$d$ univariate polynomial on every line segment within the ball. In particular, we show that for any representation if $g$ as a polynomial of some degree $m$, there exists a radial line $L_{\bm 0, \bm b}$ such that $g$ restricted to this line also has degree $m$.  However, by step (1), $g$ restricted to any line has degree at most $d$, and this implies that $m \leq d$.

\ignore{
Finally, we show that our construction allows us to efficiently recover the value of $g(\bm p)$ for any $\bm p \in \mathbb{R}^n$. To see this, observe that $\bm p$ lies along the line $L_{\bm 0, \bm p}$, so it suffices to ``learn'' the representation of the univariate degree-$d$ polynomial $g(x \bm p)$ in the formal variable $x$, and then evaluate it at $x=1$.
To recover the representation of $g(x \bm p)$, we evaluate it on $d+1$ points $v_1,\ldots, v_{d+1} \in \mathbb{R}$ such that $v_i \bm p \in \ball(\bm 0,\srad)$, which by definition can be done with queries to $f$. The $d+1$ evaluations $g(v_i \bm p)$ uniquely determine $g(x \bm p)$, which allows us to recover it by using polynomial interpolation.}

\paragraph{Approximate Testing Polynomials over the Reals.} 
The major new challenge that arises in approximate testing is that we must ensure that our tester accepts all functions that are \emph{pointwise $\delta$-close} to  being a polynomial; i.e., we should accept $f$ if there exists a degree-$d$ polynomial $\hat f$ such that for every $\bm x \in \mathbb{R}^n$
\[ |f(\bm x) - \hat f(\bm x)| \leq \delta. \]
We work in the setting where the unknown distribution $\mathcal{D}$ is known to satisfy the condition that $1-\eps/4$ fraction of the mass of $\mathcal{D}$ is contained within $\ball(\bm 0, \lrad)$ for a given parameter $\lrad$ (see \autoref{fig:proj2dPlane}).

\begin{figure}
\center
\begin{tikzpicture}[scale=0.8]
    \filldraw[color=white!0, even odd rule,inner color=red!50,outer color=white] (0,0) circle (3);
    \draw[color=black!60, fill=green!9, very thick] (0,0) circle (2cm);
  \draw[color=black!60, fill=betterYellow!25, very thick] (0,0) circle (1cm);

  \draw[black!60, thick ,dashed] (1.4,-1.4)  -- (0,0);
  \draw[black!60, thick ,dashed] (.7,0.7)  -- (0,0);

  \node[text width=.5cm] at (0 , 0) {$\bm 0$};
  \node[text width=.5cm] at (1.3 , -0.8) {$\lrad$};
  \node[text width=.5cm] at (.4 , 0.5) {$\srad$};
  \node[text width=.5cm] at (0 , -2.5) {$\mathcal{D}$};
    \end{tikzpicture}
  \caption{The ball $\ball(\bm 0, \lrad)$, containing at least $1-\varepsilon/4$ of the mass of the $(\varepsilon/4,R)$-concentrated distribution $\mathcal{D}$, and the ball $\ball(\bm 0, \srad)$ from which $g$ is extrapolated. }
 
  \label{fig:proj2dPlane}
\end{figure}
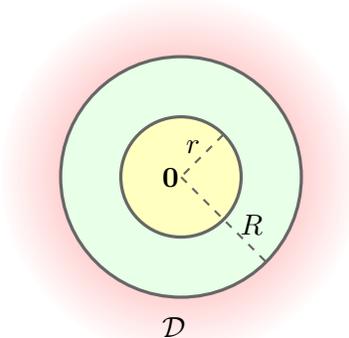

We begin by constructing a self-corrected function $g$, except that now, it is in terms of the median\footnote{The use of median in the context of approximate testing is not a new idea; see, e.g.,~\cite{KiwiMS01}.} instead of the majority. Our analysis then follows the three-step outline mentioned above for the exact case.  In the first step, we argue that $g$ approximately satisfies the univariate characterization of degree-$d$ polynomials on every line restricted to $\ball(\bm 0,\srad)$, and hence, $g$ is pointwise close to a low degree univariate polynomial on every such line segment. The last conclusion is due to a theorem of Gajda~\cite{Gaj91} from the literature on Hyers-Ulam stability results for functional equations; see the book~\cite{HIR12} for a comprehensive survey of this area. 

Our main technical contribution comes in the second step of the analysis. We show that being pointwise close to a multivariate low-degree polynomial is approximately a `lifted' property~\cite{GKS13}.

\begin{lem}\label{lem:maintech}
    Let $m\in(0,1],\delta>0$, and let $h\colon [-m,m]^n \to \mathbb{R}$. 
    If for every line $\pline$, there exists  a degree-$d$ univariate polynomial $\hat{h}_{\pline}$ that is pointwise $\delta$-close to  $h_{\pline}$ (the restriction of $h$ on the line $\pline$), 
    then $h$ is pointwise $((2/m)^{(n^{40d} )} \delta )$-close to a  degree-$d$ polynomial. 
\end{lem}

The proof of \autoref{lem:maintech} is by induction on $n$, where we show in each step, that (i) the function  is pointwise close to a degree-$2d$ polynomial, and then that (ii) the function {from step (i)} is pointwise close to a degree-$d$ polynomial. Both parts refine the corresponding analysis in the exact case. 
\begin{itemize}
\item
For part (i), we choose $d+1$ hyperplanes $\cube_0\triangleq\{x_n = c_0\}, \dots, \cube_d \triangleq \{x_n=c_d\}$ where $c_0, \dots, c_d$ are the scaled \emph{Chebyshev nodes}. By induction, there exist degree-$d$ polynomials $\widehat{g}_i$ that are pointwise close to $g$ on $\cube_i$.  Now, for any line $\pline$ parallel to the $x_n$ axis, we look at the univariate degree $d$ polynomial $g_{\pline}$ that $g$ is pointwise close to, and the degree $d$ polynomial $\hat{g}_{\pline}$ that agrees with $\hat{g}_i$ for each of the intersections between $\pline$ and $\cube_i$. The difference $g_{\pline} - \hat{g}_{\pline}$ is small at the Chebyshev nodes, which implies that $g_{\pline} - \hat{g}_{\pline}$ is small everywhere inside  $[-m,m]^n$. This argument yields a degree-$2d$ polynomial that is pointwise close to $g$ on $[-m,m]^n$.

\item
We prove a more general result that implies what we need in part (ii). 
\begin{restatable}{thm}{multivariateDegreeReduction}\label{thm:multivarite_degree_reduction}
Let $m\in (0,1]$, $n\geq 2$ and  $p$ be an $n$-variate polynomial of total degree at most $\ell$, for some $d\leq \ell$. If for every $\bm a\in[-m,m]^n$, the univariate polynomial $p_{\bm 0,\bm a}(t) =p(\bm a t)$ which is the restriction of $p$ to the radial line $\pline_{\bm 0,\bm a}$, is pointwise $\varepsilon$-close to a degree-$d$ univariate polynomial on the interval $t\in[-1,1]$, then $p$ is pointwise $\eta$-close to $p^{\leq d}$ (the truncation of $p$ to degree $d$) on $[-m,m]^n$ for $\eta=2 (2/m)^{2n^{18\ell}}\varepsilon$. 
\end{restatable}

In order to prove \autoref{thm:multivarite_degree_reduction}, suppose for the sake of contradiction that $p - p^{\leq d}$ is large at some point in $[-m,m]^n$. By a straightforward argument, this implies that there must be coefficient $\alpha_I$ of a degree $\geq d$ monomial in $p$ which has large magnitude. From this, we would like to conclude that the restriction of $p$ to some radial line $\pline_{{\bm 0}, {\bm a}}$ must not be pointwise close to a degree-$d$ polynomial, and hence we would have a contradiction. Let the restriction of $p$ to this line be defined as
\[
p({\bm a}t) = \sum_{k \leq \ell} \gamma_k({\bm a}) T_k(t),
\] 
where $T_k$ is the $k$th \emph{Chebyshev polynomial}. It turns out that in order to show $p(\bm at)$ is not close to a degree-$d$ polynomial, it suffices to show $\gamma_k({\bm a})$ is large for some $k>d$.

The large coefficient $\alpha_I$ of $p$ appears in some coefficient $\gamma_{k^*}({\bm a})$ for $k^*>d$\footnote{In fact, $k^*$ is either $d+1$ or $d+2$.}. Note that $\gamma_{k^*}$ is itself a degree-$\ell$ multivariate polynomial when we consider ${\bm a}$ as variables. In order to conclude that $\gamma_k(\bm a)$ is large for some $\bm a$, we will  choose values for $\bm a$ such that $\gamma_{k^*}$ is a degree-$d$ univariate polynomial (in some variable $z$) and there is a monomial in $\gamma_{k^*}$ with a large coefficient; anti-concentration then implies that there is a setting of $z$ which makes $\gamma_{k^*}$ large. To satisfy this, we want to choose a substitution for $\bm a$ in $z$ such that the monomials under this substitution have exactly the same coefficients as those of $\gamma_{k^*}(\bm a)$ (that is, no two monomials collapse to the same monomial).

Fixing a formal variable $z$, we set ${\bm a}$ to be $(z^{y_1}, \dots, z^{y_n})$ for an integer valued vector $\bm{y} = (y_1, \dots, y_n)$, and define $\widetilde{\gamma}_{k^*}(z) = \gamma_{k^*}(z^{y_1}, \dots, z^{y_n})$. We choose $\bm{y}$ in such a way that distinct monomials of ${\bm a}$ in $\gamma_{k^*}$ lead to distinct powers of $z$ in $\widetilde{\gamma}_{k^*}$; such a $\bm{y}$ exists due to a probabilistic argument. 

At this point, we have a univariate polynomial $\widetilde{\gamma}_{k^*}$ that has at least one large coefficient, and we would like to conclude that it has a large value at some point. This is a statement about the \emph{anti-concentration} of the polynomial $\widetilde{\gamma}_{k^*}$. If the largest coefficient were the leading term, then it is well-known that Chebyshev polynomials attain the smallest uniform norm on $[-1,1]$ among all such polynomials. In our situation, the largest coefficient may not be the leading one; nevertheless, we can show a lower bound on the uniform norm by making a connection to Chebyshev polynomials\footnote{Note that (scaled) Chebyshev polynomials are bounded by $2^{1-d}\eta$ within $[-1,1]$. We leave it open whether the lower bound of $2^{-O(d^2)}$ can be improved to $2^{-O(d)}$. However, for our application, this improvement would not be significant.}: 
\begin{restatable}{lem}{nonmoniclowerbound}\label{lem:non-monic-lower-bound}
    Let $p(x) = \sum_{i=0}^d \alpha_i x^i$ be a degree-$d$ polynomial and let $\eta >0$. If $|\alpha_i| \geq \eta$ for some $i\geq 1$, then there exists $x \in [-1,1]$ such that $|p(x)| \geq 2^{-2d^2}\eta $.
\end{restatable}
\end{itemize}

We now return to the main thread of describing the three-step analysis for \autoref{thm:approximate_low_degree_tester}. In the last step, we need to extrapolate our definition of $g$ from within the small ball $\ball(\bm 0, \srad)$ to the bigger ball $\ball(\bm 0,\lrad)$, within which the underlying distribution is concentrated. Again, using properties of Chebyshev polynomials, we show that if $g$ is pointwise $\eta$-close to a degree-$d$ polynomial in some ball $\ball(\bm 0,\srad')$\footnote{In our analysis, we have, and choose $\srad'$ that are strictly smaller than $\srad$.}, then its extrapolation is pointwise
$(O(R/r'))^d\eta$ pointwise close to a degree-$d$ polynomial in $\ball(\bm 0,\lrad)$. After this, the rest of the analysis mirrors the one for the exact case. 

\ignore{
As a warm-up to the polynomial case (and because we obtain better error parameters), we first show a simple noise-resistant tester for \emph{additivity}. Recall that a function is additive if for every $\bm x, \bm y \in \mathbb{R}^n$, $h(\bm x + \bm y) = h(\bm x) + h(\bm y)$; $h$ is \emph{linear} if additionally $h(c \bm x) = c h(\bm x)$ for all $c \in \mathbb{R}$. In this setting we would like to determine whether a given function $f$ is $\delta$-additive over an unknown distribution $\mathcal{D}$, meaning that 
\[ | f(\bm x + \bm y) - f( \bm x) - f(\bm y) | \leq \delta, \]
or $\varepsilon$-far from any additive function over $\mathcal{D}$ with sufficiently high probability.

To develop our approximate additivity tester, we significantly modify the linearity tester of~\cite{FlemingY20}. Our algorithm first tests whether the $f$ satisfies some constraints that characterize the notion of $\delta$-additivity over $\cN(\bm 0,I)$. If this is the case, then we argue that we can construct  a  function $g\colon \mathbb{R}^n \to \mathbb{R}$ which is a \emph{self-corrected} version of $f$, satisfying 
\begin{itemize}
    \item $O(\delta)$-additive in a small ball,
    \item pointwise close to an additive function $h$,
\end{itemize}
and furthermore, if $f$ is indeed $\delta$-additive, then $f$ and $g$ (and therefore $f$ and $h$) are close. 
Moreover, we show that we can estimate the evaluation of $g$ on any point within a small fixed ball around $\bm 0$ with high probability by querying $f$ on correlated points. 
Altogether, this allows us to test whether $f$ is $\delta$-additive by estimating the distance between $f$ and $g$ (and therefore $f$ and $h$) by approximating the value of $g$ on sampled points and comparing it with the value of $f$ on the same point. At a high-level, the construction of $g$ follows similar ideas to the construction of the self-corrected function in~\cite{FlemingY20}, outlined at the beginning of this section: first, $g$ is constructed within a small ball, and then extrapolated to $\mathbb{R}^n$. 
However, this is significantly complicated by the approximate setting. In order to mitigate errors, we define $g$ as a median, rather than a majority, and our analysis crucially relies on a Hyers-Ulam-type stability theorem from~\cite{Kominek89}.

By a significantly non-trivial generalization of the techniques that we used in order to establish our exact tester for polynomials, and the noise-resistant tester for additivity, we give a noise-resistant tester for degree-$d$ polynomials. \nnote{mention some of the trouble that we encountered in proving this}
}
\paragraph{Exactly Testing Polynomials over Discrete Domains.}
For \autoref{thm:low_degree_test_discrete}, the main complication is that we can no longer evaluate points (even approximately!) on points drawn from $\mathcal{N}(0, I)$. We crucially relied on properties of the Gaussian (e.g., it is stable) for showing the self-correction properties of $g$ in the above results. Instead here, we sample from \emph{discrete Gaussian} distributions on lattices in order to define the self-corrected function $g$.  Discrete Gaussians are a fundamental object of study in lattice cryptography (see, e.g.,~\cite{MR07, Reg09}). Ours seems to be the first application of discrete Gaussians in a property testing setting.

For a lattice $\caL$, the discrete Gaussian $\caG(\caL, s)$ is proportional to the density function of $\mathcal{N}(0,s)$ on the lattice points. The self-corrected function $g$ is defined as $\mathsf{maj}_{\bm{q}\sim \caG(\caL', 1)}[g_{\bm{q}}(\bm{p})]$, where $g_{\bm{q}}$ is the same as in the exact testing analysis over $\R^n$. We perform the same three-step analysis here as above. For the first step, in order to show that $g$ satisfies the degree-$d$ characterization over lattice points, we derive explicit bounds on the TV distance between discrete Gaussians that were implicit in previous literature. For the second step, we follow the argument in the exact case, but we need to ensure that the lattice is large enough so that a nonzero low-degree polynomial is nonzero on at least one lattice point. Finally, in the third step, we extrapolate $g$ from its self-corrected values on lattice points of $\caL'$ inside a small ball $\ball(\bm 0,\srad)$ to lattice points of $\caL$ on which $\caD$ is supported. By the concentration property of $\caD$ and from taking $\caL'$ fine enough, we can find $d+1$ lattice points of $\caL'$ on any line from the origin to a point in $\caL \cap \ball(\bm 0,R)$. This suffices for the extrapolation and the rest of the analysis.

\ignore{
 We show that we can do the self-correction as described above on lattice points inside the ball $\ball(\bm 0, \srad)$, and we extend our local-to-global recovery result for polynomials to this discrete setting. In order to carry through our analysis, we assume that our lattice $\frac{1}{B}\mathbb{Z}^n$ is sufficiently dense --- $B=\Omega ({(dn)}^{2d} )$. This corresponds to assuming that our algorithm is operating with $O(\log n)$ bits of precision. Furthermore, we require that the unknown distribution $\mathcal{D}$ places the majority of its mass on points that are of distance $\mathrm{poly}(n)$ from the origin \ynote{Such an assumption is not mentioned in the theorem statement}.
}


\subsection{Further Remarks}
We leave the question of improving the bounds for the query complexity and the other parameters in \autoref{thm:approximate_low_degree_tester} and \autoref{thm:low_degree_test_discrete} as interesting open problems. Also, it would be very interesting to obtain a separation between the complexities of the exact and approximate testing problems, in terms of query complexity. For the case of $d=1$, we have an improved analysis that appears in \autoref{sec:additivity}.

It is also natural to ask about {\em tolerant testing} \cite{PRR06} in our setting. This is distinct from approximate testing, because in the completeness case, the function is only required to equal a degree-$d$ polynomial $P$ with some probability over the distribution $\caD$ which may be less than 1. Our test should still work under an appropriate choice of parameters, because by the union bound, we can upper bound the probability that one of the queries does not come from $P$. 

Our work also opens the way for investigating the testability of other multivariate functional equations. Is there a general theory that characterizes testability (under natural assumptions) just as there is for finite fields \cite{KS08, BFHHL13}?

\subsection{Organization}
In the following section we discuss some preliminaries used for the exact testing. In  \autoref{sec:exact} we give the full proof for the existence of an exact tester for low-degree polynomials, proving \autoref{thm:low_degree_test}. \autoref{sec:approximate} is devoted to proving \autoref{thm:approximate_low_degree_tester}, giving the approximate tester, wherein in \autoref{sec:chebyshev}, we give more preliminaries needed for the approximate tester. And \autoref{sec:discrete} contains the tester for discrete domains, as specified in \autoref{thm:low_degree_test_discrete}, wherein in \autoref{sec:prelim_lattices} we give some more preliminaries needed for the discrete case. As the exact tester is the starting point for the other settings, in the later sections we rely on the proofs from \autoref{sec:exact}, and show what changes need to be done.
Finally, in \autoref{sec:additivity}, we prove a sub-case of the approximate tester, where $d=1$ and show a better result. While in the other appendices we show full proofs of some intermediate lemmata/theorems, that we skipped in the paper for the convenience of the reader. 

%% file: prelims.tex

\section{Preliminaries}\label{sec:prelims}
Here we record some notations and definitions which will be used throughout the paper. 
For a positive integer $n$, let $[n] = \{1,2,\ldots,n\}$.
We will reserve non-boldface symbols (such as $a \in \mathbb{R}$) to represent variables and scalars, and we will use boldface (such as ${\bm a} \in \mathbb{R}^n$) to represent vectors. 

For any $S \subseteq \mathbb{R}^n$, we say that $f\colon \mathbb{R}^n \to \mathbb{R}$ is a \emph{polynomial over $S$} if there exists a degree-$d$ polynomial $g:\mathbb{R}^n \to \mathbb{R}$ such that $f(\bm x) = g(\bm x)$ for every $\bm x \in S$.
A \emph{line} is a polynomial of the form ${\bm a} + i{\bm b}$, where $i$ is a variable, and we will denote by $\pline_{\bm a,\bm b} \triangleq \{{\bm a} + i{\bm b} : i \in \mathbb{R} \}$, the set of points on this line. A \emph{radial line} is a line that passes through the origin; that is, a line of the form $\pline_{\bm 0,\bm b}$ for some ${\bm b} \in \mathbb{R}^n$. Throughout this paper, it will be convenient to talk about functions restricted to lines.
For $\bm{a},\bm{b} \in \mathbb{R}^n$, let $f_{\bm{a},\bm{b}}: \mathbb{R} \to \mathbb{R}$ be defined as the restriction of $f$ on $\pline_{\bm{a},\bm{b}}$, i.e., $f_{\bm{a},\bm{b}}(x) = f(\bm{a} + x\bm{b})$. 

\paragraph{Local Characterization of Degree-$d$ Polynomials.}
In order to test whether a univariate function $f$ is consistent with a degree-$d$ polynomial, we will use a characterization of degree-$d$ polynomials which is more amenable to this task. This characterization involves inspecting the \emph{finite forward differences} of $f$, defined as 
\begin{equation}\label{eq:definition-of-difference}
    \Delta_{h}[f](x) \triangleq f(x+h)-f(x),
\end{equation}
for $h \in \mathbb{R}$.  This difference is a linear operator, i.e., for functions $f$, and $g$,
\begin{equation}\label{eq:difference-is-linear}
    \Delta_h[f + g](x) = \Delta_h[f](x) + \Delta_h[g](x).
\end{equation}
Higher order finite forward differences are defined inductively as, 
\begin{equation}\label{eq:higher-order-difference}
     \Delta_{h}^{(m)}[f](x) \triangleq \Delta_h \Big[\Delta_h^{(m-1)}[f] \Big](x)  = {(-1)}^{m+1}\sum_{i=0}^{m}\alpha_i\cdot f(x+ih),
\end{equation} 
where $\alpha_i \triangleq {(-1)}^{i+1}{\binom{m}{i}},m\in\mathbb{Z}_{>1}$, and $\Delta_h^{(1)}=\Delta_h$.
Finite forward differences are related to the standard notion of a derivative, and we explain this further in \autoref{sec:appendix}. 

We will use the following characterization of degree-$d$ polynomials, that follow from well-known results in analysis (see \autoref{sec:appendix} for details).

\begin{restatable}{local_characterization}{localchar}\label{thm:main_poly_characterization}
Let $a,b \in \mathbb{R}$ such that $a<b$, and let $g\colon(a,b)\to\mathbb{R}$ be a univariate, bounded function. If for every $x\in(a,b)$ and sufficiently small $h>0$, such that $a<x<x+(d+1)h<b$, $\Delta_h^{(d+1)}[g](x)= 0$,  then $g$ is a degree-$d$ polynomial.
\end{restatable}

A discrete variant of this theorem, given in \autoref{sec:discrete}, will be used for our lattice-based tester.


\paragraph{Sampling from Gaussian Distributions.} In order to test that the local characterization holds, we will sample points from the $\bm p\sim\cN(\bm 0,\beta I)$ for various values of $\beta$. This is possible, given sampling access to $\cN(\bm 0,I)$, by multiplying sampled vectors with the respective $\sqrt{\beta}$'s, since $\sqrt{\beta}\bm v\sim\cN(\bm 0,\beta I)$, if $\bm v\sim\cN(\bm 0,I)$.

In order to generalize our tester to distributions that need not be centered at the origin, but say, at $\bm c\in\R^n$, we can test the local characterization at points sampled from Gaussians that are centered at such $\bm c$'s. This again is possible, given sampling access to $\cN(\bm 0,I)$, by translating the sampled vectors by the respective $\bm c$'s, since $\bm c+\bm v\sim\cN(\bm c,I)$, if $\bm v\sim\cN(\bm 0,I)$.

Throughout this work, we will need to relate points sampled from different Gaussian distributions. For two distributions $\mathcal{D}$ and $\mathcal{D'}$ on the same domain $\Omega$, the total variation distance between them is defined as $$\dtv(\mathcal{D},\mathcal{D'})\triangleq\frac{1}{2}\int_{\Omega}|\mathcal{D}(x)-\mathcal{D'}(x)|dx.$$
We will use the following lemma (a proof can be found in~\cite{FlemingY20}) to bound the total variation distance between two Gaussian distributions. Let $\| \cdot \|_2$ denote the operator norm on matrices.

\begin{lem}\label{lem:bound-TV-different-means}
  Consider two Gaussian distributions $\mathcal{N}(\bm{\mu}_1,\bm{\Sigma}), \mathcal{N}(\bm{\mu}_2,\bm{\Sigma})$ with shared invertible covariance matrices $\bm{\Sigma} \in \mathbb{R}^{n \times n}$. Then $\dtv(\mathcal{N}(\bm{\mu}_1,\bm{\Sigma}), \mathcal{N}(\bm{\mu}_2,\bm{\Sigma})) \leq \phi$ holds, if $\| \bm{\mu}_1 -\bm{\mu}_2 \|_2 \leq  2 \phi / \sqrt{ \| \bm{\Sigma}^{-1} \|_2}$.
\end{lem}

An immediate corollary of \autoref{lem:bound-TV-different-means} is the following:

\begin{lem}\label{lem:usefulDTVBound}
	For any integer $k \geq 1$, real $\srad>0$, and $\bm{p} \in \mathbb{R}^n$, such that $\|\bm{p}\|_2 \leq \srad$, it follows that $\dtv(\mathcal{N}(\bm{0}, kI), \mathcal{N}(\bm{p}, kI)) \leq k\srad/2$.
\end{lem}
\begin{proof}
    Observe that the spectral norm of $kI$ is $k$, and therefore $\|\bm p\|_2 \sqrt{\|kI\|_2} \leq k\srad$. 
    It follows from \autoref{lem:bound-TV-different-means}, that $\dtv(\mathcal{N}(\bm{0}, kI), \mathcal{N}(\bm{p}, kI)) \leq k\srad/2$.
\end{proof}

%% file: exact.tex

\section{Exact Testing}\label{sec:exact}
In this section, we develop a distribution-free tester for low-degree polynomials over the $\mathbb{R}^n$, assuming that we can exactly query the input function. Our tester is given in \autoref{alg:low_degree_main_algorithm} and uses the subroutines given in \autoref{alg:subroutines}. The \Call{CharacterizationTest}{} checks properties of $f$ which will be sufficient to guarantee that $g$ --- the \emph{self-corrected} version of $f$ --- is a degree-$d$ polynomial. 
\Call{Query-$g$}{} retrieves the value of $g(\bm p)$ for a given point $\bm p$ by running the subroutine \Call{Query-$g$-InBall}{}, which in turn obtains the values of $g$ on points within the small ball $\ball(\bm 0, \srad)$ by evaluating $f$.

Recall that in \autoref{thm:low_degree_test}, $f$ is assumed to be bounded in $\ball(\bm 0, L)$, for some $L>0$. Throughout this section, we assume $L=2d\sqrt{n}$. This is without loss of generality as we can define $f': \R^n \to \R$ as $f'(\bm x) = f(\bm xL/(2d\sqrt{n}))$ which is bounded in $\ball(\bm{0}, 2d\sqrt{n})$, and the tester can query $f'$ via queries to $f$. If $f$ is a degree-$d$ polynomial, so is $f'$. If $f$ is $\eps$-far from degree-$d$ polynomials over a distribution $\mathcal{D}$, so is $f'$ over the distribution $\mathcal{D}'$, where a sample $\bm y\sim \mathcal{D}'$ is generated as $\bm y = \frac{2d\sqrt{n}}{L}\bm x$ where $\bm x \sim \mathcal{D}$. 

\paragraph{The Self-Corrected Function.}
As outlined in \autoref{sec:proof_technique}, by sampling points from the standard Gaussian, we will construct a self-corrected version $g$ of the input function $f$ such that, if our tests (in particular \Call{CharacterizationTest}{} in \autoref{alg:subroutines}) pass with sufficiently high probability, then we can guarantee that $g$ is a degree-$d$ polynomial. 
Let $\srad =(3d)^{-6}$, and $\ball(\bm 0,\srad)$ be the \emph{open} ball of radius $\srad$, centered at the origin. We will guarantee that $g$ is a degree-$d$ polynomial for points $\bm p \in\ball(\bm 0,r)$ first, and then extended the characterization to points outside of this ball. The advantage of restricting our attention to this small ball is that for any $\bm p \in \ball(\bm 0,\srad)$, $\bm p +\bm x$ is approximately distributed as $\bm x$.
\begin{figure}
    \centering
    \begin{tikzpicture}
        \filldraw[color=white!0, even odd rule,inner color=red!50,outer color=white] (0,0) circle (3);
        \filldraw[color=black!60, fill=green!10, thick](0,0) circle (1);
        

        \draw[very thick, color=black!60](2,2) -- (-2,-2);
        
       \draw [cyan, very thick, xshift=0] plot [smooth, tension=1] coordinates { (-1.95, -2.1)(-1.8, -1.8)(-1.65, -1.5) (-1.5,-1.5) (-1.35, -1.5) (-1.2,-1.2) (-1.05,-0.9) (-0.9,-0.9) (-0.75, -0.9) (-0.6,-0.6) (-0.45, -0.3) (-0.3,-0.3) (-0.15,-0.3) (0,0) (0.15,0.3) (0.3,0.3) (0.45,0.3) (0.6,0.6) (0.75, 0.9) (0.9,0.9) (1.05, 0.9) (1.2,1.2) (1.35, 1.5) (1.5,1.5) (1.65, 1.5) (1.8,1.8) (1.95, 2.1)};
        
        \filldraw[color=black!60, fill=betterYellow!30, thick](-0.3,-0.3) circle (0.07);
        \filldraw[color=black!60, fill=betterYellow!30, thick](0.2,0.2) circle (0.07);
        \filldraw[color=black!60, fill=betterYellow!30, thick](0.65,0.65) circle (0.07);
        \node[] at (0,-1.4) {\small $\ball(\bm 0,\srad)$};
        \node[] at (-0.3,-0.6) {\small $c_1$};
        \node[] at (0.2,-0.1) {\small $c_2$};
        \node[] at (0.7,0.4) {\small $c_3$};
        

        \node[] at (-2.2,-2.2) {$\pline_{\bm 0, \bm p}$};
        \node[] at (2.1,2.3) {$p_{\bm p}$};
    \end{tikzpicture}
    \caption{The definition of $g(\bm p)$ for $\bm p \not \in \ball( \bm 0, \srad)$. First, a degree-$d$ univariate polynomial $p_{\bm p}$ is defined by the value of $g$ on $d+1$ points $c_i \in \ball(\bm 0, \srad)$ on the line $\pline_{\bm 0, \bm p}$, such that $p_{\bm p}(c_i) = g(\bm p c_i)$. Then, the value of $g(\bm p)$ is defined to be $p_{\bm p}(1)$.}\label{fig:extrapolation}
\end{figure}
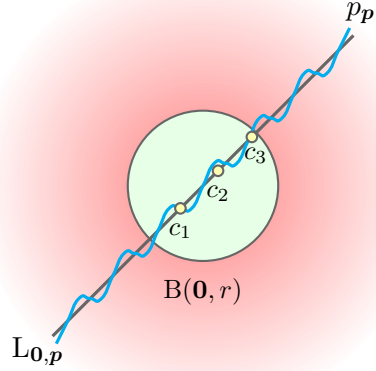

We define $g\colon \mathbb{R}^n \to \mathbb{R}$ formally as follows: let $\alpha_i = {(-1)}^{i+1} {\binom{d+1}{i}}$, and for any $\bm p \in \ball(\bm 0, \srad)$ and $\bm q \in \mathbb{R}^n$,
$g_{\bm q}(\bm p) \triangleq \sum_{i=1}^{d+1} \alpha_i \cdot f(\bm p+i \bm q)$.
The intuition behind $g_{\bm q}(\bm p)$ is that it is the value of the univariate, degree-$d$ polynomial at the point $\bm p$, that is uniquely defined by the $d+1$ evaluations $\{f(\bm p + i \bm q) : i \in [d+1]\}$.
For points $\bm p \in \ball(\bm 0,\srad)$, we define the value of $g$ to be
\[ g(\bm p) \triangleq \maj_{\bm q \sim \mathcal{N}(\bm 0, I)} \left[ g_{\bm q}(\bm p) \right].\]

For points $\bm p \not \in \ball(\bm 0, \srad)$, we define $g(\bm p)$ by interpolating the evaluations of $g$ on points within $\ball(\bm 0,\srad)$ as follows (see \autoref{fig:extrapolation}). Consider the radial line $\pline_{\bm 0, \bm p} = \{x \bm p : x \in \mathbb{R}\}$ and fix $d+1$ (arbitrary) ``distinguished'' points along this line $c_0,\ldots, c_d \in \mathbb{R}$ such that $c_i \bm p \in \ball(\bm 0, \srad)$ for all $i$; in \autoref{alg:subroutines} we choose $c_i = ir/((d+1) \|\bm p|_2)$. Let $p_{\bm p}\colon \mathbb{R}^n \to \mathbb{R}$ be the degree-$d$, univariate polynomial uniquely defined by these $d+1$ points, such that $p_{\bm p}(c_i)= g(c_i \bm p)$, for every $i\in[d+1]$. The value of $g(\bm p)$ is defined as $p_{\bm p}(1)$.
Note that if $g$ was a degree-$d$ polynomial to begin with, then we would indeed have $p_{\bm p}(1) = g(\bm p)$.

\begin{algorithm}[ht]
\caption{Distribution-Free Low-Degree Tester}\label{alg:low_degree_main_algorithm}
\Procedure{\emph{\Call{LowDegreeTester}{$f,d,\mathcal{D},\varepsilon$}}}{
  \Given {Query access to $f\colon \mathbb{R}^n \rightarrow \mathbb{R}$, degree $d\in\mathbb{N}$, sampling access to an unknown distribution $\mathcal{D}$, and farness parameter $\eps > 0$.}
    \textbf{Reject} if \emph{\Call{CharacterizationTest}{}}  \textbf{rejects}\;
    $N_{\ref{alg:low_degree_main_algorithm}} \gets O(\varepsilon^{-1})$\;
    \For{$N_{\ref{alg:low_degree_main_algorithm}}$ times}{
      Sample $\bm{p} \sim \mathcal{D}$\;
      \textbf{Reject} if $f(\bm{p}) \neq$ \emph{\Call{Query-$g$}{$\bm{p}$}} or if \emph{\Call{Query-$g$}{$\bm{p}$}}  \textbf{rejects}.
    }
    \textbf{Accept}.
  }  
\end{algorithm}

\begin{algorithm}[!t]
\caption{Subroutines}\label{alg:subroutines}
[Recall $\alpha_{i}\triangleq (-1)^{i+1}\binom{d+1}{i}$.]\\
\Procedure{\emph{\Call{CharacterizationTest}{}}}{
$N_{\ref{alg:subroutines}} \gets O(d^2)$ \;
\For{$N_{\ref{alg:subroutines}}$ times }{
\For{$j\in \{1,\dots, d+1\}$}{
    \For{$t\in \{0,\dots, d+1\}$}{
      Sample $\bm{p}\sim\mathcal{N}(\bm{0},j^2(t^2+1) I),\bm{q}\sim\mathcal{N}(\bm{0}, I)$;\Comment{[$j^2(t^2+1)$ vs. $1$ Test.]}\\
      \textbf{Reject} if $\sum_{i=0}^{d+1}\alpha_i\cdot f(\bm{p}+i\bm{q})\neq 0$\;
      Sample $\bm{p}\sim\mathcal{N}(\bm{0},j^2 I),\bm{q}\sim\mathcal{N}(\bm{0},(t^2+1) I)$;\Comment{[$j^2$ vs. $t^2+1$ Test.]}\\
      \textbf{Reject} if $\sum_{i=0}^{d+1}\alpha_i\cdot f(\bm{p}+i\bm{q})\neq 0$\;
      }
      Sample $\bm{p},\bm{q}\sim\mathcal{N}(\bm{0},j^2 I)$;\Comment{[$j^2$ vs. $j^2$ Test.]}\\
      \textbf{Reject} if $\sum_{i=0}^{d+1}\alpha_i\cdot f(\bm{p}+i\bm{q})\neq 0$\;
    }}
    \textbf{Accept}\;
    }

\Procedure{\emph{\Call{Query-$g$}{$\bm{p}$}}}{
$\srad\gets (3d)^{-6}$\;
\If{$\bm{p}\in \ball(\bm{0},\srad)$}{
\textbf{return} \emph{\Call{Query-$g$-InBall}{$\bm{p}$}}\;}

\For{$i\in \{1,\dots,d+1\}$}{
 $c_i \gets i\srad/((d+1)\Vert \bm{p}\Vert_{2})$\;
 $v(c_i)\gets$
\emph{\Call{Query-$g$-InBall}{$c_i\bm{p}$}}\;
}
Let $p_{\bm p} \colon \mathbb{R}\to\mathbb{R}$ be the unique degree-$d$ polynomial such that $p_{\bm p}(c_i) = v(c_i)$  for $i \in [d+1]$\;
\textbf{return} $p_{\bm p}(1)$\;
}

\Procedure{\emph{\Call{Query-$g$-InBall}{$\bm{p}$}}}{
$N'_{\ref{alg:subroutines}} \gets O\big(\log\frac{1}{\varepsilon}\big)$\;
Sample $\bm q_1,\dots,\bm q_{N'_{\ref{alg:subroutines}}}\sim \mathcal{N}(\bm 0, I)$\;
\textbf{Reject} if there exists $j\in\set{2,\dots,N'_{\ref{alg:subroutines}}}$ such that $ \sum_{i=1}^{d+1} \alpha_i \cdot f(\bm{p}+i\bm{q}_1)\neq \sum_{i=1}^{d+1} \alpha_i \cdot f(\bm{p}+i\bm{q}_j)$\;
\textbf{return} $ \sum_{i=1}^{d+1} \alpha_i \cdot f(\bm{p}+i\bm{q}_1)$\;
}
\end{algorithm}

The following lemma records the properties of $g$ that will be guaranteed by our tester. 

\begin{restatable}{lem}{exactmain}\label{lem:main_lemma_well_definedness_and_querying_g}
    If \Call{CharacterizationTest}{} fails with probability at most $2/3$, then $g$ is a degree-$d$ polynomial, and furthermore for any $\bm{p}\in\mathbb{R}^n$,  $g(\bm{p})=$ \Call{Query-$g$}{$\bm{p}$} with probability at least $1-\frac{\varepsilon}{2}$.
\end{restatable}

We prove the main theorem of this section assuming  \autoref{lem:main_lemma_well_definedness_and_querying_g} holds; we restate it next for convenience.

\exact*

\begin{proof}[{Proof of \autoref{thm:low_degree_test}}]


First we analyze the query complexity.
\Call{CharacterizationTest}{} performs $O(d^2)$ independent tests, each of which requires $O(d)$ evaluations of $f$, and is repeated  $N_{\ref{alg:subroutines}}=O(d^2)$ times. \Call{Query-$g$-InBall}{} samples $N'_{\ref{alg:subroutines}}=O(\log(1/\eps))$ points, each requiring $O(d)$ evaluations of $f$. \Call{Query-$g$}{} picks $O(d)$ points in $\ball(\bm 0,\srad)$ and calls \Call{Query-$g$-InBall}{} on them. \Call{LowDegreeTester}{} calls \Call{CharacterizationTest}{} once, and then calls \Call{Query-$g$}{}, $N_{\ref{alg:low_degree_main_algorithm}}=O(1/\eps)$ times. Altogether, our algorithm makes 
$O(d^5+\frac{d^2}{\eps}\log(\frac{1}{\eps}))$ queries.

Next, we argue that the tester is correct. If $f$ is a degree-$d$ polynomial, then it accepts with probability $1$. Indeed, in this case $f$ restricted to a line $\bm{p} + i {\bm q}$ is also a degree-$d$ polynomial, $g = f$, and all of the tests pass with probability $1$.

Now, assume that $f$ is $\varepsilon$-far from any degree-$d$ polynomial (according to $\mathcal{D}$). If \Call{CharacterizationTest}{} fails with probability at least $2/3$, then we reject with probability at least $2/3$. Otherwise, 
by \autoref{lem:main_lemma_well_definedness_and_querying_g}, $g$ is a degree-$d$ polynomial and so
$\Pr_{\bm p \sim \mathcal{D}}[f(\bm p) \neq g(\bm p)] > \varepsilon$. 
The probability that we do not reject in any of the $N_{\ref{alg:low_degree_main_algorithm}}$ steps of \autoref{alg:low_degree_main_algorithm} is at most the probability that $f(\bm{p})=g(\bm{p})$ or that  \emph{\Call{Query-$g$}{$\bm{p}$}}, instead of rejecting, returned some value other than $g(\bm{p})$. The latter happens with probability at most $\frac{\varepsilon}{2}$ by \autoref{lem:main_lemma_well_definedness_and_querying_g}, and so
\[\Pr_{\bm{p}\sim \mathcal{D}}[f(\bm{p})=g(\bm{p}) \lor g(\bm{p}) \neq \Call{Query-\mbox{$g$}}{\bm{p}}]\leq 1-\varepsilon +\frac{\varepsilon}{2}\leq 1-\frac{\varepsilon}{2}.\]
Thus, \autoref{alg:low_degree_main_algorithm} accepts with probability at most $(1-\frac{\varepsilon}{2})^{N_{\ref{alg:low_degree_main_algorithm}}}< \frac{1}{3}$, by choosing the constant in  $N_{\ref{alg:low_degree_main_algorithm}}=O(\varepsilon^{-1})$ to be sufficiently large. \qedhere

\end{proof}

In the remainder of this section we will prove \autoref{lem:main_lemma_well_definedness_and_querying_g}. First, in \autoref{sec:poly-rep-in-ball}, we show that $g$ agrees with a degree-$d$, univariate polynomial on every line segment in $\ball(\bm 0, \srad)$. Then, we show that $g$ is consistent with a degree-$d,n$-variate polynomial within $\ball(\bm 0, \srad)$. Finally, by the fact that for points outside $\ball(\bm 0,\srad)$, $g$ is defined by interpolating evaluations out of $\ball(\bm 0,\srad)$, we show that it is a degree-$d,n$-variate polynomial on $\mathbb{R}^n$.

\subsection{Polynomial Representation on Every Line Within the Ball}\label{sec:poly-rep-in-ball}

We will prove that if \Call{CharacterizationTest}{} passes with high probability, then $g$ is consistent with a degree-$d$ polynomial when projected to any line segment that lies within the open ball $\ball(\bm 0, \srad)$ for $\srad = (3d)^{-6}$.

For $\bm{a},\bm{b}\in \ball(\bm{0},\srad)$, we will denote by $\pline_{\bm{a},\bm{b}}^{\ball}$ the line segment obtained by restricting the line $\pline_{\bm{a},\bm{b}}$ to the ball $\ball(\bm{0},\srad)$. 
The main theorem of this section states that evaluations of $g$ on every point on any line segment within the open ball $\ball(\bm{0},\srad)$, are consistent with a unique, univariate, degree-$d$ polynomial.

\begin{thm}{(Polynomial Representation on Lines)}\label{lem:g-on-lines-in-B-is-poly}
If \Call{CharacterizationTest}{} fails with probability at most $2/3$, and $f$ is bounded on $\ball(\bm 0,2d\sqrt{n})$,
then for every $\bm{a},\bm{b}\in \ball(\bm 0,\srad)$, 
the univariate function $g_{\bm a, \bm b}(x) = g(\bm a + x \bm b)$ 
defined on points $x \in \pline_{\bm a,\bm b}^{\ball}$ is a degree-$d$, univariate polynomial. 
\end{thm}

In order to prove this theorem we will need the following auxiliary lemmas.

\begin{lem}\label{thm:characterization_in_small_ball}
    If \Call{CharacterizationTest}{} fails with probability at most $2/3$, then for every
    $\bm{p},\bm{q}\in \ball(\bm{0},\srad)$, for all sufficiently small $h>0$, such that $\bm p+ih\bm q\in\ball(\bm 0,r)$ for every $i\in[d+1]$, $\sum_{i=0}^{d+1}\alpha_i\cdot g(\bm{p}+ih\bm{q})=0$.
\end{lem}

\begin{lem}\label{lem:g-is-bounded}
If $f$ is bounded on $\ball(\bm 0,2d\sqrt{n})$, then $g$ is bounded on $\ball(\bm 0,\srad)$.
\end{lem}

We prove \autoref{lem:g-on-lines-in-B-is-poly} assuming these lemmas, and prove them afterwards. 

\begin{proof}[Proof of \autoref{lem:g-on-lines-in-B-is-poly}] Since $f$ is bounded on $\ball(\bm 0,2d\sqrt{n})$, by \autoref{lem:g-is-bounded}, $g$ is bounded on $\ball(\bm 0,\srad)$. Fix some $\bm{a},\bm{b}\in \ball(\bm{0},\srad)$. We would like to show that $g_{\bm a, \bm b}(x)$ is consistent with a degree-$d$ polynomial on every point $x$ in $\{x \in \mathbb{R} : \bm a + x \bm b \in \ball(\bm 0, \srad) \}$; fix such an $x$. By the \localCharacterization, it suffices to show that for all sufficiently small $h>0$, satisfying $\bm a+(x+ih)\bm b\in\pline_{\bm a,\bm b}^{\ball}$ for every $i\in[d+1]$,
\[  \Delta^{(d+1)}_h[g_{\bm a, \bm b}](x) =\sum_{i=0}^{d+1}\alpha_i\cdot g_{\bm a,\bm b}(x+ih) = \sum_{i=0}^{d+1} \alpha_i \cdot g(\bm a + x \bm b + i h \bm b)= 0. \]
From \autoref{thm:characterization_in_small_ball}, it follows that $\sum_{i=0}^{d+1}\alpha_i\cdot g(\bm{p}+ih\bm{q}) =0$ for every $\bm{p},\bm{q}\in \ball(\bm{0},\srad)$ and all sufficiently small $h>0$, satisfying $\bm p+ih\bm q\in\ball(\bm 0,r)$ for every $i\in[d+1]$. Let $\bm p \triangleq \bm a + x \bm b$ and $\bm q \triangleq \bm b$. Observe that $\bm p,\bm q\in\ball(\bm 0,\srad)$, and therefore since $\ball(\bm 0, \srad)$ is an open ball, $\bm p + ih\bm q \in \ball(\bm 0, \srad)$ for every $i\in[d+1]$.
Thus,
\[  \Delta^{(d+1)}_h[g_{\bm a, \bm b}](x) =  \sum_{i=0}^{d+1} \alpha_i \cdot g(\bm a + x \bm b + i h \bm b) =   \sum_{i=0}^{d+1} g(\bm p + ih\bm q) = 0. \qedhere \]
\end{proof}

In the remainder of this subsection we prove \autoref{thm:characterization_in_small_ball} and \autoref{lem:g-is-bounded}. For this, it will be convenient to let $\rho$ denote the \emph{smallest} upper-bound on the  probability that each of the tests in the  \Call{CharacterizationTest}{} failed.
That is, for every $j\in [d+1]$ and $t\in \{0,\dots,d+1\}$, $\rho$ is the smallest value such that
 \begin{align}
      \Pr_{\substack{\bm{p}\sim\mathcal{N}(\bm{0},j^2(t^2+1) I)\\\bm{q}\sim\mathcal{N}(\bm{0}, I)}}\left[\sum_{i=0}^{d+1}\alpha_i\cdot f(\bm{p}+i\bm{q})\neq 0\right]&\leq \rho, && \text{[$j^2(t^2+1)$ vs. $1$ Test.]} \label{eq:test_bound_1}\\
      \Pr_{\substack{\bm{p}\sim\mathcal{N}(\bm{0},j^2 I)\\\bm{q}\sim\mathcal{N}(\bm{0},(t^2+1) I)}}\left[\sum_{i=0}^{d+1}\alpha_i\cdot f(\bm{p}+i\bm{q})\neq 0\right]&\leq \rho, && \text{[$j^2$ vs. $t^2+1$ Test.]}\label{eq:test_bound_2}\\
      \Pr_{\substack{\bm{p}\sim\mathcal{N}(\bm{0},j^2 I)\\\bm{q}\sim\mathcal{N}(\bm{0},j^2 I)}}\left[\sum_{i=0}^{d+1}\alpha_i\cdot f(\bm{p}+i\bm{q})\neq 0\right]&\leq \rho. && \text{[$j^2$ vs. $j^2$ Test.]}\label{eq:test_bound_3}
  \end{align}

A bound on the rejection probability of \Call{CharacterizationTest}{} implies the following bound on $\rho$.

\begin{claim}\label{claim:rho-is-small}
If \Call{CharacterizationTest}{} fails with probability at most $2/3$, then $\rho$ is at most $(30d)^{-2}$.
\end{claim}
\begin{proof}
Each of the tests \eqref{eq:test_bound_1},~\eqref{eq:test_bound_2} and \eqref{eq:test_bound_3} are invoked $N_{\ref{alg:subroutines}}=O(d^2)$ in the \Call{CharacterizationTest}{}. If any of these tests fail with probability more than $1/(30d)^{2}$, then \Call{CharacterizationTest}{} passes with probability at most $(1-\frac{1}{(30d)^2})^{O(d^2)}<1/3$, which contradicts our assumption.
\end{proof}

The proof of \autoref{thm:characterization_in_small_ball} will heavily rely on the fact that if $\rho$ is small then  $g_{\bm q_1}$ and $g_{\bm q_2}$ agree on points in $\ball(\bm 0, \srad)$ with high probability. 

\begin{lem}\label{lem:well_defined_of_g}
   For every $\bm{p}\in \ball(\bm{0},\srad)$, and every $t\in  \{0,\ldots, d+1\}$, $$\Pr_{\substack{\bm{q}_1\sim\mathcal{N}(\bm{0},(t^2+1)I)\\ \bm{q}_2\sim\mathcal{N}(\bm{0},I)}}\left[g_{\bm{q}_1}(\bm{p})\neq g_{\bm{q}_2}(\bm{p})\right]\leq 4d\rho+48d^5\srad. $$ 
\end{lem}
\begin{proof}
    Let $t\in \{0,\ldots, d+1\}$ and fix some $\bm{p}\in \ball(\bm{0},\srad)$. We will bound the probability that $g_{\bm{q}_1}(\bm{p})$ and $q_{\bm{q}_2}(\bm{p})$ are different from $\sum_{i=1}^{d+1}\sum_{j=1}^{d+1} \alpha_i \alpha_j \cdot f(\bm{p}+i\bm{q}_1+j\bm{q}_2)$; the lemma will then follow by a union bound.
    
    By definition, $g_{\bm{q}_2}(\bm{p}) = \sum_{i=1}^{d+1} \alpha_i \cdot f(\bm{p}+i\bm{q}_2)$. Fixing $i \in [d+1]$, we have
\begin{align*}
	&\Pr_{\substack{\bm{q}_1\sim\mathcal{N}(\bm{0},(t^2+1)I)\\ \bm{q}_2\sim\mathcal{N}(\bm{0},I)}} [ f(\underbrace{\bm{p}+i \bm{q}_1}_{\triangleq\bm m}) \neq g_{\bm{q}_2}(\bm{p}+i\bm{q}_1) ]
	= \Pr_{\substack{\bm m \sim \mathcal{N}(\bm p,i^2(t^2+1)I)\\
	\bm{q}_2 \sim \mathcal{N}(\bm{0},I)}} \Big[ f(\bm m) \neq \sum_{j=1}^{d+1} \alpha_j \cdot f(\bm{m} +j \bm{q}_2) \Big]  \\
	&\leq \Pr_{\substack{\bm{m}\sim\mathcal{N}(\bm{0},i^2(t^2+1)I)\\ \bm{q}_2\sim\mathcal{N}(\bm{0},I)}} \Big[ \sum_{j=0}^{d+1} \alpha_j \cdot f(\bm{m} +j \bm{q}_2) \neq 0\Big] + 2\dtv(\mathcal{N}(\bm{0},i^2(t^2+1)I), \mathcal{N}(\bm{p},i^2(t^2+1)I))   \\
	&\leq \rho + i^2(t^2+1)\srad  \tag{By \eqref{eq:test_bound_1} and \autoref{lem:usefulDTVBound}}\leq \rho+20d^4\srad.
\end{align*}

By a similar calculation, for every $j \in [d+1]$, we have that
\begin{align*}
	&\Pr_{\substack{\bm{q}_1\sim\mathcal{N}(\bm{0},(t^2+1)I)\\ \bm{q}_2\sim\mathcal{N}(\bm{0},I)}} \big[ f(\underbrace{\bm{p}+j\bm{q}_2}_{\triangleq\bm m}) \neq g_{\bm{q}_1}(\bm{p}+j\bm{q}_2) \big]
	=\Pr_{\substack{ \bm{m}\sim\mathcal{N}(\bm{p},j^2I)\\ \bm{q}_1\sim\mathcal{N}(\bm{0},(t^2+1)I)}}\Big[f(\bm{m})\neq\sum_{i=1}^{d+1}\alpha_i\cdot f(\bm{m}+i\bm{q}_1)\Big]\\
	&\leq\Pr_{\substack{ \bm{m}\sim\mathcal{N}(\bm{0},j^2I)\\ \bm{q}_1\sim\mathcal{N}(\bm{0},(t^2+1)I)}}\left[\sum_{i=0}^{d+1}\alpha_1\cdot f(\bm{m}+i\bm{q}_1)\neq 0\right] + 2\dtv(\mathcal{N}(\bm{0},j^2I),\mathcal{N}(\bm{p},j^2I))\\
	&\leq \rho + j^2\srad \tag{By \eqref{eq:test_bound_2} and \autoref{lem:usefulDTVBound}}\leq \rho+4d^2\srad.
\end{align*}

Taking a union bound over $i \in [d+1]$ and $j \in [d+1]$ respectively, it follows that
\begin{align*}
	\Pr_{\substack{\bm{q}_1\sim\mathcal{N}(\bm{0},(t^2+1)I)\\ \bm{q}_2\sim\mathcal{N}(\bm{0},I)}} \Bigg[ \underbrace{\sum_{i=1}^{d+1} \alpha_i  \cdot f(\bm{p}+i\bm{q}_1)}_{=g_{\bm{q}_1}(\bm{p})} \neq \sum_{i=1}^{d+1} \sum_{j=1}^{d+1} \alpha_i \alpha_j  \cdot f((\bm{p}+i\bm{q}_1) + j\bm{q}_2) \Bigg] &\leq (d+1)(\rho+20d^4\srad),\\
	\Pr_{\substack{\bm{q}_1\sim\mathcal{N}(\bm{0},(t^2+1)I)\\ \bm{q}_2\sim\mathcal{N}(\bm{0},I)}} \Bigg[ \underbrace{\sum_{j=1}^{d+1} \alpha_j \cdot f(\bm{p}+j\bm{q}_2)}_{=g_{\bm{q}_2}(\bm{p})} \neq \sum_{j=1}^{d+1} \sum_{i=1}^{d+1} \alpha_i \alpha_j \cdot f((\bm{p}+j\bm{q}_2) + i\bm{q}_1) \Bigg] &\leq (d+1)(\rho +4d^2\srad).
\end{align*}
The first inequality is at most $2d\rho +40d^5\srad$, while the second is at most $2d\rho+8d^3\srad$.
Thus, by a union bound over the two previous inequalities we can conclude that 
\[ \Pr_{\substack{\bm{q}_1\sim\mathcal{N}(\bm{0},(t^2+1)I)\\ \bm{q}_2\sim\mathcal{N}(\bm{0},I)}} [g_{\bm{q}_1}(\bm{p}) \neq g_{\bm{q}_2}(\bm{p})] \leq 4d\rho+48d^5\srad. \qedhere \]
\end{proof}

The next corollary follows immediately by instantiating the parameters in the previous lemma. 

\begin{cor}\label{cor:replacement-is-good}
    If \Call{CharacterizationTest}{} fails with probability at most $2/3$, then for every $\bm p \in \ball(\bm 0,\srad)$ and every  $t\in \{0,\ldots, d+1\}$,
    \[ \Pr_{\bm{q}\sim\mathcal{N}(\bm{0},(t^2+1)I)}[g(\bm p)\neq g_{\bm{q}}(\bm{p})] \leq \frac{1}{7d}. \]
\end{cor}
\begin{proof}
 Observe that for any $t \in \{0,\ldots, d+1\}$, and any $\bm p\in\ball(\bm 0,\srad)$,
    \begin{align*}
        \Pr_{\bm{q}\sim\mathcal{N}(\bm{0},(t^2+1)I)}[g(\bm p)\neq g_{\bm{q}}(\bm{p})]&\leq\Pr_{\bm{q}_1\sim\mathcal{N}(\bm{0},I)}[g(\bm{p})\neq g_{\bm{q}_1}(\bm{p})]+\Pr_{\substack{\bm{q}\sim\mathcal{N}(\bm{0},(t^2+1)I)\\ \bm{q}_1\sim\mathcal{N}(\bm{0},I)}}[g_{\bm{q}_1}(\bm p)\neq g_{\bm{q}}(\bm{p})].
    \end{align*}
    By \autoref{lem:well_defined_of_g}, this is at most $2(4d\rho+48d^5\srad)$, where for the first term, we have used the fact that $g_{\bm q_1}(\bm p)$ is defined as the majority of $\bm q \sim \mathcal{N}(\bm 0,I)$.
   By \autoref{claim:rho-is-small} and by our choice of $\srad =(3d)^{-6}$, this probability is at most $1/(7d)$.
\end{proof}

We are now ready to prove \autoref{thm:characterization_in_small_ball}.

\begin{proof}[Proof of \autoref{thm:characterization_in_small_ball}]
Fix $\bm{p},\bm{q} \in \ball(\bm{0},\srad)$, and let $h>0$ be such that $\bm p +ih \bm q \in \ball(\bm 0,\srad)$ for every $i\in[d+1]$; note that $h$ exists as $\ball(\bm 0,\srad)$ is an \emph{open} ball containing $\bm p$. We will argue that the following hold simultaneously with non-zero probability over $\bm{q}_1,\bm{q}_2 \sim \mathcal{N}(\bm{0},I)$:
\begin{equation}\label{eq:replacement-is-good}
    \sum_{i=0}^{d+1} \alpha_i \cdot g(\bm{p}+ih\bm{q}) =\sum_{i=0}^{d+1} \alpha_i\cdot g_{\bm{q}_1+i\bm{q}_2}(\bm{p}+ih\bm{q}),
\end{equation}
\begin{equation}\label{eq:post-replacement-is-good}
    \sum_{i=0}^{d+1} \alpha_i \cdot f(\bm{p}+ j\bm{q}_1+ i(h\bm{q}+j\bm{q}_2)) =0\text{, for every }j \in [d+1].
\end{equation}

We will complete the proof assuming that these bounds hold. Fix any $\bm{q}_1,\bm{q}_2$ satisfying both~\eqref{eq:replacement-is-good}, and~\eqref{eq:post-replacement-is-good}. Then,
\begin{align*}
    \sum_{i=0}^{d+1} \alpha_i \cdot g(\bm{p}+ih\bm{q}) &= \sum_{i=0}^{d+1} \alpha_i\cdot g_{\bm{q}_1+i\bm{q}_2}(\bm{p}+ ih\bm{q}) \tag{By~\eqref{eq:replacement-is-good}} \\
    &=\sum_{i=0}^{d+1} \alpha_i\left( \sum_{j=1}^{d+1} \alpha_j \cdot f(\bm{p}+ih\bm{q}+j(\bm{q}_1+i\bm{q}_2))\right) \\
    &=\sum_{j=1}^{d+1} \alpha_j\left( \sum_{i=0}^{d+1} \alpha_i \cdot f(\bm{p}+j\bm{q}_1+i(h\bm{q}+j\bm{q}_2))\right) \\
    &= \sum_{j=1}^{d+1}\alpha_j \cdot 0 = 0.\tag{By~\eqref{eq:post-replacement-is-good}}
\end{align*}

 Next, we argue that~\eqref{eq:replacement-is-good} and~\eqref{eq:post-replacement-is-good} hold separately with sufficiently high probability.
{\allowdisplaybreaks
\begin{align*}
    \Pr_{\bm{q}_1,\bm{q}_2 \sim \mathcal{N}(\bm{0},I)}[\eqref{eq:replacement-is-good}]&=\Pr_{\bm{q}_1,\bm{q}_2 \sim \mathcal{N}(\bm{0},I)} \bigg[\sum_{i=0}^{d+1} \alpha_i \cdot g(\bm{p}+ih\bm{q}) =\sum_{i=0}^{d+1}\alpha_i\cdot g_{\bm{q}_1+i\bm{q}_2}(\bm{p}+ih\bm{q})\bigg]	\\ 
	&\geq \Pr_{\bm{q}_1,\bm{q}_2 \sim \mathcal{N}(\bm{0},I)} \left[g(\bm{p}+ih\bm{q}) = g_{\bm{q}_1+i\bm{q}_2}(\bm{p}+ih\bm{q}),~~ \forall i \in \{0,\ldots, d+1\} \right] \\
	&=1-\Pr_{\bm{q}_1,\bm{q}_2\sim\mathcal{N}(\bm{0},I)} \left[\exists i\in\{0,\ldots,d+1\}:g(\bm{p}+ih\bm{q}) \neq g_{\bm{q}_1+i\bm{q}_2}(\bm{p}+ih\bm{q}) \right] \\
	&\geq 1-\sum_{i=0}^{d+1} \Pr_{\bm{q}_1,\bm{q}_2 \sim \mathcal{N}(\bm{0},I)}\left[g(\bm{p}+ih\bm{q}) \neq g_{\bm{q}_1+i\bm{q}_2}(\bm{p}+ih\bm{q}) \right] \tag{By union bound} \\
	&= 1-\sum_{i=0}^{d+1}\Pr_{\bm{m}\sim\mathcal{N}(\bm{0},(i^2+1)I)}\left[g(\bm{p}+ih\bm{q})\neq g_{\bm{m}}(\bm{p}+ih\bm{q}) \right] \tag{Letting $\bm{m}\triangleq\bm{q}_1+i\bm{q}_2$} \\
	&\geq 1-\frac{d+2}{7d}>\frac{1}{2}. \tag{Applying \autoref{cor:replacement-is-good}, as $\bm{p}+ih\bm{q}\in \ball(\bm{0},\srad)$}
\end{align*}
}

For~\eqref{eq:post-replacement-is-good}, fix some $j \in [d+1]$, then
\begin{align*}
    &\Pr_{\bm{q}_1,\bm{q}_2 \sim \mathcal{N}(\bm{0},I)} \bigg[ \sum_{i=0}^{d+1} \alpha_i \cdot f(\underbrace{\bm{p}+j\bm{q}_1}_{\triangleq\bm{z}_1} + i (\underbrace{h\bm{q}+j\bm{q}_2}_{\triangleq\bm{z}_2} )) \neq 0 \bigg]= \Pr_{\substack{\bm{z}_1 \sim \mathcal{N}(\bm{p},j^2I) \\ \bm{z}_2 \sim \mathcal{N}(h\bm{q},j^2I)}} \bigg[ \sum_{i=0}^{d+1} \alpha_i \cdot f(\bm{z}_1 + i \bm{z}_2) \neq 0\bigg] \\
    &\leq  \Pr_{\substack{\bm{z}_1 \sim \mathcal{N}(\bm{0},j^2I) \\ \bm{z}_2 \sim \mathcal{N}(\bm{0},j^2I)}} \bigg[ \sum_{i=0}^{d+1} \alpha_i f(\bm{z}_1 + i \bm{z}_2) \neq 0\bigg]+ 2 ( \dtv(\mathcal{N}(\bm{0},j^2I), \mathcal{N}(\bm{p},j^2I)) + \dtv(\mathcal{N}(\bm{0},j^2I), \mathcal{N}(h\bm{q},j^2I) ) ) \\
    &\leq  \rho +j^2\srad + j^2 h\srad\leq \rho + 8d^2\srad. \tag{By \eqref{eq:test_bound_3} and \autoref{lem:usefulDTVBound}}
\end{align*}
By a union bound over all $j \in [d+1]$, 
\[ \Pr_{\bm{q}_1,\bm{q}_2 \sim \mathcal{N}(\bm{0},I)}[\eqref{eq:post-replacement-is-good}]=\Pr_{\bm{q}_1,\bm{q}_2 \sim \mathcal{N}(\bm{0},I)} \bigg[ \forall j \in [d+1], \sum_{i=0}^{d+1} \alpha_i \cdot f(\bm{p}+j\bm{q}_1 + i(h\bm{q}+j\bm{q}_2)) =0 \bigg] \geq 1-(2d\rho + 16d^3\srad), \]
which is at least $2/3$ by our choice of $r=(3d)^{-6}$ and \autoref{claim:rho-is-small}.
A final union bound over~\eqref{eq:replacement-is-good}, and~\eqref{eq:post-replacement-is-good} concludes that both hold simultaneously with non-zero probability.
\end{proof}

Finally, in order to conclude that $g$ is indeed a polynomial by using \localCharacterization on lines within $\ball(\bm 0, \srad)$, we will argue that $g$ is bounded in $\ball(\bm 0, \srad)$

\begin{proof}[Proof of \autoref{lem:g-is-bounded}]
It suffices to prove $g_{\bm q}(\bm p)$ is bounded for every $\bm p\in\ball(\bm 0,\srad)$, and every $\bm q\in\R^n$ such that $g_{\bm q}(\bm p) = g(\bm p)$. By \autoref{cor:replacement-is-good}, $g(\bm p) = g_{\bm q}(\bm p)$ with probability at least $1-1/7d$ for $\bm q \sim \mathcal{N}(\bm 0,I)$. By \cite[Theorem 2.9]{blum_hopcroft_kannan_2020}, at least $99\%$ of the mass in $\cN(\bm 0,I)$ lies in the annulus $|r-\sqrt{n}|\leq 20$. Therefore, we can conclude that $g(\bm p)$ agrees with $g_{\bm q}(\bm p)$ for $\bm q$ satisfying $|\|\bm q\|_2-\sqrt{n}|\leq 20$. Note that $g_{\bm q}(\bm p)$ depends only on $\{f(\bm p+i\bm q)\}_{i=0}^{d+1}$, and $\max_i\{\|\bm p+i\bm q\|_2\}\leq (d+2)\max\{\|\bm p\|_2,\|\bm q\|_2\}\leq(d+2)(\sqrt{n}+20)$. Thus, if $f$ is bounded on $\ball(\bm 0,2d\sqrt{n})$, then $g$ is bounded on $\ball(\bm 0,\srad)$. 
\end{proof}

\subsection{Polynomial Representation Within a Hypercube}\label{sec:poly-rep-in-cube}

Let $m \in \mathbb{R}$ be be the largest value, such that the hypercube ${[-m,m]}^n$ is strictly contained within the open ball $\ball(\bm 0,\srad)$; in particular, $\srad/ (2 \sqrt{n})\leq m<\srad/\sqrt{n}$. 
We show that if the conditions of \autoref{lem:g-on-lines-in-B-is-poly} are met, then $g$ is consistent with a degree-$d$ multivariate polynomial on ${[-m,m]}^n$.
This is done in two steps; first, in~\autoref{lem:local_to_global} we show that $g$ is consistent with a finite bounded degree polynomial.
Then, in~\autoref{lem:radial_lines_consistent_force_low_degree}, we show that this degree can be reduced to $d$.

Let $\bm e_i$ denote the $i$th standard basis vector, defined as $\bm e_{i,j} = 0$ if $j \neq i$ and $\bm e_{i,i} =1$.

\begin{lem}{(Local to Global)}
    \label{lem:local_to_global}
    Let $m > 0$, and let $h\colon {[-m,m]}^n \to \mathbb{R}$.
    If for every $i \in [n]$, and $\bm a \in {[-m,m]}^n$ such that $ a_i = 0$, the restriction of $h$ to the line segment $\pline_{\bm a, \bm e_i}$, the univariate function $h_{\bm a, \bm e_i}\colon[-m,m]\to\R$ is consistent with a degree-$d$ univariate polynomial on the interval $[-m,m]$, then $h$ is consistent with an $n$-variate polynomial of degree at most $dn$.

\end{lem}

\begin{proof}
We will show that $h$ is a degree-$dn$ polynomial  by induction on the dimension $n$. For the base case when $n=1$, we have that $h=h_{\bm 0,  1}$ and therefore is of degree $d$ by assumption.

Assume the statement is true for dimension $n-1$. 
Let $c\in[-m,m]$ and define $h^{(c)}\colon[-m,m]^{n-1}\to\mathbb{R}$ as 
\[h^{(c)}(x_1,\dots,x_{n-1}) \triangleq h(x_1,\dots,x_{n-1},c).\]
We will argue that $h^{(c)}$ satisfies the conditions of \autoref{lem:local_to_global}: 
Fix $i\in [n-1]$, and  $\bm{a}\in{[-m,m]}^{n-1}$ with $a_i=0$, and define $\bm a^\uparrow=(a_1,\dots,a_{n-1},c) \in {[-m,m]}^n$ to be an extension of $\bm a$ to dimension $n$.  By assumption, $h_{\bm a^\uparrow, \bm e_i}\colon[-m,m]\to\mathbb{R}$ is a degree-$d$ polynomial. For every $x\in[-m,m]$, we have
\[h_{\bm a^\uparrow, \bm e_i}(x)=h(a_1,\dots,a_{i-1},x,a_{i+1},\dots,a_{n-1},c)=h^{(c)}(\bm{a}+x\bm{e_i})=h^{(c)}_{\bm{a,e_i}}(x),\]
and so $h^{(c)}_{\bm{a,e_i}}(x)$ is a degree-$d$ polynomial on the domain $[-m,m]$. Thus, by the inductive hypothesis we can conclude that $h^{(c)}\colon{[-m,m]}^{n-1}\to\mathbb{R}$ is a degree-$d(n-1)$ multivariate polynomial.
\begin{figure}
    \centering
    \begin{tikzpicture}[scale=0.9]
        \filldraw[color=black!60, fill=green!9, very thick](3,0.25) circle (3.9);
        \filldraw[color=black!60, fill=red!5, very thick](0,-1) -- (2,-3) -- (6,-2) -- (6,1.5) -- (4,3.5) -- (0,2.5);
        \draw[color=black!60, very thick] (4,0)-- (4,3.5) -- (6,1.5) -- (6,-2) -- cycle;
        \draw[color=purple!20, very thick](3.2,-0.2)-- (3.2,3.3) -- (5.2,1.3) -- (5.2,-2.2) -- cycle;
        \filldraw[color=purple!60, fill=betterYellow!25, very thick](2.4,-0.4)-- (2.4,3.1) -- (4.4,1.1) -- (4.4,-2.4) -- cycle;
        \draw[color=purple!20, very thick] (0.8,-0.8)-- (0.8,2.7) -- (2.8,0.7) -- (2.8,-2.8) -- cycle;
        \draw[color=purple!20, very thick] (1.6,-0.6)-- (1.6,2.9) -- (3.6,0.9) -- (3.6,-2.6) -- cycle;
        \draw[color=black!60, very thick] (0,-1)-- (0,2.5) -- (2,0.5) -- (2,-3) -- cycle;
        \node[text width=1cm] at (3.1,-3) {$c_0$};
        \node[text width=1cm] at (4,-2.8) {$c_1$};
        \node[text width=1cm] at (4.8,-2.6) {$c_2$};
        \node[text width=1cm] at (5.6,-2.4) {$c_3$};
        \node[text width=1.2cm] at (3,-4) {$\ball(\bm 0,\srad)$};
        \node[text width=1.3cm] at (1.4,3.2) {\small ${[-m,m]}^n$};
        \draw[color=black!60, very thick, <->] (6.3,1.6) -- (4.3,3.6);
        \filldraw[fill=green!9,  color =green!9] (5.4,2.6) circle (6pt);
		\node[text width=1cm] at (5.78,2.55)  { $2m$};
		\draw [cyan, very thick, xshift=4cm] plot [smooth, tension=1] coordinates { (-1.6,0.7) (-1.1,1.5) (-0.7,-0.5) (0,0.1) (0.4,-2.4)};
		\draw[color=purple!60, very thick](2.4,-0.4)-- (2.4,3.1) -- (4.4,1.1) -- (4.4,-2.4) -- cycle;
		\draw[color=black!60, very thick] (2,-3) -- (6,-2);
        \draw[color=black!60,  very thick] (0,-1) -- (2.4,-.4);
        \draw[color=black!60,  very thick] (0,2.5) -- (4,3.5);
        \draw[color=black!60, very thick] (2,0.5) -- (6,1.5);
		\node[text width=1cm] at (3.1,0) {$h^{(c_2)}$};
    \end{tikzpicture}
    \caption{The construction of the polynomial $h(a_1,\ldots, a_{n-1})$. $d+1$ slices of the cube ${[-m,m]}^n$ are chosen, where the $i$th slice corresponds to setting $x_n = t = c_i$. The picture depicts setting $t = c_2$ and thus $\delta_{c_i}(t) = 0$ for all $i \neq 2$ and $\delta_{c_2}(t) = 1$, selecting the polynomial representation $h^{(c_2)}$ of $h$ on the $2$nd slice.}\label{fig:local_to_global}
\end{figure}
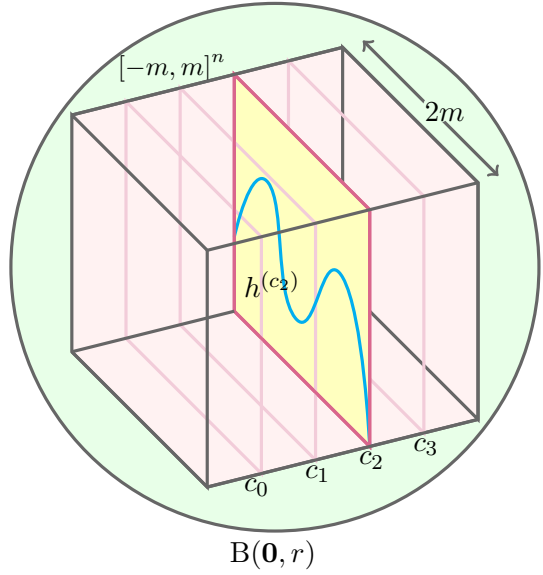

It remains to show that $h$ is a degree-$dn$ polynomial. Let $c_0,c_1,\dots, c_d\in[-m,m]$ be any $d+1$ distinct values. Denote by $\delta_{c_i}$ the unique degree-$d$ polynomial satisfying 

\begin{equation}\label{eq:def-delta}
    \delta_{c_i}(c_j)\triangleq\begin{cases}1 & i=j,\\0 &i\neq j.\end{cases}
\end{equation} 
 Using $\delta_{c_i}$, we will show that $h$ can be written as a polynomial of degree $dn$.

First, we claim that for every fixed $\bm{a}\in {[-m,m]}^{n-1}$ and variable $t$,
  \[h(a_1,\dots,a_{n-1},t) = \sum_{i=0}^{d}\delta_{c_i}(t)h^{(c_i)}(a_1,\dots,a_{n-1}).\]
  To see this, observe that  $h^{(c_i)}(a_1,\ldots, a_{n-1})$ is a constant and therefore $\delta_{c_i}(t) h^{(c_i)}(a_1,\ldots, a_{n-1})$ is a degree-$d$ polynomial. Thus, $\sum_{i=0}^{d}\delta_{c_i}(t)h^{(c_i)}(a_1,\dots,a_{n-1})$ and $h(a_1,\ldots, a_{n-1},t)$ are degree-$d$ polynomials (the latter is by assumption). Furthermore, these degree-$d$ polynomials agree on the $d+1$ distinct points  $c_0,\ldots, c_d$ and therefore they must be equal.
As this equality holds for every $(\bm a,t) \in {[-m,m]}^n$, it follows that for every $\bm x \in {[-m,m]}^n$,
  \[h(x_1,\dots,x_n)= \sum_{i=0}^{d}\underbrace{\delta_{c_i}(x_n)}_{\text{degree } d}\underbrace{h^{(c_i)}(x_1,\dots,x_{n-1})}_{\text{degree } d(n-1)},\]
  which is a degree $dn$ representation of $h$.
 \end{proof}
 
\begin{lem}{(Degree Reduction)}\label{lem:radial_lines_consistent_force_low_degree}
Let $\alpha \in \mathbb{N},m>0$, and $h\colon\R^n\to\R$ be a multivariate polynomial of finite degree $\alpha$. If for every radial line segment in the cube ${[-m,m]}^n$, the restriction of $h$ to that line segment is consistent with a polynomial of degree at most $d$, then $\alpha\leq d$.
\end{lem}

\begin{figure}
    \centering
    \begin{tikzpicture}
        \filldraw[color=black!60, fill=green!9, very thick](0,0) circle (2.2);
        \filldraw[color=black!0, very thick, fill= red!5] (-1.5,-1.5) -- (1.5,-1.5) -- (1.5,1.5) -- (-1.5,1.5) -- cycle;
		\node[text width=1cm] at (0,-2.5)  { $\ball(\bm 0, \srad)$};
		\node[text width=1.3cm] at (0,1.8)  { ${[-m,m]}^n$};
		\node[text width=1cm] at (1.99,0.4)  { $\pline_{\bm 0, \bm b}$};
		\draw[color=red!50, very thick] (-1.5,-1.5) -- (1.5,1.5);
		\draw[color=red!50, very thick] (-1.5,1.5) -- (1.5,-1.5);
		\draw[color=red!50, very thick] (-1.5,0) -- (1.5,0);
		\draw[color=red!50, very thick] (0,1.5) -- (0,-1.5);
		
		\draw[color=red!50, very thick] (0.75,1.5) -- (-0.75,-1.5);
		\draw[color=red!50, very thick] (-0.375,1.5) -- (0.375,-1.5);
		\draw[color=red!50, very thick] (0.375,1.5) -- (-0.375,-1.5);
		
		\draw[color=red!50, very thick] (-0.75,1.5) -- (0.75,-1.5);
		\draw[color=red!50, very thick] (-1.125,1.5) -- (1.125,-1.5);
		\draw[color=red!50, very thick] (1.125,1.5) -- (-1.125,-1.5);
		
		\draw[color=red!50, very thick] (1.5,0.75) -- (-1.5,-0.75);
		\draw[color=red!50, very thick] (1.5,1.125) -- (-1.5,-1.125);
		\draw[color=red!50, very thick] (1.5,-1.125) -- (-1.5,1.125);
		
		\draw[color=red!50, very thick] (1.5,-0.75) -- (-1.5,0.75);
		\draw[color=red!50, very thick] (1.5,-0.375) -- (-1.5,0.375);
		\draw[color=blue!50, very thick] (1.5,0.375) -- (-1.5,-0.375);
		\draw[color=black!60, very thick] (-1.5,-1.5) -- (1.5,-1.5) -- (1.5,1.5) -- (-1.5,1.5) -- cycle;
    \end{tikzpicture}
    \caption{The radial lines $\pline_{\bm 0, \bm b}$ within the hypercube ${[-m,m]}^n$ in two dimensions.}\label{fig:bounded_to_d}
\end{figure}
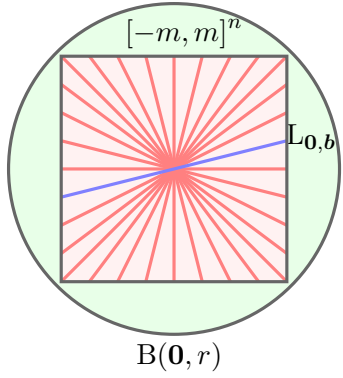

\begin{proof}
Fix some $\bm b \in {[-m,m]}^n$ and consider the radial line $\pline_{\bm 0, \bm b}$. The $n$-variate polynomial $h$, restricted to this line, $h_{\bm 0, \bm b}(x) = h(x \bm b)$ for $x$ such that $x\bm b\in {[-m,m]}^n$, can be written as
\[h(x \bm{b})=\sum_{k=0}^{\alpha}\sum_{i_1+\cdots+i_n=k}c_{i_1,\dots,i_n}\prod_{j=1}^{n}{(x b_j)}^{i_j}=\sum_{k=0}^{\alpha}\bigg(\sum_{i_1+\cdots+i_n=k}c_{i_1,\dots,i_n}\prod_{j=1}^{n}b_j^{i_j}\bigg)x^k,\]
which is a univariate degree-$\alpha$ polynomial in $x$.
Consider the coefficient $c_{\alpha}$ of $x^\alpha$ in $h(x \bm b)$ as a function of $\bm b$,
\[c_\alpha(\bm b) \triangleq \sum_{i_1+\cdots+i_n=\alpha}c_{i_1,\dots,i_n}\prod_{j=1}^{n}b_j^{i_j},\]
this is a $n$-variate polynomial of degree $\alpha$ in the variables $\bm b$. 
Note that $c_\alpha \not\equiv 0$, as otherwise $h$ would have degree less than $\alpha$. 
Fix some $\bm b=\bm b^* \in [-m,m]^n$ such that $c_\alpha(\bm b^*) \neq 0$, such a point exists since $c_{\alpha}$ has finite number of roots, and view $x$ as the only variable; as $c_\alpha(\bm b^*) \neq 0$, $h_{\bm 0, \bm b^*}$ is a univariate polynomial of degree $\alpha$. However, by assumption $h_{\bm 0, \bm b^*}(x)$ has degree at most $d$, and hence $\alpha \leq d$.
\end{proof}

\subsection{Polynomial Representation Everywhere}\label{sec:poly-rep-everywhere}

We are now ready to prove that $g$ is a degree-$d$ polynomial over $\mathbb{R}^n$.

\begin{lem}\label{thm:g-is-mult-d-poly} If \Call{CharacterizationTest}{} fails with probability at most $2/3$, then $g$ is a degree-$d$, $n$-variate polynomial. 
\end{lem}
\begin{proof}
Consider the largest $n$-dimensional hypercube $\cube\triangleq{[-m,m]}^n$ that can be inscribed in the open ball $\ball(\bm{0},\srad)$. 
By \autoref{lem:g-on-lines-in-B-is-poly}, $g$ restricted to any line segment $\pline_{\bm{p},\bm{q}}^{\ball} = \pline_{\bm p, \bm q} \cap \ball(\bm 0,\srad)$ within the ball $\ball(\bm 0, \srad)$ is consistent with a univariate degree-$d$ polynomial, and therefore the same holds for $g$ restricted to any line segment $\pline_{\bm{p},\bm{q}}^{\cube}$, as $\cube \subset \ball(\bm 0,\srad)$. 

By \autoref{lem:local_to_global} and \autoref{lem:radial_lines_consistent_force_low_degree}, we can conclude that $g(\bm{x})\colon {[-m,m]}^n\to\mathbb{R}$ is consistent with a polynomial of degree at most $d$ within $\cube$. Hence, for every $\bm \alpha\in\mathbb{N}^n$ such that $\|\bm\alpha\|_1\leq d$, there exists $c_{\bm\alpha}\in\R$, such that for every $\bm x \in \cube$, we can write
\begin{align}
    g(\bm{x})=\sum_{\bm{\alpha}\in\mathbb{N}^n:\|\bm{\alpha}\|_1\leq d}c_{\bm{\alpha}}\prod_{j=1}^{n}x_j^{\alpha_j}. \label{eq:coefficient_representation}
\end{align}

Next, we argue that $g$ is also consistent with this polynomial representation for every point within $\ball(\bm 0, \srad)$. By \autoref{lem:g-on-lines-in-B-is-poly}, for any $\bm y \in  \ball(\bm{0},\srad) \setminus \cube$ and $x \in \mathbb{R}$, for which $x \bm y \in \ball(\bm 0, \srad)$, it follows that $g(x \bm y)$ has a unique representation as a univariate polynomial. This polynomial must be consistent with (\ref{eq:coefficient_representation}) on any point $x' \bm y \in \cube$, with $x' \in \mathbb{R}$. As these are both polynomials (agreeing on at least $(d+1)$ points), it follows that both polynomial must be consistent on any point on the line segment $\pline_{\bm 0, \bm y}^{\ball}$. As we know that $g$ is consistent with the univariate representation within $\ball(\bm 0, \srad)$, it follows that the representation (\ref{eq:coefficient_representation}) holds for $g(\bm y)$ for any $\bm y \in \ball(\bm 0,\srad)$.

It remains to argue that $g$ is consistent with this degree-$d$ polynomial representation everywhere. Recall that we defined $g(\bm p)$ for $\bm p \not \in \ball(\bm 0, \srad)$, by extrapolating its representation within $\ball(\bm 0,\srad)$ along line $\pline_{\bm 0, \bm p}$, to obtain a representation of $g(x \bm p)$ as a degree-$d$ (univariate) polynomial. 

Thus, $g$ is consistent with a degree-$d$, $n$-variate polynomial over $\mathbb{R}^n$.

\end{proof}

Finally, having established \autoref{thm:g-is-mult-d-poly}, we are ready to prove \autoref{lem:main_lemma_well_definedness_and_querying_g}, restated here for convenience.

\exactmain*

\begin{proof}[Proof of \autoref{lem:main_lemma_well_definedness_and_querying_g}] 
Suppose that \Call{CharacterizationTest}{} fails with probability at most $2/3$, then by \autoref{thm:g-is-mult-d-poly}, $g$ is a degree-$d$ polynomial. 
It remains to bound the probability that $g(\bm p) \neq$ \Call{Query-$g$}{$\bm p$} for $\bm p \in \mathbb{R}^n$. 
To query $g$ on a point $\bm{p} \in \mathbb{R}^n$, \Call{Query-$g$}{$\bm p$} call  \Call{Query-$g$-InBall}{$\bm p$} if $\bm p\in\ball(\bm 0, \srad)$ or otherwise it attempts to obtain $d+1$ distinct points on the line segment $\pline_{\bm 0, \bm p}^{\ball}$ using \Call{Query-$g$-InBall}{$\cdot$} for each  and then interpolate $g$ along this line. For each of these $d+1$ points $\bm s$, \Call{Query-$g$-InBall}{$\bm s$} samples an additional $N'_{\ref{alg:subroutines}}$ points $\bm q_1,\ldots, \bm q_{N'_{\ref{alg:subroutines}}}\sim\cN(\bm 0,I)$, and checks whether  
\[\sum_{i \in [d+1]} \alpha_i \cdot f(\bm s + i \bm q_1) = \sum_{i \in [d+1]} \alpha_i \cdot f(\bm s + i \bm q_j),\]
for all $j \in [N'_{\ref{alg:subroutines}}]$; it fails if any of these checks fail. Note that by the definition of $g_{\bm q}$, this is equivalent to checking whether $g_{\bm q_1}(\bm s) \neq g_{\bm q_j}(\bm s)$. 
By \autoref{cor:replacement-is-good}  the probability that $g_{\bm q_1}(\bm s) \neq g(\bm s)$ is  at most $1/(7d)$, since $\bm s\in \ball(\bm 0,\srad)$. The probability that \Call{Query-$g$-InBall}{$\bm s$} returns an incorrect value is the probability that $g(\bm s) \neq g_{\bm q_1}(\bm s) = g_{\bm q_j}(\bm s)$ for every $\bm q_j$, which is at most ${(7d)}^{-N'_{\ref{alg:subroutines}}}\leq \frac{\varepsilon}{4d}$ by choosing $N'_{\ref{alg:subroutines}} = O(\log (1/ \varepsilon))$. As \Call{Query-$g$}{$\bm p$} evaluate at most $d+1$ points  using \Call{Query-$g$-InBall}{$\cdot$}, the probability that these points are all evaluated correctly, is at least $1-\varepsilon/2$.
\end{proof}

%% file: local_to_global_linear_fucntions.tex
\section{Approximately Testing Polynomials}\label{sec:approximate}

In this section, we generalize our polynomial tester to be robust against noise. Given query access to a function $f\colon \mathbb{R}^n \to \mathbb{R}$ bounded on the ball $\ball(\bm 0, 2d\sqrt{n})$,
and sampling access to an unknown $(\varepsilon/4, R)$-concentrated distribution $\caD$, and constants $\alpha, \varepsilon>0,\beta\geq\alpha$, a \emph{point-wise approximate tester} for degree-$d$ polynomials is an algorithm that distinguishes between the following two cases with probability at least $2/3$:
\begin{itemize}
    \item \textbf{Yes Case:} There exists a degree-$d$ polynomial $h\colon \mathbb{R}^n \to \mathbb{R}$ such that for every $\bm x \in \mathbb{R}^n$,
    \[|f(\bm x)-h(\bm x)| \leq \alpha;\]
    \item \textbf{No Case:} For any degree-$d$ polynomial $h\colon\mathbb{R}^n \to\mathbb{R}$, 
    \[ \Pr_{\bm x \sim \mathcal{D}} \left[|f(\bm x) - h(\bm x)| > \beta\right] > \varepsilon. \]
\end{itemize}

An alternative interpretation of this model is as follows: we would like to design a low-degree tester for a function $f^*\colon \mathbb{R}^n \to \mathbb{R}$; however, on every $\bm p \in \mathbb{R}^n$, we are only able to obtain ``noisy'' evaluations of $f(\bm p)$ within an accuracy of up to $\alpha$.  We represent this by giving the tester query-access to a function $f\colon \mathbb{R}^n \to \mathbb{R}$, such that for every $\bm p \in \mathbb{R}^n$, 
\[ |f^*(\bm p) - f(\bm p)| \leq \alpha .\] 
This setup is quite natural, and captures the setting in which we are only able to observe a small number of bits of precision of the evaluations of $f^*(\bm p)$. The main theorem of this section is the following.\footnote{We note that it is possible to relax the condition on $f$ to be bounded in $\ball(\bm 0, L)$ for some known $L>0$. This then leads to $\beta$ being dependent on $L$ as well. To avoid complicating the parameters, we have chosen to present the less general theorem here.}

\approximateLowDegree*

Our self-corrected function $g$ will be the same as the self-corrected function in the exact case, with one small twist: We use the median rather than the majority, as the median is more robust to errors. 

\paragraph{The Self-Corrected Function.} Let $\srad$ be sufficiently small ($r = (4d)^{-6}$ suffices). We first define our self-corrected function for the points $\bm p \in \ball(\bm 0, \srad)$ as the (weighted) median value of $g_{\bm q}(\bm p) \triangleq \sum_{i=1}^{d+1}\alpha_i\cdot f(\bm p+i\bm q)$, weighted according to the probability of $\bm q \sim \mathcal{N}(\bm 0,I)$. That is, 
\[ g(\bm p) \triangleq \underset{\bm q \sim \mathcal{N}(\bm 0, I)}{\mathsf{med}}[g_{\bm q}(\bm p)]. \]
For points $\bm p \not \in \ball(\bm 0, \srad)$ we define the value of $g$ by extrapolating it from within the ball $\ball(\bm 0, \srad)$ along the radial line $\pline_{\bm 0, \bm p}$. To do so, we will interpolate a univariate polynomial on the line  $\pline_{\bm 0, \bm p}$ using the evaluation of $g$ on $d+1$ points in $\ball(\bm 0, \srad)$. 
For our analysis, it will be convenient to take these points to be  $c_0,\ldots, c_d$, where\footnote{These are the Chebyshev nodes of the $(d+1)$-st Chebyshev polynomial, scaled to lie on $\pline_{\bm 0, \bm p} \cap \ball(\bm 0, \srad)$, as in \autoref{sec:chebyshev}.} $c_i \triangleq (r/\|\bm p\|_2) \cos(\pi(i+1/2)/(d+1))$. Let $p_{\bm p}$ be the unique univariate degree-$d$ polynomial such that $p_{\bm p}(  c_i) = g(\bm p c_i)$ for all $i$. Then, we define $g(\bm p) \triangleq p_{\bm p}(1)$.

Our tester is given in \autoref{alg:low_degree_main_algorithm_approx}, with subroutines in \autoref{alg:subroutines_approx}.

\begin{algorithm}[h]
\caption{Low-Degree Approximate Tester}\label{alg:low_degree_main_algorithm_approx}
\Procedure{\emph{\Call{LowDegreeApproxTester}{$f,d,\mathcal{D},\alpha,\eps,\lrad$}}}{
  \Given {Query access to $f\colon \mathbb{R}^n \to \mathbb{R}$, a degree $d\in\mathbb{N}$, sampling access to an unknown $(\varepsilon/4, R)$-concentrated distribution $\mathcal{D}$, a noise parameter $\alpha > 0$, and a farness parameter $\varepsilon> 0$. 
  }
   $\delta\gets 2^{d+1}\alpha$\;
   $r\gets (4d)^{-6}$\;
    \textbf{Reject} if \emph{\Call{ApproxCharacterizationTest}{}}  \textbf{rejects}\;
    \For{$N_{\ref{alg:low_degree_main_algorithm_approx}} \gets O(\varepsilon^{-1})$ times}{
      Sample $\bm{p} \sim \mathcal{D}$;\\
      \If{$\bm p\in\ball(\bm 0,\lrad)$}{
      \textbf{Reject} if $|f(\bm{p}) -$ \emph{\Call{ApproxQuery-$g$}{$\bm{p}$}}$|>2\cdot 2^{(2n)^{45d}}\lrad^d\delta$, or if \emph{\Call{ApproxQuery-$g$}{$\bm{p}$}}  \textbf{rejects}.}
    }
    \textbf{Accept}.
  }  
\end{algorithm}
\begin{algorithm}[!t]
\caption{Approximate Subroutines}\label{alg:subroutines_approx}
[Recall $\alpha_{i}\triangleq (-1)^{i+1}\binom{d+1}{i}$ and $\delta=2^{d+1}\alpha$.]\\
\Procedure{\emph{\Call{ApproxCharacterizationTest}{}}}{
$N_{\ref{alg:subroutines_approx}} \gets O(d^2)$ \;
\For{$N_{\ref{alg:subroutines_approx}}$ times }{
\For{$j\in \{1,\dots, d+1\}$}{
    \For{$t\in \{0,\dots, d+1\}$}{
      Sample $\bm{p}\sim\mathcal{N}(\bm{0},j^2(t^2+1) I),\bm{q}\sim\mathcal{N}(\bm{0}, I)$;\Comment{[$j^2(t^2+1)$ vs. $1$ Test.]}\\
      \textbf{Reject} if $|\sum_{i=0}^{d+1}\alpha_i\cdot f(\bm{p}+i\bm{q})|>\delta$\;
      Sample $\bm{p}\sim\mathcal{N}(\bm{0},j^2 I),\bm{q}\sim\mathcal{N}(\bm{0},(t^2+1) I)$;\Comment{[$j^2$ vs. $t^2+1$ Test.]}\\
      \textbf{Reject} if $|\sum_{i=0}^{d+1}\alpha_i\cdot f(\bm{p}+i\bm{q})|>\delta$\;
      }
      Sample $\bm{p},\bm{q}\sim\mathcal{N}(\bm{0},j^2 I)$;\Comment{[$j^2$ vs. $j^2$ Test.]}\\
      \textbf{Reject} if $|\sum_{i=0}^{d+1}\alpha_i\cdot f(\bm{p}+i\bm{q})|>\delta$\;
    }}
    \textbf{Accept}\;
    }

\Procedure{\emph{\Call{ApproxQuery-$g$}{$\bm{p}$}}}{
\If{$\bm{p}\in \ball(\bm{0},\srad)$}{
\textbf{return} \emph{\Call{ApproxQuery-$g$-InBall}{$\bm{p}$}}\;}

\For{$i\in \{0,1,\hdots,d\}$}{
 $c_i\gets \frac{\srad}{\Vert \bm{p}\Vert_{2}} \cos\big(\frac{\pi(i+1/2)}{d+1}\big)$\;
 $v({c_i})\gets$ \emph{\Call{ApproxQuery-$g$-InBall}{$c_i\bm p$}}  \;
}
Let $p_{\bm p} \colon \mathbb{R}\to\mathbb{R}$ be the unique degree-$d$ polynomial such that $p_{\bm p}(c_i) = v(c_i)$  for $i \in \{0,\hdots,d\}$\;
\textbf{return} $p_{\bm p}(1)$\;
}

\Procedure{\emph{\Call{ApproxQuery-$g$-InBall}{$\bm{p}$}}}{
$N'_{\ref{alg:subroutines_approx}} \gets O(\log\frac{1}{\varepsilon})$\;
Sample $\bm q_1,\dots,\bm q_{N'_{\ref{alg:subroutines_approx}}}\sim \mathcal{N}(\bm 0, I)$\;
\textbf{Reject} if there exists $j\in\set{2,\dots,N'_{\ref{alg:subroutines_approx}}}$ such that $ |\sum_{i=1}^{d+1} \alpha_i \cdot f(\bm{p}+i\bm{q}_1)- \sum_{i=1}^{d+1} \alpha_i \cdot f(\bm{p}+i\bm{q}_j)|>2^{d+2}\delta$\;
\textbf{return} $ \sum_{i=1}^{d+1} \alpha_i \cdot f(\bm{p}+i\bm{q}_1)$\;
}
\end{algorithm}

\paragraph{Bridging the gap between Median and Majority.} The following lemma will allow us to port the techniques that we used in \autoref{sec:exact}, where $g$ was defined as a \emph{majority} over the standard gaussian, to our setting where $g$ is defined as a median. This lemma gives sufficient conditions for the median of any distribution to be close to a random element.

\begin{restatable}{lem}{medIsClose}
\label{thm:median-is-close}
Let $\Omega$ be a sample space, $g:\Omega\to\R$ and $\caD$ be a distribution over $\Omega$. For any $\eta \in [0, 1/4]$, $\delta \in \mathbb{R}$, if
 $\Pr_{\bm{q}_1,\bm{q}_2\sim\caD}[|g(\bm q_1)-g(\bm q_2)|<\delta]>1-\eta,$ then $\Pr_{\bm{q}_1\sim\caD}[|g_{\med}-g(\bm q_1)|<\delta]>1-4\eta$, where $g_{\med}=\mathsf{med}_{\bm q\sim\caD}\{g(\bm q)\}$.
\end{restatable}
The proof is given in \autoref{sec:appendixB}.

\subsection{Preliminaries on Chebyshev Polynomials}\label{sec:chebyshev}

Our proof will heavily rely on properties of the Chebyshev polynomials (of the first kind), which we recall next; further details on Chebyshev polynomials can be found in \cite{mason2002chebyshev}. Denote by $T_d(x)$, the $d$-th Chebyshev polynomial. $T_d$ is a degree-$d$ polynomial and has $d$ roots $\hat c_i \triangleq \cos(\pi(i+1/2)/d)$ for $i \in \{0,\ldots, d-1\}$ in the interval $[-1,1]$, known as \emph{Chebyshev nodes}. 
On the interval $[-1,1]$, the extrema of the Chebyshev polynomials are either $-1$ or $1$, and thus we have
\begin{equation}\label{eq:chebyshev_polys_bound}
    x\in[-1,1]\implies |T_d(x)|\leq 1.
\end{equation}
Chebyshev polynomials form a basis of polynomials, and in particular satisfy the following orthogonality properties.

\paragraph{Orthogonality.}
The polynomials $T_d$ are orthogonal with respect to the weight function $w(x) \triangleq (1-x^2)^{-1/2}$
on the interval $[-1,1]$. Formally, 
\begin{equation}\label{eq:chebyshev_orthogonality_property}
    \int _{-1}^{1}T_n(x)\,T_m(x)\,{\frac {\mathrm {d} x}{\sqrt {1-x^{2}}}}= \begin{cases} 0 &\mbox{if } n \neq m, \\ \pi &\mbox{if } n = m =0, \\ 
\pi/2 &\mbox{if } n=m \neq 0.
\end{cases}
\end{equation}

\paragraph{Discrete orthogonality.} The polynomials $T_d$ are also discretely orthogonal:
\begin{equation}\label{eq:Chebyshev_poly_discrete_orthogonality}
    \sum _{k=0}^{d}{T_{i}(\hat c_{k})\,T_{j}(\hat c_{k})}={\begin{cases}0&~{\text{ if }}~i\neq j,\\d+1&~{\text{ if }}~i=j=0,\\{\frac {d+1}{2}}&~{\text{ if }}~i=j\neq 0,\end{cases}}
\end{equation}
where $d\geq \max(i, j)$, and the $\hat c_k$ are the $d+1$ Chebyshev nodes of $T_{d+1}$.

The following lemma will be useful throughout our proof, and follows in a straightforward fashion from properties of Chebyshev polynomials. 

\begin{lem}\label{lem:stability_univariate_interpolation}
Let $f\colon\R\to \R$ be a univariate polynomial of degree at most $d$. And let $ \hat c_0,\ldots, \hat c_d$ be the Chebyshev nodes of $T_{d+1}$. If $|f(\hat c_k)|\leq \varepsilon$ for every $k\in \{0,\ldots, d\}$, then for every $x\in[-1,1]$ 
\[|f(x)|\leq \sqrt{2}(d+1)\varepsilon .\]

\end{lem}
\begin{proof}
Let $f(x)=\sum_{i=0}^d\alpha_iT_i(x)$ be the Chebyshev expansion of $f$. Since for every $k$, we have $f(\hat c_k)^2\leq\varepsilon^2$, 
\[\sum_{k=0}^{d}f(\hat c_k)^2\leq (d+1)\varepsilon^2. \]
On the other hand, 
\begin{align*}
    \sum_{k=0}^{d}f(\hat c_k)^2&=\sum_{k=0}^{d}\left(\sum_{i=0}^d\alpha_iT_i(\hat c_k)\right)^2\\
    &=\sum_{k=0}^{d}\sum_{i,j=0}^d\alpha_i\alpha_jT_i(\hat c_k)T_j(\hat c_k)\\
    &=\sum_{i,j=0}^d\alpha_i\alpha_j\sum_{k=0}^{d}T_i(\hat c_k)T_j(\hat c_k)\\
    &\geq \frac{d+1}{2}\sum_{i=0}^d\alpha_i^2. \tag{By (\ref{eq:Chebyshev_poly_discrete_orthogonality})}
\end{align*}

Combining the above bounds, we have that $|\alpha_i|\leq \sqrt{2}\varepsilon$ for every $i\in\{0,\dots,d\}$. Thus, for every $x\in[-1,1]$, 
\begin{align*}
    |f(x)|=\bigg|\sum_{i=0}^d\alpha_iT_i(x)\bigg|
    \leq \sum_{i=0}^d|\alpha_i\|T_i(x)|
    \leq (d+1)\sqrt{2}\varepsilon. \tag{by (\ref{eq:chebyshev_polys_bound})}
\end{align*}

\end{proof}

By scaling the Chebyshev nodes, we can obtain the following corollary, which is a scaled version of \autoref{lem:stability_univariate_interpolation} to any given interval, rather than $[-1,1]$. 

\begin{cor}\label{lem:stability_univariate_interpolation_scaled}
Let $f$ be a univariate degree-$d$ polynomial, let $m\in \mathbb{R}_{>0}$, and  $c_{k} \triangleq m\cos\big({\frac{\pi }{d+1}(k+1/2)}\big)$ for $k\in \{0,\ldots, d\}$ be the Chebyshev nodes of $T_{d+1}$ scaled to the interval $[-m,m]$.
If $|f(c_k)|\leq \varepsilon$ for every $k\in \{0,\ldots, d\}$, then for any $x\in[-m,m]$,
\[|f(x)|\leq \sqrt{2}(d+1)\varepsilon .\]
\end{cor}
\begin{proof}
In the proof of \autoref{lem:stability_univariate_interpolation} we represent $f(x)=\sum_{i=0}^d\alpha_i T_i(x/m)$ as a linear combination of the Chebyshev polynomials with the back-scaled variable. The other parts of the proof are the same.
\end{proof}

\subsection{Correctness of the Approximate Polynomial Tester}

In the remainder of this section we will argue the correctness of our tester (\autoref{thm:approximate_low_degree_tester}). The next lemma records the properties of $g$ that it guarantees. 

\begin{lem}\label{lem:main_approx_poly}
    Let $\srad=(4d)^{-6}$, $\delta=2^{d+1}\alpha$, as  set in \autoref{alg:subroutines_approx}, and $\lrad>\srad$. 
    If \Call{ApproxCharacterizationTest}{} fails with probability at most $2/3$, then $g$ is pointwise $2^{(2n)^{45d}}\lrad^d\delta$-close to a degree-$d$ polynomial in $\ball(\bm 0,\lrad)$.
    Furthermore, for every point $\bm p \in \ball(\bm 0,\lrad)$ \Call{ApproxQuery-$g$}{$\bm{p}$} well approximates $g(\bm p)$ with high probability, that is, 
    \[\Pr\left[|g(\bm{p})-\text{\Call{ApproxQuery-$g$}{$\bm{p}$}}|\leq(12\lrad/\srad)^d2^{d+4}\delta\right]\geq 1-\frac{\eps}{4}.\]
\end{lem}

We prove the main theorem of this section assuming that \autoref{lem:main_approx_poly} holds. 

\begin{proof}[Proof of \autoref{thm:approximate_low_degree_tester}]
  If $f$ is point-wise $\alpha$-close to a degree-$d$ polynomial $h$, then for any $\bm p,\bm q\in\R^n$,
\[\left|\sum_{i=0}^{d+1}\alpha_i\cdot f(\bm p+i\bm q)\right|=\left|\sum_{i=0}^{d+1}\alpha_i\cdot( f(\bm p+i\bm q)-h(\bm p+i\bm q))\right|\leq\sum_{i=0}^{d+1}|\alpha_i|\cdot\alpha\leq 2^{d+1}\alpha=\delta.\]
Thus, \Call{ApproxCharacterizationTest}{} always passes, and \Call{ApproxQuery-$g$}{$\bm p$} returns a value that is $2^{d+3}\delta$-close to $g(\bm p)$, without rejecting  with probability $1$, and  \autoref{alg:low_degree_main_algorithm_approx} always accepts. To see this observe, for any $\bm p,\bm q_1,\bm q_j\in\R^n$,
\begin{align*}
    |g_{\bm q_1}(\bm p)-g_{\bm q_j}(\bm p)|&=\Big|\sum_{i=1}^{d+1}\alpha_i f(\bm p+i\bm q_1)-\sum_{i=1}^{d+1}\alpha_i f(\bm p+i\bm q_j)\Big|=\Big|\sum_{i=0}^{d+1}\alpha_i f(\bm p+i\bm q_1)-\sum_{i=0}^{d+1}\alpha_i f(\bm p+i\bm q_j)\Big|\\
    &=\Big|\sum_{i=0}^{d+1}\alpha_i(f(\bm p+i\bm q_1)-h(\bm p+i\bm q_1))-\sum_{i=0}^{d+1}\alpha_i(f(\bm p+i\bm q_j)-h(\bm p+i\bm q_j))\Big|\\
    &\leq 2\sum_{i=0}^{d+1}|\alpha_i(f(\bm p+i\bm q_1)-h(\bm p+i\bm q_1))|\leq 2\sum_{i=0}^{d+1}|\alpha_i|\cdot\alpha=2^{d+2}\alpha<2^{d+2}\delta.
\end{align*}
So, by \autoref{thm:median-is-close}, we may claim $\Pr_{\bm q_1\sim\cN(\bm 0,I)}[|g(\bm p)-\text{\Call{ApproxQuery-$g$}{}}(\bm p)|<2^{d+2}\delta]=1$, where\\ \Call{ApproxQuery-$g$}{}$(\bm p)\triangleq g_{\bm q_1}(\bm p)$, by \Call{ApproxQuery-$g$-InBall}{$\bm p$}, and $g(\bm p)=\mathsf{med}_{\bm q\sim\cN(\bm 0,I)}\{g_{\bm q}(\bm p)\}$.

  Next, we show that if $f$ is $\beta$-far from all degree-$d$ polynomials, for $\beta\triangleq 2\cdot 2^{(2n)^{45d}}\lrad^d\delta$, then Algorithm~\ref{alg:low_degree_main_algorithm_approx} rejects with probability at least $2/3$. Let $\delta_1\triangleq 2^{(2n)^{45d}}\lrad^d\delta$, and $\delta_2\triangleq(12\lrad/\srad)^d2^{d+4}\delta$.
  If \Call{ApproxCharacterizationTest}{} fails with probability at least $2/3$, then we reject $f$ with probability at least $2/3$. 
  Otherwise, by \autoref{lem:main_approx_poly}, $g$ is  pointwise $\delta_1$-close in $\ball(\bm 0,\lrad)$ to some degree-$d$ polynomial $H$, 
  and for every $\bm p\in\ball(\bm 0,\lrad),\Pr[|g(\bm{p})-\text{\Call{ApproxQuery-$g$}{$\bm{p}$}}|>\delta_2]<\frac{\eps}{4}$. Hence,
$\Pr_{\bm p \sim \mathcal{D}}[|f(\bm p) - g(\bm p)|>\beta-\delta_1] > \varepsilon$, noting $\beta-\delta_1\geq\delta_1$. 

The probability that we do not reject in any of the $N_{\ref{alg:low_degree_main_algorithm_approx}}$ steps of \autoref{alg:low_degree_main_algorithm_approx} is at most the probability that either $\bm p\not\in\ball(\bm 0,\srad)$, for every sampled point $\bm p$, or $|f(\bm{p})-g(\bm{p})|\leq\delta_1$, 
or that \emph{\Call{ApproxQuery-$g$}{$\bm{p}$}} returned a value that is $\delta_2$-far from $g(\bm{p})$ (instead of rejecting). The first event happens with probability at most $\frac{\eps}{4}$, while the last happens with probability at most $\frac{\varepsilon}{4}$ by \autoref{lem:main_approx_poly}. Thus,
\[\Pr_{\bm{p}\sim \mathcal{D}}[\bm p\not\in\ball(\bm 0,\lrad)\lor |f(\bm{p})-g(\bm{p})|\leq\delta_1 \lor |g(\bm{p}) - \Call{Query-\mbox{$g$}}{\bm{p}}|>\delta_2]\leq\frac{\eps}{4} +1-\varepsilon +\frac{\varepsilon}{4}\leq 1-\frac{\varepsilon}{2},\]
and  \autoref{alg:low_degree_main_algorithm_approx} accepts with probability at most ${\left(1-\frac{\varepsilon}{2}\right)}^{N_{\ref{alg:low_degree_main_algorithm_approx}}}< \frac{1}{3}$ for sufficiently large $N_{\ref{alg:low_degree_main_algorithm_approx}}=O(1/\eps)$.

Finally, the bound on the query complexity of the tester follows the same argument, as in the exact case, for \autoref{alg:low_degree_main_algorithm}, noting that for sampled points $\bm p\sim\caD$ that don't fall in $\ball(\bm 0,\lrad)$, \Call{LowDegreeApproxTester}{} makes no queries to $f$, and thus matches the $O(d^5+\frac{d^2}{\eps}\log(\frac{1}{\eps}))$ query complexity of the \Call{LowDegreeTester}{}. 
\end{proof}

In the remainder of this section, we will prove \autoref{lem:main_approx_poly}. This will be done in three steps, similar to the proof outline for \autoref{lem:main_lemma_well_definedness_and_querying_g}. First, we show that $g$ is pointwise close a univariate polynomial of degree $d$ on every line segment in $\ball(\bm 0, \srad)$. 
Then, we show that $g$ is pointwise close to a degree-$d,n$-variate polynomial within $\ball(\bm 0, \srad)$. 
Finally, by the fact that $g$ is defined by interpolating evaluations out of $\ball(\bm 0,\srad)$, we show that it is pointwise close to a degree-$d,n$-variate polynomial on $\mathbb{R}^n$.

\subsection{Polynomial Approximation on Every Line Within the Ball}\label{sec:poly-on-lines-ball-approx}

First, we will argue that $g$ is approximately consistent with a degree $d$ polynomial on every line within the ball $\ball(\bm 0, r)$. The following is an approximate analogue of \autoref{lem:g-on-lines-in-B-is-poly}.  

\begin{thm}\label{lem:g-on-lines-in-B-is-poly_approx}
If $\Call{ApproxCharacterizationTest}{}$ fails with probability at most $2/3$, and $f$ is bounded on $\ball(\bm 0,2d\sqrt{n})$,
then for every $\bm{a},\bm{b}\in \ball(\bm 0,\srad)$, 
the univariate function $g_{\bm a, \bm b}(x) = g(\bm a + x \bm b)$ 
defined on points $x \in \pline_{\bm a,\bm b}^{\ball}$ is pointwise $2^{15d^2}\cdot\delta$-close to a degree-$d$, univariate polynomial. 
\end{thm}

The main technical tool in the proof of this theorem will be the following corollary of a result\footnote{Stated in \autoref{sec:pushed-proofs} as \autoref{thm:approx-characterization}.} from \cite{Gaj91}, which guarantees that any bounded function $f$ defined on a line segment, which has small $(d+1)$-st order finite forward differences, is point-wise close to a degree-$d$ polynomial, on that line segment.


\begin{thm}\label{cor:gajda-cor}
Let $x_0\in\R,d\in\mathbb N,\phi,a\in(0,\infty)$, and a bounded function $f:(x_0-a,x_0+a)\to\R$, such that for all $x\in(x_0-a,x_0+a)$, and $h\in(-a,a)$, with $x+(d+1)h\in(x_0-a,x_0+a),|\Delta_h^{(d+1)}[f](x)|\leq\phi$. Then, there exists a degree-$d$ polynomial $g:\R\to\R$, such that for every $x\in(x_0-a,x_0+a),|f(x)-g(x)|\leq 2^{8d^2}\phi$.
\end{thm}

Thus, in order to prove  an approximate analogue of \autoref{lem:g-on-lines-in-B-is-poly},
it suffices to show that the self-corrected function $g$ satisfies the conditions of \autoref{cor:gajda-cor}; i.e., along every line the $(d+1)$st order finite differences of the restriction of $g$ to these lines are small, which will occupy the remainder of this subsection.

Let $\rho$ denote the bound of the  probability that each of the tests in the  \Call{ApproxCharacterizationTest}{} fails.
That is, for every $j\in \{1,\dots,d+1\}$ and $t\in \{0,\dots,d+1\}$: 

  \begin{align}
      \Pr_{\substack{\bm{p}\sim\mathcal{N}(\bm{0},j^2(t^2+1) I)\\\bm{q}\sim\mathcal{N}(\bm{0}, I)}}\bigg[\Big|\sum_{i=0}^{d+1}\alpha_i\cdot f(\bm{p}+i\bm{q})\Big|>\delta\bigg]&\leq \rho. && \text{[$j^2(t^2+1)$ vs. $1$ Test.]} \label{eq:test_bound_1_approx}\\
      \Pr_{\substack{\bm{p}\sim\mathcal{N}(\bm{0},j^2 I)\\\bm{q}\sim\mathcal{N}(\bm{0},(t^2+1) I)}}\bigg[\Big|\sum_{i=0}^{d+1}\alpha_i\cdot f(\bm{p}+i\bm{q})\Big|>\delta\bigg]&\leq \rho. && \text{[$j^2$ vs. $t^2+1$ Test.]}\label{eq:test_bound_2_approx}\\
      \Pr_{\substack{\bm{p}\sim\mathcal{N}(\bm{0},j^2 I)\\\bm{q}\sim\mathcal{N}(\bm{0},j^2 I)}}\bigg[\Big|\sum_{i=0}^{d+1}\alpha_i\cdot f(\bm{p}+i\bm{q})\Big|>\delta\bigg]&\leq \rho. && \text{[$j^2$ vs. $j^2$ Test.]}\label{eq:test_bound_3_approx}
  \end{align}
Following the same argument as in \autoref{claim:rho-is-small}, we first bound $\rho$:
\begin{claim}\label{claim:rho-is-small-approx}
If \Call{ApproxChacterizationTest}{} fails with probability at most $2/3$, then $\rho\text{ is at most }(30d)^{-2}$.
\end{claim}
Then, we prove an approximate version of \autoref{lem:well_defined_of_g} (which lower bounded collision probabilities), via an identical argument, the proof of which can be found in \autoref{sec:pushed-proofs}:

\begin{restatable}{lem}{gwelldefapprox}\label{lem:well_defined_of_g_approx}
   For every $\bm{p}\in \ball(\bm{0},\srad)$, and every $t\in  \{0,\ldots, d+1\}$, $$\Pr_{\substack{\bm{q}_1\sim\mathcal{N}(\bm{0},(t^2+1)I)\\ \bm{q}_2\sim\mathcal{N}(\bm{0},I)}}\left[|g_{\bm{q}_1}(\bm{p})- g_{\bm{q}_2}(\bm{p})|>2^{d+2}\delta\right]\leq 4d\rho+48d^5\srad. $$ 
\end{restatable}

An immediate corollary is the following. 

\begin{cor}\label{cor:replacement-is-good_approx}
    If $\Call{ApproxCharacterizationTest}{}$ fails with probability at most $2/3$, then for every $\bm p \in \ball(\bm 0,\srad)$ and every  $t\in \{0,\ldots, d+1\}$,
    \[ \Pr_{\bm{q}\sim\mathcal{N}(\bm{0},(t^2+1)I)}[|g(\bm p)- g_{\bm{q}}(\bm{p})|>2^{d+3}\delta] < \frac{1}{7d}. \]
\end{cor}
\begin{proof}
By \autoref{claim:rho-is-small-approx}, $\rho$ at at most $(30d)^{-2}$. Observe that for any $t \in \{0,\ldots, d+1\}$,
    \begin{align*}
        &\Pr_{\bm{q}\sim\mathcal{N}(\bm{0},(t^2+1)I)}[|g(\bm p)- g_{\bm{q}}(\bm{p})|>2^{d+3}\delta]\\
        &\leq\Pr_{\bm{q}_1\sim\mathcal{N}(\bm{0},I)}[|g(\bm{p})- g_{\bm{q}_1}(\bm{p})|>2^{d+2}\delta]+\Pr_{\substack{\bm{q}\sim\mathcal{N}(\bm{0},(t^2+1)I)\\ \bm{q}_1\sim\mathcal{N}(\bm{0},I)}}[|g_{\bm{q}_1}(\bm p)- g_{\bm{q}}(\bm{p})|>2^{d+2}\delta]\\
        &\leq 5(4d\rho+48d^5\srad)\tag{By \autoref{thm:median-is-close}, $\because g(\bm p)=\mathsf{med}_{\bm q\sim\cN(\bm 0,I)}\{g_{\bm q}(\bm p)\}$, and \autoref{lem:well_defined_of_g_approx}}.
    \end{align*}
    By choosing $\srad =(4d)^{-6}$ and with $\rho \leq (30d)^{-2}$, we get $5(4d\rho+48d^5\srad)\leq 1/(7d)$.
\end{proof}

Next, we prove the approximate analogue of \autoref{thm:characterization_in_small_ball}, (which showed that the $(d+1)$st order finite differences of $g$'s restrictions to all lines in $\ball(\bm 0,\srad)$ vanish) via an identical argument, and the proof of which can also be found in \autoref{sec:pushed-proofs}. 

\begin{restatable}{lem}{smallballcharapprox}\label{thm:characterization_in_small_ball_approx}
    If $\Call{ApproxCharacterizationTest}{}$ fails with probability at most $2/3$, then for every
    $\bm{p},\bm{q}\in \ball(\bm{0},\srad)$ and sufficiently small $h\in\mathbb{R}$, such that $\bm p+ih\bm q\in\ball(\bm 0,\srad)$ for every $i\in[d+1]$, we have $|\sum_{i=0}^{d+1}\alpha_i\cdot g(\bm{p}+ih\bm{q})|\leq 2^{2d+5}\delta$.
\end{restatable}

We are now ready to prove \autoref{lem:g-on-lines-in-B-is-poly_approx}. 

\begin{proof}[Proof of \autoref{lem:g-on-lines-in-B-is-poly_approx}] First note that since $f$ is bounded on $\ball(\bm 0,2d\sqrt{n})$, by the same argument as in \autoref{lem:g-is-bounded}, $g$ is bounded on $\ball(\bm 0,\srad)$.
Next, fix some $\bm{a},\bm{b}\in \ball(\bm{0},\srad)$;
we would like to show that $g_{\bm a, \bm b}(x)$ is pointwise close a unique degree-$d$ polynomial for every point $x$ in $\{x \in \mathbb{R} : \bm a + x \bm b \in \ball(\bm 0, \srad) \}$; fix such an $x$. By \autoref{cor:gajda-cor}, it suffices to show that for all sufficiently small $h\in\R$, such that $\bm a+(x+ih)\bm b\in\pline_{\bm a,\bm b}^{\ball}$ for every $i\in[d+1]$,
\[  |\Delta^{(d+1)}_h[g_{\bm a, \bm b}](x)| = \Big|\sum_{i=0}^{d+1} \alpha_i \cdot g(\bm a + x \bm b + i h \bm b)\Big|\leq 2^{7d^2}\delta. \]
By \autoref{thm:characterization_in_small_ball_approx}, we have that for every $\bm{p},\bm{q}\in \ball(\bm{0},\srad)$ and sufficiently small $h\in\mathbb{R}$, such that $\bm p+ih\bm q\in\ball(\bm 0,\srad)$ for every $i\in[d+1]$, $|\sum_{i=0}^{d+1}\alpha_i\cdot g(\bm{p}+ih\bm{q})|\leq 2^{2d+5}\delta$. Let $\bm p \triangleq \bm a + x \bm b$ and $\bm q \triangleq \bm b$. Since $\ball(\bm 0, r)$ is an open ball containing $\bm p\text{, and }\bm q$, we have $\bm p + ih\bm q \in \ball(\bm 0, \srad)$ for every $i\in[d+1]$. Thus,
\[  |\Delta^{(d+1)}_h[g_{\bm a, \bm b}](x)| =  \Big|\sum_{i=0}^{d+1} \alpha_i \cdot g(\bm a + x \bm b + i h \bm b)\Big| =  \Big| \sum_{i=0}^{d+1} g(\bm p + ih\bm q)\Big| \leq 2^{2d+5}\delta\leq 2^{7d^2}\delta. \qedhere \]
\end{proof}

\subsection{Polynomial Approximation Within the Hypercube} 
Let $0<m\leq 1$ be a large value such that the hypercube $[-m,m]^n$ is contained within $\ball(\bm 0, \srad)$; setting $m=r/(2\sqrt{n})$ suffices. We will prove that the self-corrected function $g$ is close to a degree-$d$ polynomial on $[-m,m]^n$. 
The following lemma is the approximate analogue of \autoref{lem:local_to_global}, and \autoref{lem:radial_lines_consistent_force_low_degree} combined into one. 

\begin{lem}\label{lem:local_to_global_approx}
    Let $m\in(0,1],\delta>0$, and let $h\colon [-m,m]^n \to \mathbb{R}$. 
    If for every line $\pline$, the restriction of $h$ to this line  $h_{\pline}$, is pointwise $\delta$-close to a degree-$d$ polynomial $\hat{h}_{\pline}$,  
    then $h$ is pointwise $((2/m)^{n^{40d}} \delta )$-close to a  degree-$d,n$-variate polynomial. 
\end{lem}
The  proof of \autoref{lem:local_to_global_approx} is by induction. At each inductive step we build a degree-$2d$ polynomial and then reduce it to degree $d$ using the following Lemma, the proof of which is in the of which is deferred until the following subsection. 

\multivariateDegreeReduction*

\begin{proof}[Proof of \autoref{lem:local_to_global_approx}]
We will show that $h$ is pointwise close to an $n$-variate  degree-$d$ polynomial $H_n$ by induction on the dimension $n$.
Set $\delta_n \triangleq (2/m)^{n^{40d}} \delta$.
For the base case, when $n=1$, we have that $h=h_{ 0,  1}$  is  pointwise $\delta$-close to a univariate polynomial $\hat{h}_{ 0,  1}$ of degree $d$ by assumption, so we let $H_1=\hat{h}_{  0,  1}$ and $\delta_1=(2/m)\delta\geq \delta$.

Assume that the statement is true for $n-1$, with $\delta_{n-1} =(2/m)^{(n-1)^{40d}} \delta $. 
For any $c\in[-m,m]$, define $h^{(c)}\colon[-m,m]^{n-1}\to\mathbb{R}$ as 
\[h^{(c)}(x_1,\dots,x_{n-1}) \triangleq h(x_1,\dots,x_{n-1},c).\]
We will argue that $h^{(c)}$ is pointwise $\delta_{n-1}$-close to an $(n-1)$-variate polynomial of total degree at most $d$.
Fix $i\in [n-1]$, and  $\bm{a}\in[-m,m]^{n-1}$ with $a_i=0$, and let $\bm a^\uparrow=(a_1,\dots,a_{n-1},c) \in [-m,m]^n$ to be an extension of $\bm a$ to dimension $n$. As well, let $\bm e_i$ denote the $i$th standard basis vector. 
By assumption, $h_{\bm a^\uparrow, \bm e_i}\colon[-m,m]\to\mathbb{R}$ is  pointwise $\delta$-close to some univariate degree-$d$ polynomial, which we will denote by $\hat{h}_{\bm a^\uparrow, \bm e_i}$. For every $x\in[-m,m]$, we have
\[h_{\bm a^\uparrow, \bm e_i}(x)=h(a_1,\dots,a_{i-1},x,a_{i+1},\dots,a_{n-1},c)=h^{(c)}(\bm{a}+x\bm{e_i})=h^{(c)}_{\bm{a,e_i}}(x),\]
and so $h^{(c)}_{\bm{a,e_i}}(x)$ is $\delta$-close to 
$\hat{h}_{\bm a^\uparrow, \bm e_i}$ on $[-m,m]$. 
Thus, by the induction hypothesis, $h^{(c)}\colon[-m,m]^{n-1}\to\mathbb{R}$ is pointwise $\delta_{n-1}$-close to an $(n-1)$-variate polynomial of total degree at most $d$, which we will denote by $H^{(c)}_{n-1}$.

It remains to show that $h$ is  pointwise $\delta_n$-close to an $n$-variate polynomial $H_{n}$ of total degree at most $d$ on $[-m,m]^n$. Let $c_0,\dots, c_{d} \in[-m,m]$ be the scaled Chebyshev nodes $c_i \triangleq m\cos{\left(\frac{\pi}{2}(i+1/2)/(d+1) \right)}$. 
Let $\delta_{c_i}\colon\R\to\R$ be the unique degree-$d$ polynomial which satisfies \[\delta_{c_i}(c_j)=\begin{cases}1 & i=j,\\0 &i\neq j.\end{cases}\]
Using $\delta_{c_i}$, we build a degree at most $2d$ polynomial
\[H(x_1,\dots,x_n) \triangleq \sum_{i=0}^{d}\underbrace{\delta_{c_i}(x_n)}_{\text{degree } d}\underbrace{H^{(c_i)}_{n-1}(x_1,\dots,x_{n-1})}_{\text{degree  }\leq d}.\]
Next, we argue that  $H$ is pointwise close to $h$. Fix some $\bm{b}=(b_1,\ldots, b_{n-1}) \in [-m,m]^{n-1}$ and let $\bm b^\uparrow=(b_1,\dots,b_{n-1},0)$ be an extension of $\bm b$ to dimension $n$. Consider the following two univariate functions in the variable $t$. The first function is 
  \[h_{\bm b^{\uparrow},\bm{e_n}}(t)=h(b_1,\dots,b_{n-1},t),\]
  which by assumption is pointwise $\delta$-close to a univariate degree-$d$  polynomial $\hat{h}_{\bm b^{\uparrow},\bm{e_n}}(t)$.
  The second is the polynomial $H$ with the first $n-1$ variables fixed to $\bm b$, 
  \[H(b_1,\dots,b_{n-1},t)= \sum_{i=0}^{d}\delta_{c_i}(t)H^{(c_i)}_{n-1}(b_1,\dots,b_{n-1}).\]
  Since the $H^{(c_i)}_{n-1}(b_1,\dots,b_{n-1})$ are constants in $t$,  $H_n(b_1,\dots,b_{n-1},t)$ is a univariate polynomial of degree $d$. 
  
  Observe that for $c_0,\dots c_d$,
  \begin{align*}
      |\hat{h}_{\bm b^{\uparrow},\bm{e_n} }(c_i)-H(b_1,\dots,b_{n-1},c_i)| &\leq |\hat{h}_{\bm b^{\uparrow},\bm{e_n} }(c_i)-h_{\bm b^{\uparrow},\bm{e_n} }(c_i)|+|h^{(c_i)}(\bm b)-H(b_1,\dots,b_{n-1},c_i)|\\
      &\leq \delta+\delta_{n-1},
  \end{align*}
  where the first inequality follows because $h_{\bm b^{\uparrow},\bm{e_n} }(c_i)=h^{(c_i)}(\bm b)$ and the second follows by the inductive hypothesis, since $H(b_1,\dots,b_{n-1},c_i)=H^{(c_i)}_{n-1}(b_1,\dots,b_n)$ by definition.
  
  Applying \autoref{lem:stability_univariate_interpolation_scaled} to the error function $e(t)=\hat{h}_{\bm b^{\uparrow},\bm{e_n} }(t)-H(b_1,\dots,b_{n-1},t)$,  we have that for every $t\in[-m,m]$, the difference between the two degree-$d$ polynomials is at most
  \[|\hat{h}_{\bm b^{\uparrow},\bm{e_n} }(t)-H(b_1,\dots,b_{n-1},t)|\leq \sqrt{2}(d+1)(\delta+\delta_{n-1}).\]
  Since this is true for every $\bm b\in[-m,m]^{n-1}$ and $t\in[-m,m]$, we have that for every $\bm x\in[-m,m]^n$ 
  \[|H(\bm x)-h(\bm x)|\leq |H(\bm x)-\hat{h}_{(x_1,\dots,x_{n-1}),\bm e_n}(x_n)|+|\hat{h}_{(x_1,\dots,x_{n-1}),\bm e_n}(x_n)-h(\bm x)|\leq  \sqrt{2}(d+1)(\delta+\delta_{n-1})+\delta.\]
  Note that for every  $\bm a\in[-m,m]^n$,  the restriction $H_{\bm 0, \bm a}$ on the radial line $\pline_{\bm 0,\bm a}$ is a univariate polynomial which is pointwise $(10d\delta_{n-1})$-close to the degree $d$ univariate polynomial $\hat{h}_{\bm 0,\bm a}$ on points in the cube $[-m,m]^n$, since $|H_{\bm 0, \bm a}(t)-\hat{h}_{\bm 0,\bm a}(t)|\leq |H_{\bm 0, \bm a}(t)-{h}_{\bm 0,\bm a}(t)|+|h_{\bm 0, \bm a}(t)-\hat{h}_{\bm 0,\bm a}(t)| \leq \sqrt{2}(d+1)(\delta+\delta_{n-1})+2\delta$ for every $t\in [-1,1]$.
  
  Applying \autoref{thm:multivarite_degree_reduction} on $H$, we have that $H_n\triangleq H^{\leq d}$ is pointwise $(20d (2/m)^{2n^{36d}}\delta_{n-1})$-close to $H$ on the cube $[-m,m]^n$. Thus, for every $\bm x\in[-m,m]^n$,  we have
  \begin{align*}
      |h(\bm x)-H_n(\bm x)|&\leq |h(\bm x)-H(\bm x)|+|H(\bm x)-H_n(\bm x)|\\
      &\leq (9d\delta_{n-1})+(20d (2/m)^{2n^{36d}}\delta_{n-1})\\
      &\leq 30d (2/m)^{2n^{36d}}\delta_{n-1}\\
      &\leq \delta_n. \qedhere
  \end{align*}
 \end{proof}

\subsubsection{Proof of \autoref{thm:multivarite_degree_reduction}}
Consider the monic Chebyshev polynomials $\widetilde{T}_n(x)\triangleq 2^{1-n}T_n(x)$, with $\|\widetilde{T}_n(x)\|_\infty=2^{1-n}$ on the interval $x\in[-1,1]$. Then, by the extremal property that Chebyshev polynomials have the minimum maximal absolute value among all monic polynomials of the same degree on the interval $[-1,1]$, we have the following fact and the subsequent lemma.

\begin{fact}\label{fact:monic-lower-bound}
For every monic polynomial $p(t)$ of degree $d\geq 1$ there exists $x\in[-1,1]$ such that $|p(x)|\geq 2^{1-d}$.
\end{fact}

\begin{cor}\label{cor:monic-scaled-lower-bound}
    Let $m \leq 1$ 
      and $p(t)$ is a monic polynomial of degree $d\geq 1$. Then, there exists $x\in[-m,m]$ such that $|p(x)|\geq m^d2^{1-d}$.
\end{cor}
\begin{proof}
    Let $p(x)=x^d+\sum_{i=0}^{d-1}\alpha_ix^i$, and note that $p(xm) = (xm)^d + +\sum_{i=0}^{d-1}\alpha_i(xm)^{i}$. Then $p(xm)/m^d$ is a degree-$d$ \emph{monic} polynomial.  Thus, by \autoref{fact:monic-lower-bound}, there exists $x \in [-1,1]$ such that $|p(xm)/m^d| \geq 2^{1-d}$. That is, there is $x \in [-m,m]$ such that $|p(x)| \geq m^d2^{1-d}$.
\end{proof}

\nonmoniclowerbound*
\begin{proof}
Let $\ell$ be the largest index such that $|\alpha_{\ell}| \geq 2^{(-\sum_{k=1}^\ell k)}\eta$; $\ell$ exists since $|\alpha_i|\geq \eta\geq 2^{(-\sum_{k=1}^i k)}\eta$.
Note that by the maximality of $\ell$, for every $j > \ell$, $|\alpha_j| < 2^{(-\sum_{k=1}^j k)} \eta$. Therefore,
\begin{align*}
    \frac{p(x)}{\alpha_{\ell}}=\sum_{j=0}^{\ell}\frac{\alpha_j}{\alpha_{\ell}}x^j +\sum_{j=\ell+1}^d\frac{\alpha_j}{c_{\alpha}}x^j=\frac{p^{\leq \ell}(x)}{\alpha_{\ell}} + \frac{p^{>\ell}(x)}{\alpha_{\ell}}.
\end{align*}
Observe that $p^{\leq \ell}(x)/\alpha_{\ell}$ is a monic polynomial of degree at most $\ell$, and thus by  \autoref{fact:monic-lower-bound} there exists $x\in[-1,1]$ such that $|p^{\leq \ell}(x)|/|\alpha_{\ell}|\geq 2^{1-\ell}$. On the other hand, for every $x \in [-1,1]$ we have that
\[\frac{|p^{>\ell}(x)|}{|\alpha_{\ell}|} \leq \sum_{j=\ell+1}^d \frac{|\alpha_j|}{|\alpha_{\ell}|}\leq\sum_{j=\ell+1}^d 2^{-\left(\sum_{k=\ell+1}^j k\right)} \leq \sum_{j>\ell} 2^{-j} \leq 2^{-\ell}.\]

Altogether this implies that there is some $x\in[-1,1]$ such that
\[\frac{|p(x)|}{|\alpha_{\ell}|}\geq\frac{|p^{\leq \ell}(x)|}{|\alpha_{\ell}|}-\frac{|p^{>\ell}(x)|}{|\alpha_{\ell}|}\geq 2^{1-\ell}-2^{-\ell}=2^{-\ell},\]
and it follows that $|p(x)|\geq 2^{-\ell}|\alpha_{\ell}|\geq 2^{-\ell}(2^{-\sum_{k=1}^\ell k})\eta=(2^{-\ell(\ell+3)/2})\eta\geq \eta 2^{-2d^2}$. 
\end{proof}

\begin{cor}\label{cor:non-monic-lower-bound_in_small_cube}
Fix $\eta >0$, $m<1$, and let  $p(x)=\sum_{i=0}^d\alpha_i x^i$ be a degree-$d$ polynomial. If $|\alpha_i|\geq\eta$ for some $i\in [d]$, then there exists $x\in[-m,m]$ such that $|p(x)|\geq 2^{-2d^2}m^d\eta$.
\end{cor}
\begin{proof}
    Writing $p(x)=p^{\leq\ell}(x)+p^{>\ell}(x)$ as in \autoref{lem:non-monic-lower-bound}, and noticing $p^{\leq\ell}(x)/\alpha_\ell$ is a monic, degree $\ell$ polynomial, we invoke \autoref{cor:monic-scaled-lower-bound} to claim, there exists $x\in[-m,m]$ such that $|p^{\leq\ell}(x)|/|\alpha_\ell|\geq 2^{1-\ell}m^\ell$. While simultaneously, we have for every $x\in[-m,m]$:
    \[\frac{|p^{>\ell}(x)|}{|\alpha_\ell|}\leq\sum_{j=\ell+1}^d\frac{|\alpha_j||m|^j}{|\alpha_\ell|}\leq\sum_{j=\ell+1}^d 2^{-\left(\sum_{k=\ell+1}^d k\right)}m^j\leq\sum_{j>\ell}2^{-j}m^j\leq m^\ell\sum_{j>\ell}2^{-j}\leq 2^{-\ell}m^\ell.\]
    Therefore, there exists $x\in[-m,m]$ such that $|p(x)|/|\alpha_\ell|\geq 2^{-\ell}m^\ell$, and we have $|p(x)|\geq 2^{-\ell}m^\ell|\alpha_\ell|\geq 2^{-2d^2}m^d\eta$.
\end{proof}

For vectors $\bm y, \bm z\in \R^n$, denote by $\langle \bm y, \bm z \rangle=\sum_{j\in[n]}y_jz_j$ the standard inner product between them.

\begin{lem}\label{lem:isolation}
    For $k\geq n^{8d}$ there exists $\bm y \in \{0,\ldots,k\}^n$ such that for any $\bm z^{(1)} \neq \bm z^{(2)} \in \{0,\ldots, d\}^n$ satisfying $\sum_{j=1}^n  z^{(i)}_j \leq 2d$ for $i \in \{1,2\}$, it holds that 
    $\langle \bm y, \bm z^{(1)} \rangle \neq \langle \bm y, \bm z^{(2)} \rangle$.
\end{lem}

\begin{proof}
Let ${\cal Z}=\{\bm z\in \{-d,\dots,0,\dots,d\}^n\colon \|\bm z\|_0\leq 4d\}$, where $\| \cdot \|_0$ gives the number of non-zero coordinates. 
    Note that $\bm z^{(1)} - \bm z^{(2)} \in {\cal Z}$. 
    Thus, it suffices to show that there exists $\bm y$ such that for any $\bm z \in {\cal Z}$, if $\langle \bm y, \bm z \rangle = 0$ then $\bm z = \bm 0$. 
    Suppose $\bm z \neq \bm 0$, and let $\ell$ be such that $z_\ell \neq 0$. 
    Sample $\bm y$ uniformly from $\{0,\ldots, k\}^n$. Then,
    \[ \Pr_{\bm y}[ \langle \bm y, \bm z \rangle = 0] = \Pr \bigg[ y_\ell = - \frac{1}{z_\ell} \sum_{j \neq \ell} y_j z_j  \bigg]\leq \frac{1}{k}.\]
    By a union bound over all $\bm z\neq \bm 0$,
    \[ \Pr_{\bm y}[ \exists \bm z: \langle \bm y, \bm z \rangle=0] \leq \frac{|{\cal Z}|-1}{k}<\frac{{\binom{n}{4d}}(2d+1)^{4d}}{k} \leq \frac{n^{8d}}{k}. \]
    Thus, choosing $k\geq n^{8d}$, there exists $\bm y$ such that for every $\bm 0\neq \bm z\in{\cal Z}$ it holds $\langle \bm y, \bm z \rangle\neq 0$.
\end{proof}
Let us introduce some notation. For an $n$-variate polynomial $p(\bm x)=\sum_{I\in \mathbb{N}^n}\alpha_{I}\prod_{j\in [n]}x_j^{I_j}$, let 
\[p^{\leq d}(\bm x)\triangleq\sum_{I\in\mathbb{N}^n : \|I\|_1\leq d} \alpha_{I}\prod_{j\in [n]}x_j^{I_j}\] 
be the truncation of $p$ to degree $d$.

\begin{fact}
Let $p(\bm x)=\sum_{I\in \mathbb{N}^n}\alpha_{I}\prod_{j\in [n]}x_j^{I_j}$ be an $n$-variate polynomial. Then, for every point $\bm x \in[-1,1]^n$, 
\[|p(\bm x)-p^{\leq d}(\bm x)|\leq\sum_{I\in\mathbb{N}^n : \|I\|_1> d}|\alpha_{I}|.\]
\end{fact}
\begin{proof} Observe, that for every $\bm x\in[-1,1]^n$, for every $j$, $|x_j|\leq 1$, and hence
\[ |p(\bm x)-p^{\leq d}(\bm x)|=\Big|\sum_{I\in\mathbb{N}^n : \|I\|_1> d} \alpha_{I}\prod_{j\in [n]}x_j^{I_j}\Big| \leq \sum_{I\in\mathbb{N}^n : \|I\|_1> d} |\alpha_{I}|\prod_{j\in [n]} |x_j|^{I_j}\leq \sum_{I\in\mathbb{N}^n : \|I\|_1> d} |\alpha_{I}|.
\qedhere
\]
\end{proof}

\begin{cor}\label{cor:multivariate_close}
Let $p(\bm x)=\sum_{I\in \mathbb{N}^n}\alpha_{I}\prod_{j\in [n]}x_j^{I_j}$ be a polynomial of total degree $\ell \geq d$. If for every $I\in \mathbb{N}^n$ such that $\|I\|_1>d$, we have $|\alpha_{I}|\leq\eta$, then $p$ is pointwise $\tilde{\eta}$-close to $p^{\leq d}$ on $[-1,1]^n$, where $\tilde{\eta}=\eta |\{I\ \mid\ d<\|I\|_1\leq \ell\}| \leq \eta (n+1)^{\ell}$. 
\end{cor}


We are now ready to prove \autoref{thm:multivarite_degree_reduction}. 

\begin{proof}[Proof of \autoref{thm:multivarite_degree_reduction}]
Let 
\[p(\bm x) = \sum_{I\in\mathbb{N}^n:\|I\|_1\leq \ell} \alpha_{I}\prod_{j\in [n]}x_j^{I_j}.\]
Assume by contradiction that $p$ is not pointwise $\eta$-close to $p^{\leq d}$ on $[-m,m]^n$. Then, by \autoref{cor:multivariate_close} there exists $\tilde{I}$ such that $\|\tilde I\|_1>d$ and $|\alpha_{\tilde{I}}|> (n+1)^{-\ell}\eta$.
Fix $\bm a = (a_1,\dots,a_n) \in [-m,m]^n$, then the restriction of $p$ to the line $L_{(\bm 0, \bm a)}$ is
\[ p_{\bm 0, \bm a}(t)=\sum_{I\in\mathbb{N}^n : \|I\|_1\leq \ell} \alpha_{I}\prod_{j\in [n]}a_j^{I_j}t^{\|I\|_1}=\sum_{r=0}^{\ell}\Bigg(\sum_{I\in\mathbb{N}^n : \|I\|_1=r}\alpha_{I}\prod_{j\in [n]}a_j^{I_j}\Bigg)t^r.\]
By the Fourier-Chebyshev expansion, we can write each monomial $t^r=\sum_{k=0}^r\beta_{k,r}T_{k}(t)$, 
where 
\[
    \beta_{k,r}={\begin{cases}0&\text{ if }k\not\equiv r\bmod{2},\\ 2^{1-r}\binom{r}{(r-k)/2}&\text{ if }k\equiv r\bmod{2}\text{, and }k\neq 0,\\ 2^{-r}\binom{r}{(r-k)/2}&\text{ if }k=0\text{, and }r\equiv 0\bmod{2}. \end{cases}}
\]
which gives
\begin{align*}
    p_{\bm 0, \bm a}(t)&=\sum_{r=0}^{\ell}\Bigg(\sum_{I\in\mathbb{N}^n : \|I\|_1=r}\alpha_{I}\prod_{j\in [n]}a_j^{I_j}\Bigg)\sum_{k=0}^r\beta_{k,r}T_{k}(t) \\
    &= \sum_{r=0}^{\ell}\sum_{k=0}^r\Bigg(\sum_{I\in\mathbb{N}^n : \|I\|_1=r}\beta_{k,r}\alpha_{I}\prod_{j\in [n]}a_j^{I_j}\Bigg)T_{k}(t) \\
    &= \sum_{k=0}^{\ell}\Bigg(\underbrace{\sum_{r=k}^{\ell}\sum_{I\in\mathbb{N}^n : \|I\|_1=r}\beta_{k,r}\alpha_{I}\prod_{j\in [n]}a_j^{I_j}}_{\triangleq q_k(\bm a)}\Bigg)T_{k}(t).
\end{align*}
Let $q_k(\bm a)$ be the coefficient of $T_k(t)$ in the previous expansion and let $\tilde{r}=\|\tilde I\|_1$. Note that by the values of the coefficients $\beta_{k,r}$, we have $\alpha_{\tilde{I}}$ appears either in $q_{d+1}$ or in $q_{d+2}$ depending on the parity of $\tilde{r}$; let $i=d+1$ or $i=d+2$ be such that $i\equiv \tilde{r} \bmod{2}$.
Thus, the coefficient of the monomial $\prod_{j\in [n]}a_j^{\tilde{I}_j}$ in
$q_i$ is
\[ \beta_{i,\tilde{r}}a_{\tilde{I}}=2^{1-\tilde{r}}\binom{\tilde r}{(\tilde r-i)/2}\alpha_{\tilde{I}}\geq 2^{-\tilde{r}}\alpha_{\tilde{I}}\geq (2n+2)^{-\ell}  \eta.\] 
Using this, we will derive a contradiction to the following claim. 
\begin{claim}\label{clm:high_order_coefficients_are_small}
For all $\bm a\in[-m,m]^n$, $|q_i(\bm a)| \leq \sqrt{2}\varepsilon$.
\end{claim}

We defer the proof of \autoref{clm:high_order_coefficients_are_small} until later and complete the proof first. As $\|I\|_1 \leq \ell$ for every $I$, let $\bm y\in \{0,\ldots, n^{8\ell}\}^n$ be given by \autoref{lem:isolation} and consider the univariate polynomial 
$\widetilde{q_i}(z) \triangleq q_i(z^{\bm y_1},z^{\bm y_2}\dots, z^{\bm y_n})$
in $z$. 
By the guarantee of \autoref{lem:isolation}, for any $I \neq I'$ with $\|I\|_1,\|I'\|_1\leq \ell$, it holds that $\langle \bm y, I \rangle \neq \langle \bm y, I' \rangle$, and thus the coefficients of $\widetilde q_i(z)$ are exactly the same as coefficients of $q_i$ (that is, no two monomials become the same after the substitution of $z^{\bm y_i}$).
Therefore, there exists a coefficient in $\widetilde{q_i}$ which is at least $ (2n+2)^{-\ell}  \eta$. 
On the other hand, since $\langle \bm y, I \rangle \leq \|I\|_1 n^{8\ell}\leq \ell n^{8\ell}$, the degree of $\widetilde q_i$ is at most $ \ell n^{8\ell}\leq n^{9\ell}$, and thus \autoref{cor:non-monic-lower-bound_in_small_cube} implies that there is some $z\in[-m,m]$, such that
$|\widetilde q_i(z)| \geq 2^{-2 n^{18\ell}}m^{ n^{9\ell}}  \eta \geq (\frac{2}{m})^{-2n^{18\ell}} \eta > \sqrt{2}\varepsilon$.
This contradicts \autoref{clm:high_order_coefficients_are_small}.
\end{proof}

\begin{proof}[Proof of \autoref{clm:high_order_coefficients_are_small}]
Let $\widehat p_{\bm 0, \bm a}(t)$ be the univariate degree-$d$ polynomial which is pointwise $\varepsilon$-close to $p_{\bm 0, \bm a}(t)$ on $t\in[-1,1]$, and let its Fourier-Chebyshev expansion be $\widehat p_{\bm 0, \bm a}(t) = \sum_{k=0}^d \gamma_k T_k(t)$. Consider the error polynomial 
\[ e(t) \triangleq  p_{\bm 0, \bm a}(t)-\widehat p_{\bm 0, \bm a}(t)=\sum_{k=0}^{\ell}(q_k(\bm a) - \gamma_k)T_k(t),\]
where we define $\gamma_k\triangleq 0$ for $k>d$. Note that since $p_{\bm 0, \bm a}$ and $\widehat p_{\bm 0, \bm a}$ are $\varepsilon$-close on $[-1,1]$, $|e(t)| \leq \varepsilon$ for all $t \in [-1,1]$. Letting $w(t) \triangleq (1-t^2)^{-1/2}$ be the Chebyshev weight function, and noting that $\int_{-1}^1 w(t) dt = \pi$, we have,
\begin{align*}
 \varepsilon^2 \pi &\geq \int_{-1}^1 e^2(t)w(t) dt \\
 &= \int_{-1}^1 \bigg( \sum_{k=0}^{\ell} (q_k(\bm a) - \gamma_k) T_k(t) \bigg)^2 w(t) dt\\
 &=\sum_{k=0}^{\ell} (q_k(\bm a) - \gamma_k)^2 \int_{-1}^1  T_k(t)T_k(t)  w(t) dt\\
 &\geq  \frac{\pi}{2}\sum_{k=0}^{\ell} (q_k(\bm a) - \gamma_k)^2 \\
 &\geq \frac{\pi}{2} \left(q_i(\bm a)\right)^2,   
\end{align*}
where the first steps by the orthogonality of Chebyshev polynomials~\eqref{eq:chebyshev_orthogonality_property}, and the final inequality follows because $i >d$ and so $\gamma_i =0$. Rearranging, we conclude that $|q_i(\bm a)| \leq \sqrt{2}\varepsilon $. 
\end{proof}


\subsubsection{Extrapolation}

In this section we show that if $g$ is pointwise close to a degree $d$ polynomial then within $\ball(\bm 0,r)$, then it must be pointwise close to a degree-$d$ polynomial within a bigger ball $\ball(\bm 0,R)$. 

\begin{lem}\label{lem:extrapolation_approx}
    Let $R>r'>0$ be any real numbers. If $g$ is pointwise $\eta$-close to a degree-$d$ polynomial in $\ball( \bm 0, r')$, then $g$ is pointwise $(12R/r')^d\eta$-close to a degree-$d$ polynomial on all points in $\ball(\bm 0,R)$.
\end{lem}
\begin{proof}
    Let $H \colon \mathbb{R}^n \to \mathbb{R}$ be the degree-$d$ polynomial which is $\eta$-close to $g$ on $\ball(\bm 0, r')$. We will argue that for any $\bm x \in \ball(\bm 0, R)$,
    \[ |g(\bm x) - H(\bm x)| \leq (12R/r')^d\eta.\]

    If $\bm x \in \ball(\bm 0, r')$, then this holds by assumption, so we consider the case when $\bm x \not \in \ball(\bm 0,r')$. Recall that we define the value of $g$ on points $\bm x \not \in \ball(\bm 0,r')$ by pretending that it \emph{is} a degree-$d$ polynomial and using $d+1$ points in $\ball(\bm 0,r')$ to extrapolate its value along radial lines from within the ball. In particular, let $c_0,\ldots, c_{d}$ be the Chebyshev nodes $c_i \triangleq (r'/\| \bm x\|_2)\cos\left({\frac{\pi }{d+1}(i+1/2)}\right)$, scaled so that they lie within $\pline_{\bm 0, \bm x} \cap \ball(\bm 0, r')$.
    Then, the value of $g(\bm x)$ for $\bm x \not \in \ball(\bm 0, r')$ is defined by interpolating a degree-$d$ univariate polynomial $p_{\bm x}$ such that $p_{\bm x}( c_i) = g(\bm x c_i)$ for $i$, and then the value of $g(\bm x)$ is defined as $p_{\bm x}(1)$. 

   Thus, in order to bound the distance between $g(\bm x)$ and $H(\bm x)$, it suffices to bound the distance between $p_{\bm x}(t)$ and  $H_{(\bm 0, \bm x)}(t)$ for $t=1$. Consider the error polynomial $e(t) \triangleq p_{\bm x}(t)  - H_{(\bm 0, \bm x)}(t)$. As $e$ is a polynomial of degree at most $d$, we can consider its Fourier-Chebyshev expansion,
 \[e(t)=\sum_{i=0}^d \alpha_iT_i(t \| \bm x\|_2/r'),\]
where $T_i(t \|\bm x\|_2/r')$ is the $i$th Chebyshev polynomial with the Chebyshev nodes back-scaled to the interval $[-1,1]$. By assumption, $e( c_i) \leq \eta$ for each $i$, which allows us (by the same argument as in \autoref{lem:stability_univariate_interpolation} and \autoref{lem:stability_univariate_interpolation_scaled}) to bound the coefficients $|\alpha_i| \leq \sqrt{2} \eta$. The $i$th Chebyshev polynomial involves at most $i+1$ terms, each of which are of degree at most $i$ and has coefficients of value at most $2^i$, and therefore $|T_i( \|\bm x\|_2/r')| \leq (i+1)2^i ( \|\bm x\|_2/r')^i$. Altogether, this allows us to bound the value of the error polynomial on $t=1$ by
\[ |e(1)| \leq \sqrt{2} \eta (d+1)^2 (2  \|\bm x\|_2 /r')^d\leq (12  R /r')^d\eta  , \]
where the second inequality is by $\sqrt{2}(d+1)^2\leq 6^d$ for every $d\geq 1$, and the last holds as $\bm x \in \ball(\bm 0,R)$.
Since $g(\bm x) = p_{\bm x}(1)$, we have that the distance between $g(\bm x)$ and $H(\bm x)$ is at most $(12R/r')^d\eta$.
\end{proof}

\subsection{Approximate Polynomial Representation in a Large Ball} 

We now prove the approximate analogue of \autoref{thm:g-is-mult-d-poly} which showed $g$ is a degree-$d$ polynomial over $\mathbb{R}^n$.

\begin{lem}\label{thm:g-is-mult-d-poly_approx}Let $\srad=(4d)^{-6}$ and $\lrad>\srad$. If \Call{ApproxCharacterizationTest}{} fails with probability at most $2/3$, then $g$ is point-wise $2^{(2n)^{45d}}\lrad^d\delta$-close to a degree-$d$, $n$-variate polynomial on all points in $\ball(\bm 0,\lrad)$. 
\end{lem}
\begin{proof}
By \autoref{lem:g-on-lines-in-B-is-poly_approx}, $g$ restricted to any line segment $\pline_{\bm{p},\bm{q}}^{\ball} = \pline_{\bm p, \bm q} \cap \ball(\bm 0,\srad)$ is point-wise $2^{15d^2}\delta$-close to a unique univariate degree-$d$ polynomial. Applying \autoref{lem:local_to_global_approx} (with $m=r/(2\sqrt{n})$), we have that $g$ is pointwise $2^{15d^2}(4\sqrt{n}/\srad)^{n^{40d}}\delta$-close to a degree $d,n$-variate polynomial on every point in the hypercube $\cube=[-m,m]^n$, contained within $\ball(\bm 0,\srad)$. We then consider a smaller ball $\ball(\bm 0,\srad')$ of radius $\srad'=m$, contained within $\cube$. By \autoref{lem:extrapolation_approx}, it follows that $g$ is point-wise $2^{15d^2}(4\sqrt{n}/\srad)^{n^{40d}}(24\sqrt{n}R/r)^d\delta\leq 2^{(2n)^{45d}}\lrad^d\delta$-close to a degree-$d,n$-variate polynomial in $\ball(\bm 0,\lrad)$.
\end{proof}

Finally, we are ready to prove the main lemma of this section.

\begin{proof}[Proof of \autoref{lem:main_approx_poly}] 
Suppose that \Call{ApproxCharacterizationTest}{} fails with probability at most $2/3$, then by \autoref{thm:g-is-mult-d-poly_approx}, $g$ is pointwise $2^{(2n)^{45d}}\lrad^d\delta$-close to a degree-$d$ polynomial in $\ball(\bm 0,\lrad)$. It remains to bound $\Pr[|g(\bm p) -\text{\Call{ApproxQuery-$g$}{$\bm p$}}|>(12\lrad/\srad)^d2^{d+4}\delta]$, where $\bm p \in \ball(\bm 0,\lrad)$. In the \textbf{YES} case, $f$ is point-wise $\alpha$-close to  a degree-$d$ polynomial $h$, and so for any $\bm p,\bm q$

\[|\sum_{i=0}^{d+1}\alpha_i\cdot f(\bm p+i\bm q)|=|\sum_{i=0}^{d+1}\alpha_i\cdot( f(\bm p+i\bm q)-h(\bm p+i\bm q))|\leq\sum_{i=0}^{d+1}|\alpha_i|\cdot\alpha\leq 2^{d+1}\alpha=\delta.\]
Therefore, \Call{ApproxCharacterizationTest}{} always passes, and  \Call{ApproxQuery-$g$}{$\bm p$} returns a value that is $2^{d+3}\delta$-close to $g(\bm p)$. Assume that $f$ is not a degree-$d$ polynomial. To query $g$ on a point $\bm{p} \in \ball(\bm 0,\lrad)$, \Call{ApproxQuery-$g$}{$\bm p$} attempts to obtain $d+1$ points on the line segment $\pline_{\bm 0, \bm p}^{\ball}$ and then interpolate $g$ along this line. For these points $\bm s\in\{c_k\bm p\}_{k=0}^{d}$, \Call{ApproxQuery-$g$-InBall}{$\bm s$} samples an additional $N'_{\ref{alg:subroutines_approx}} = O(\log (1/\varepsilon))$ points $\bm q_1,\ldots, \bm q_{N'_{\ref{alg:subroutines_approx}}}\sim\cN(\bm 0,I)$, and checks whether
\[\Big|\sum_{i \in [d+1]} \alpha_i \cdot f(\bm s + i \bm q_1) - \sum_{i \in [d+1]} \alpha_i \cdot f(\bm s + i \bm q_j)\Big|\leq 2^{d+2}\delta,\]
for all $j \in [N'_{\ref{alg:subroutines_approx}}]$; it rejects if any of these checks fail. By the definition of $g_{\bm q}$, this is equivalent to checking whether $|g_{\bm q_1}(\bm s) - g_{\bm q_j}(\bm s)|>2^{d+2}\delta$; by \autoref{lem:well_defined_of_g_approx}, this occurs with probability at most $1/(7d)$, since $\bm s\in \ball(\bm 0,\srad)$. The probability that \Call{ApproxQuery-$g$-InBall}{$\bm s$} doesn't reject, 
yet $|g(\bm s)-$
\Call{ApproxQuery-$g$-InBall}{$\bm  s$} $|>2^{d+4}\delta$, 
is the probability that: $|g(\bm s)-g_{\bm q_1}(\bm s)|>2^{d+4}\delta$, and $|g_{\bm q_1}(\bm s)-g_{\bm q_j}(\bm s)|\leq 2^{d+3}\delta$, (and therefore $|g(\bm s)-g_{\bm q_j}(\bm s)|> 2^{d+3}\delta$)
for every $\bm q_j$. By \autoref{cor:replacement-is-good_approx}, this probability is at most ${(7d)}^{-N'_{\ref{alg:subroutines_approx}}} < 2^{-N'_{\ref{alg:subroutines_approx}}}/{d}^{N'_{\ref{alg:subroutines_approx}}} \leq \varepsilon/(4(d+1))$, where the final inequality follows by choosing $N'_{\ref{alg:subroutines_approx}} = O(\log (1/ \varepsilon))$. As \Call{ApproxQuery-$g$}{$\bm p$} approximately recovers the value of $g$ on points $\{c_i \bm p\}_{i=0}^d$, we have that for every $i \in \{0,\ldots, d\}$,
\[\Pr[|g_{\bm 0,\bm p}(c_i)-\text{\Call{ApproxQuery-$g$-InBall}{$c_i\bm p$}}|> 2^{d+4}\delta]\leq\frac{\eps}{4(d+1)}.\]
Thus, by \autoref{lem:extrapolation_approx} and a union bound over $i$, 
\begin{align*}
    &\Pr[|g(\bm p)-\text{\Call{ApproxQuery-$g$}{$\bm p$}}|>(12\lrad/\srad)^d2^{d+4}\delta]=\Pr[|g(\bm p)-g_{\bm 0,\bm p}(1)|>(12\lrad/\srad)^d2^{d+4}\delta]\\
    \leq &\sum_{i=0}^{d}\Pr[|g_{\bm 0,\bm p}(c_i)-\text{\Call{ApproxQuery-$g$-InBall}{$c_i\bm p$}}|> 2^{d+4}\delta]\leq\frac{\eps}{4}.
    \qedhere
\end{align*}
\end{proof}

%% file: exact-discrete.tex

\section{Exact Testing over Discrete Domains}\label{sec:discrete}

In this section we show that the test for degree-$d$ polynomials from \autoref{sec:exact} can be modified to work for (sufficiently dense) discrete domains. The main theorem of this section is as follows: 

\discretesampling*

The key idea behind our tester is to define the self-corrected function $g$ relative to a \emph{discretized} Gaussian distribution defined over $\mathcal{L}$.


\begin{defn}\label{defn:discrete_gaussian}
Given a lattice $\caL$, and any $s>0$, the \emph{discrete Gaussian} $\caG(\caL, s)$ is the probability distribution over $\caL$ 
such that the probability of drawing $\bm x \in \caL$ is $\propto\tau_{s}(\bm x)\triangleq\exp(-\pi\|\bm x\|^2/s^2)$. (If unspecified, $s=1$.)
\end{defn}

That we are able to efficiently sample from a discrete Gaussian is guaranteed by the following lemma.

\begin{lem}[Lemma 2.3 in~\cite{BLPRS13}]\label{lem:latsamp}
There is a probabilistic polynomial time algorithm that given a positive integer $B$ and parameter $r \geq \Omega(\sqrt{\log n}/B)$, outputs 
a sample distributed according to $\caG (\frac1B \Z^n, r )$. 
\end{lem}

At a high-level, the design of our tester will follow the same strategy as the design of our exact tester from \autoref{sec:exact}, with several modifications to handle the lattice $\caL$. From our unknown function $f$, we will define a self-corrected function $g$ such that we have query access to $g$, and such that if our tests pass with sufficiently high probability then $g$ is a degree-$d$ polynomial on $\caL$, and equals $f$ on $\caL$ if $f$ is itself a degree-$d$ polynomial.

As before, we define $g$ on points within a small ball $\ball(\bm 0,r)$; for points $\bm p \in \caL \setminus \ball(\bm 0,r)$ we will define their value by extrapolating the value of $g$ within $\ball(\bm 0,r)$ by choosing $d+1$ points within $\ball(\bm 0,r) \cap \caL$ along the line $\pline_{\bm 0, \bm p}$, using them to interpolating a degree-$d$ univariate polynomial $p_{\bm p}$, and then using $p_{\bm p}$ to define the value of $g(\bm p)$ (see \autoref{fig:extrapolation}). In order to certify that $g$ is indeed a degree-$d$ polynomial, we will use the following variant of the \localCharacterization ; a proof of which is given in \autoref{sec:appendix}.

\begin{restatable}{discrete_local_characterization}{localchardiscrete}\label{thm:localchar_discrete}
Fix $a>0$, $M \geq d+1$, and let $S\triangleq\{\frac{ia}{M}:i\in\mathbb N\}$. If $f : [0,a] \rightarrow \mathbb{R}$ is a univariate function such that $\Delta_{\frac{a}{M}}^{(d+1)}[f](x)=0$, for every $x\in S\cap[0,a]$ satisfying $x+(d+1)\frac{a}{M}\in[0,a]$, then $f$ agrees with a degree-$d$ polynomial over the points in $S \cap [0,a]$.
\end{restatable}

However, for an arbitrary point $\bm p \in \caL \setminus \ball(\bm 0,r)$ there may not be $d+1$ points on the line segment $\pline_{\bm 0, \bm p} \cap \caL \cap \ball(\bm 0,r)$ --- this can occur if $\bm p$ is sufficiently far away from $\bm 0$ --- and thus we cannot define $p_{\bm p}$.
To remedy this, we make two modifications. First, we assume that our distribution is $(\varepsilon/4,R)$-concentrated --- that $1-\varepsilon/4$ fraction of the mass of the unknown distribution $\caD$ is in $\ball(\bm0, R)$. We define $g$ only on points within $\ball(\bm0, R)$ and we will not test whether $f$ differs from a degree-$d$ polynomial outside of $\ball(\bm 0, R)$; as these point constitute only a small $\varepsilon/4$ portion of $\caD$, which we can simply fold into the error of our tester. Second, in order to ensure that for any point $\bm p \in \ball(\bm 0,R)$, $g$ is defined on at least $d+1$ points on the line $\pline_{\bm 0, \bm p}$ within $\ball(\bm 0,r)$, we define $g$ on a finer lattice $\caL' \triangleq \frac{r}{(d+1)R} \caL = \frac{r}{BR(d+1)} \mathbb{Z}^n$ within $\ball(\bm 0, r)$.

\paragraph{The Self-Corrected Function.} Let $R>r>0$, and let our (unknown) $(\varepsilon/4,R)$-concentrated distribution $\caD$ be supported over a given lattice $\caL = \frac{1}{B}\mathbb{Z}^n$. Let $\caL' \triangleq \frac{r}{(d+1)R} \caL$ be a refinement of $\caL$. We define the self-corrected function $g$, whose domain is $\caL \cup (\caL' \cap B(\bm 0, r) )$, as follows. 
Let $\alpha_i \triangleq {(-1)}^{i+1}\binom{d+1}{i}$, and for any $\bm p \in \ball(\bm 0,r)\cap \caL'$, and $\bm q \in \mathcal{L}$, let $g_{\bm q}(\bm p) \triangleq \sum_{i=1}^{d+1} \alpha_i \cdot f(\bm p+i \bm q).$
For any $\bm p \in \ball(\bm 0,r)\cap  \caL'$, we define 
\[g(\bm p)\triangleq\maj_{\bm q\sim\caG(\caL,1)}[g_{\bm q}(\bm p)].\] 

For points $\bm p \in (B(\bm 0, R) \setminus \ball(\bm 0, \srad)) \cap \caL$ we define the value of $g(\bm p)$ by interpolating a degree-$d$ univariate polynomial along the line $\pline_{\bm 0,\bm p}$ as follows: Let $c_0, \ldots, c_d \in \mathbb{R}$ be $d+1$ ``distinguished'' points (arbitrary, but fixed) on the line $\pline_{\bm 0, \bm p}$ within $\ball(\bm 0, \srad) \cap \frac{r}{R(d+1)} \caL$; in \autoref{alg:subroutines-discrete} we choose $c_i = ir/((d+1)\|\bm p\|_2)$ and note that these points lie within $\caL'$. Let $p_{\bm p}$ be the unique univariate polynomial such that $p_{\bm p}(c_i) = g(c_i \bm p)$ for every $i \in [d+1]$. We define $g(\bm p) \triangleq p_{\bm p}(1)$.

Our tester is given in \autoref{alg:low_degree_main_algorithm_discrete}, with corresponding subroutines in \autoref{alg:subroutines-discrete}.

\begin{algorithm}[h]
\caption{Low-Degree Discrete Tester}\label{alg:low_degree_main_algorithm_discrete}
\Procedure{\emph{\Call{DiscreteLowDegreeTester}{$f,d,\mathcal{D},\eps,\lrad,B$}}}{
  \Given{Query access to $f$, a degree $d\in\mathbb{N}$, sampling access to an $(\varepsilon/4,\lrad)$-concentrated unknown distribution $\mathcal{D}$ supported over the lattice $\caL\triangleq\frac{1}{B}\Z^n$, where $\eps$ is the farness parameter, and $B$ is the density parameter.}
    \textbf{Reject} if \emph{\Call{DiscreteCharacterizationTest}{}}  \textbf{rejects}\;
    \For{$N_{\ref{alg:low_degree_main_algorithm_discrete}} \gets O(\varepsilon^{-1})$ times}{
      Sample $\bm{p} \sim \mathcal{D}$;\\
      \If{$\bm p\in\ball(\bm 0,\lrad)$}{
      \textbf{Reject} if $f(\bm{p}) \neq$ \emph{\Call{DiscreteQuery-$g$}{$\bm{p}$}} or if \emph{\Call{DiscreteQuery-$g$}{$\bm{p}$}}  \textbf{rejects}.
    }}
    \textbf{Accept}. 
  }  
\end{algorithm}
\begin{algorithm}[!t]
\caption{Discrete Subroutines}\label{alg:subroutines-discrete}

\Procedure{\emph{\Call{DiscreteCharacterizationTest}{}}}{
$N_{\ref{alg:subroutines-discrete}} \gets O(d^2)$ \;
\For{$N_{\ref{alg:subroutines-discrete}}$ times }{
\For{$j\in \{1,\dots, d+1\}$}{
    \For{$t\in \{0,\dots, d+1\}$}{
      Sample $\bm{p}\sim\caG(j\caL,j\sqrt{t^2+1}),\bm{q}\sim\caG(\caL,1)$;\Comment{[$j^2(t^2+1)$ vs. $1$ Test.]}\\
      \textbf{Reject} if $\sum_{i=0}^{d+1}\alpha_i\cdot f(\bm{p}+i\bm{q})\neq 0$\;
      Sample $\bm{p}\sim\caG(j\caL,j),\bm{q}\sim\caG(\caL,\sqrt{t^2+1})$;\Comment{[$j^2$ vs. $(t^2+1)$ Test.]}\\
      \textbf{Reject} if $\sum_{i=0}^{d+1}\alpha_i\cdot f(\bm{p}+i\bm{q})\neq 0$\;
      }
      Sample $\bm{p},\bm{q}\sim\caG(j\caL,j)$;\Comment{[$j^2$ vs. $j^2$ Test.]}\\
      \textbf{Reject} if $\sum_{i=0}^{d+1}\alpha_i\cdot f(\bm{p}+i\bm{q})\neq 0$\;
    }}
    \textbf{Accept}\;
    }

\Procedure{\emph{\Call{DiscreteQuery-$g$}{$\bm{p}$}}}{
$r \gets d\sqrt{n}/(2B)$ \;
\If{$\bm{p}\in \ball(\bm{0},\srad)$}{
\textbf{return} \emph{\Call{DiscreteQuery-$g$-InBall}{$\bm{p}$}}\;}

\For{$i\in [d+1]$}{
$c_i \gets \frac{i\srad}{(d+1)\Vert \bm{p}\Vert_{2}}$\;
 $v(c_i)\gets$ \emph{\Call{DiscreteQuery-$g$-InBall}{$c_i\bm{p}$}}  \;
}
Let $p_{\bm p} \colon \mathbb{R}\to\mathbb{R}$ be the unique degree-$d$ polynomial such that $p_{\bm p}(i) = v(c_i)$  for $i \in [d+1]$\;
\textbf{return} $p_{\bm p}(1)$\;
}

\Procedure{\emph{\Call{DiscreteQuery-$g$-InBall}{$\bm{p}$}}}{
$N'_{\ref{alg:subroutines-discrete}}  \gets O(\log\frac{1}{\varepsilon})$\;
Sample $\bm q_1,\dots,\bm q_{N'_{\ref{alg:subroutines-discrete}}}\sim \caG \big(\frac{r}{R(d+1)}\caL,1 \big)$\;
\textbf{Reject} if there exists $j\in\set{2,\dots,N'_{\ref{alg:subroutines-discrete}}}$ such that $ \sum_{i=1}^{d+1} \alpha_i \cdot f(\bm{p}+i\bm{q}_1)\neq \sum_{i=1}^{d+1} \alpha_i \cdot f(\bm{p}+i\bm{q}_j)$\;
\textbf{return} $ \sum_{i=1}^{d+1} \alpha_i \cdot f(\bm{p}+i\bm{q}_1)$\;
}
\end{algorithm}

In the remainder of this section we will prove \autoref{thm:low_degree_test_discrete}. However, before we are able to do so, we require several structural results about Lattices and discrete Gaussians, which will occupy the next subsection.

\subsection{Preliminaries on Lattices and  Discrete Gaussians}\label{sec:prelim_lattices}
First, we recall that many of the properties of Gaussian distributions translate over to their discrete variants. 

\begin{fact}[Fact 2 in~\cite{AGHS13}]\label{fact:discrete-int-multiple}
Suppose $\caL$ is a lattice, and $s>0$ is a parameter. If $\bm x$ is distributed as $\caG(\caL, s)$, then for any integer $t$, $t\bm x$ is distributed as $\caG(t\caL, ts)$.  
\end{fact}

We implicitly use this fact to sample random vectors from scaled discrete gaussians in \autoref{alg:subroutines-discrete}. Next, we record a bound on the total variation distance between two Gaussians which is analogous to \autoref{lem:bound-TV-different-means}. To state the theorem, we need the following smoothing parameter defined in~\cite{MR07} as a parameter of a lattice with the following property: if one picks a noise vector from a Gaussian distribution with radius at least as
large as the smoothing parameter, and reduces the noise vector modulo the fundamental parallelopiped of the lattice, then the resulting distribution is very close to uniform.

\begin{defn}
For a lattice $\caL \subseteq \R^n$ and a parameter $\vartheta > 0$, the \emph{smoothing parameter} $\eta_\vartheta(\caL)$ is the
smallest $s>0$ such that:
\[
\tau_{1/s}\left(\caL^*\setminus\{\bm 0\}\right)\triangleq\sum_{\bm x\in\caL^*\setminus\{\bm 0\}}\tau_{1/s}(\bm x)=\sum_{\bm x \in \caL^*\setminus \{\bm 0\}} \exp\left(-{\pi \|\bm x\|^2}{s^2}\right) \leq \vartheta,
\]
where $\tau$ is defined in \autoref{defn:discrete_gaussian}, and $\caL^* \triangleq \{{\bm x} \in \mathrm{span}(\caL): \forall \bm y \in \caL, \langle \bm x, \bm y\rangle \in \Z\}$ is the \emph{dual lattice} of $\caL$.
\end{defn}

\begin{obs}\label{obs:relating-eta-of-int-refinements}
For any lattice $\caL$, parameter $\vartheta>0$, and $i\in\mathbb{Z}_{\geq 2},{(i\caL)}^*=\frac{1}{i}\caL^*$, and
\begin{align*}
    \eta_\vartheta(i\caL)=\min_{s>0}\bigg\{\sum_{\bm x\in{(i\caL)}^*\setminus\{\bm 0\}}\exp\left(-\pi\|\bm x\|^2s^2\right)\bigg\}=\min_{s>0}\bigg\{\sum_{\bm x\in\caL^*\setminus\{\bm 0\}}\exp\left(-\pi\|\bm x\|^2s^2/i^2\right)\bigg\}\leq i^2\eta_\vartheta(\caL).
\end{align*}
\end{obs}
The next theorem follows from (a quantified version of) Theorem 3.3 of~\cite{MP13}, which we prove in \autoref{sec:AppendixC}.
Since $\vartheta$ is taken as a negligible function of $n$, so for small $k$, $4k\vartheta\ll 1$. Later, while invoking it, we will set $k\leq d+2$, which satisfies the restriction on $k\text{, and }\vartheta$ contained therein.

\begin{thm}\label{lem:indsum}
Let $k\in\Z_{>0}$, $\vartheta\in\R_{>0}$ be such that $k\leq 1/(4\vartheta)$. Let $\caL \subseteq \R^n$ be a lattice, and let $\bm s \in\R_{>0}^k$ and $\bm z \in\Z^k$ be such that $s_i\geq\sqrt{2}\|\bm z\|_{\infty}\eta_{\vartheta}(\caL)$ for every $i\in[k]$. If ${\bm y}_1, \dots, {\bm y}_k$ are sampled independently from $\bm y_i\sim\caG(\caL, s_i)$ then the distribution of $\sum_{i=1}^k z_i{\bm y}_i$ is $4k\vartheta$-close in total variation distance to
$\caG(\gcd(\bm z)\caL, \sqrt{\sum_{i=1}^k{(z_i s_i)}^2})$.
\end{thm}

Next, we recall a simple bound on the total variation distance incurred by shifting the center of a discrete Gaussian. Denote by $\mathrm{erf}(\cdot)$, the Gaussian error function.

\begin{lem}[Remark from Lemma 6 of~\cite{AGHS13}]\label{lem:shift}
For any full-rank lattice $\caL \subseteq \R^n$, $\vartheta \in (0,1/2)$, $c>1$ and a parameter $s$ such that $s > (1+2c)\eta_\vartheta(\caL)$, the following holds:
The total variation distance between $\caG(\caL, s)$ and $\caG({\bm v} + \caL, s)$  for any ${\bm v} \in \caL$ is at most 
\[
\frac{\mathrm{erf}(q)}{\mathrm{erf}(qc)}\cdot \frac{1+\vartheta}{1-\vartheta},
\]
where $q= \|\bm v\|^2_2\sqrt{\pi}/s$.
\end{lem}

By combining \autoref{lem:shift} and \autoref{obs:relating-eta-of-int-refinements} we obtain the following corollary, which bounds the distance between two discrete Gaussians.

\begin{cor}\label{cor:discrete_gaussian_tvd_bound}
    Let $i \in [d+1]$, $t$ be a non-negative integer,  $\vartheta \in (0,1/2)$, $\srad \leq d \sqrt{\eta_\vartheta(\caL)}$, and let $\caL = \frac{1}{B}\mathbb{Z}^n$. Then, for any $\bm p \in \ball(\bm 0, \srad) \cap (d+1)!\caL$,
    \[ \dtv\left(\caG(\bm p + i \caL,it), \caG(i \caL, it)\right) \leq 98d^2 \eta_\vartheta(\caL).  \] 
\end{cor}
\begin{proof} 
    To bound the total variation distance, we aim to apply \autoref{lem:shift}. To do so, we need to choose $c$ such that $it > (1+2c) \eta_\vartheta(i\mathcal{L})$. Observe that
    \[ (1+2c) \eta_\vartheta(i \caL) \leq (1+2c)i^2 \eta_\vartheta(\caL) \leq (1+2c)(d+1)^2\eta_\vartheta(\caL) <12cd^2 \eta_\vartheta(\caL), \]
    where the first inequality follows by \autoref{obs:relating-eta-of-int-refinements}. Thus, letting $c \triangleq t(12d^2 \eta_\vartheta(\caL))^{-1}$ we have $(1+2c) \eta_\vartheta(i \caL) <t \leq it$. Applying \autoref{lem:shift}, 
    \[ \dtv( \caG(\bm p+i\caL, it), \caG(i\caL, it)) \leq  \frac{(1+\vartheta)\mathrm{erf}(\|\bm p\|^2_2\sqrt{\pi}/it)}{(1-\vartheta)\mathrm{erf}(c\|\bm p\|^2_2\sqrt{\pi}/it)}.\]
Note that because $\mathrm{erf}(x)\triangleq\Pr_{y\sim \cN(0,1/2)} [y\in[-x,x] ]$ and the PDF of $\cN(0,1/2)$ is $\frac{1}{\sqrt{\pi}}e^{-x^2}$, we have $\frac{2x}{\sqrt\pi}e^{-x^2}\leq\mathrm{erf}(x)\leq \frac{2x}{\sqrt\pi}$.
This means that for any $q>0$, $\frac{e^{-q^2}}{c}\leq\frac{\mathrm{erf}(q)}{\mathrm{erf}(qc)}\leq  \frac{e^{c^2q^2}}{c}$, 
and so if $q \leq e/c$, it holds that $\frac{\mathrm{erf}(q)}{\mathrm{erf}(qc)}\leq \frac{e}{c}$.
Because $\bm p \in B(\bm 0, r)$, $q \triangleq \|\bm p\|_2^2 \sqrt{\pi}/it \leq r^2 \sqrt{\pi}/it \leq e/c$ by our choice of $c$ and $r$. Thus, 
\[ \frac{(1+\vartheta)\mathrm{erf}(\|\bm p\|^2_2\sqrt{\pi}/it)}{(1-\vartheta)\mathrm{erf}(c\|\bm p\|^2_2\sqrt{\pi}/it)} \leq \frac{1+\vartheta}{1-\vartheta} (e/c) \leq 3e/c = 98d^2 \eta_\vartheta(\caL), \] 
where the second inequality follows because $\vartheta \in (0,1/2)$.
\end{proof}

\subsubsection{Correctness of the Discrete Low Degree Tester}

We are now ready to prove \autoref{thm:low_degree_test_discrete}.
In fact, we prove a more general result, which handles lattices parameterized by their smoothness parameter. We state this generalization next. 

\begin{thm}\label{thm:discrete_precise}
     Let $R>r>0$ and $B',d>0$ satisfy  $\srad B'> (d+1)!n^{3/2+d}$. Let  $\vartheta \leq (192d)^{-1}$ be such that $\eta_\vartheta(\frac{1}{B'}\mathbb{Z}^n) \leq (6d)^{-4}$, and $\srad \leq d \sqrt{\eta_\vartheta(\frac{1}{B'}\mathbb{Z}^n)}$. Let $\caL = \frac{R (d+1)}{\srad B'} \mathbb{Z}^n$, $\caL' = \frac{1}{B'}\mathbb{Z}^n$, $\varepsilon>0$, and $f:\R^n\to\R$.
     There is a one-sided error,  $O(d^5+\frac{d^2}{\varepsilon}\log\frac{1}{\varepsilon})$-query tester for testing whether $f$ agrees with a degree-$d$ polynomial on $\caL$ with respect to an $(\varepsilon/4,R)$-concentrated distribution $\caD$ supported over $\caL$. 
\end{thm}
  \autoref{thm:low_degree_test_discrete} follows by letting $B = \frac{r}{R(d+1)} B'$, setting $\vartheta = 2^{-n}$, $r=\frac{(d+1)n^{1/4}}{4\sqrt{B'}}$, and $B' > 16 n^{5/2+2d}\cdot d^{2d}$. To see that the inequalities in the statement of \autoref{thm:discrete_precise} are satisfied, we use the following bound which can be found in \cite{MR07}:
 \[ \frac{\sqrt{n/\pi}}{B} \leq \eta_{2^{-n}} \left(\frac{1}{B} \mathbb{Z}^n \right) \leq\frac{\sqrt{n}}{B}. \]

The following lemma records the properties of $g$ which are guaranteed by our tester.

\begin{lem}\label{lem:discrete_polynomial_lem}
    Assume that the conditions of \autoref{thm:discrete_precise} hold. If \Call{DiscreteCharacterizationTest}{} fails with probability at most $2/3$, then $g$ consistent with a degree-$d$ polynomial within $\ball(\bm 0, R) \cap \caL$, and furthermore for every $\bm p\in\ball(\bm 0,\lrad)\cap\caL$, $g(\bm{p})=$ \Call{DiscreteQuery-$g$}{$\bm{p}$} with probability at least $1-\frac{\varepsilon}{4}$.
\end{lem}

We prove the main theorem assuming that this lemma holds.

\begin{proof}[Proof of \autoref{thm:discrete_precise}]
The proof follows the same argument as the proof of \autoref{thm:low_degree_test}, with two small changes. First, we use \autoref{lem:discrete_polynomial_lem} in place of \autoref{lem:main_lemma_well_definedness_and_querying_g}. Second, since we only test points within $\ball(\bm 0, R)$, we err on those points which are not. However, since $\caD$ is $(\varepsilon/4,R)$-concentrated, the probability mass of these points is at most $\varepsilon/4$, and this is folded into the error probability of our tester. 

\end{proof}

\subsection{Polynomial Representation on Lines Within a Small Ball}
We turn now to proving \autoref{lem:discrete_polynomial_lem}; this will be done over the following three subsections.
First, we prove that $g$ is consistent with a degree-$d$ polynomial on every line within $\ball(\bm 0, \srad)$.

\begin{thm}{(Polynomial representation on lines)}\label{lem:g-on-lines-in-B-is-poly-discrete}
Let $B,r>0$ be such that $(rB/(d+1)!n^{3/2+d})^n > 1$.
Let $\vartheta \leq (192d)^{-1}$ and $\caL = \frac{1}{B}\mathbb{Z}^n$ be a lattice such that $\eta_\vartheta(\caL) \leq (6d)^{-4}$.
If \Call{DiscreteCharacterizationTest}{} fails with probability at most $2/3$,
then for any $\bm{a},\bm{b}\in \ball(\bm 0,\srad)\cap (d+1)!\caL$, there is a degree-$d$ univariate polynomial which is consistent with $g_{\bm a, \bm b}(x)=g(\bm a+x\bm b)$ on every point $x$ such that $\bm a + x \bm b \in \ball(\bm 0,\srad) \cap (d+1)!\caL$. 
\end{thm}

The proof of \autoref{lem:g-on-lines-in-B-is-poly-discrete}  follows by exactly the same argument as the proof of \autoref{lem:g-on-lines-in-B-is-poly}, using the \discreteLocalCharacterization, and \autoref{thm:characterization_in_small_ball_discrete} in place of \autoref{thm:characterization_in_small_ball}. 

\begin{lem}\label{thm:characterization_in_small_ball_discrete}
    Let $B,r>0$ be such that $(rB/(d+1)!n^{3/2+d})^n > 1$.
    Let $\vartheta \leq (64(d+2))^{-1}$ and $\caL = \frac{1}{B}\mathbb{Z}^n$ be a lattice such that $\eta_\vartheta(\caL) \leq (6d)^{-4}$.
    If \Call{DiscreteCharacterizationTest}{} fails with probability at most $2/3$, then for every 
    $\bm{p},\bm{q}\in\ball(\bm 0,\srad)\cap (d+1)!\caL$ such that $\bm p+i\bm q \in \ball(\bm 0, \srad) \cap (d+1)!\caL$ for every $i\in[d+1]$, we have 
     $\sum_{i=0}^{d+1}\alpha_i\cdot g(\bm{p}+i\bm{q})=0$.
\end{lem}
In the remainder of this subsection, we will prove \autoref{thm:characterization_in_small_ball_discrete}.



Let $\rho$ denote the bound of the  probability that each of the tests in the \Call{DiscreteCharacterizationTest}{} fails.
That is, for every $j\in \{1,\dots,d+1\}$ and $t\in \{0,\dots,d+1\}$, the following are bounded: 

  \begin{align}
      \Pr_{\substack{\bm{p}\sim\caG(j\caL,j\sqrt{t^2+1})\\ \bm{q}\sim\caG(\caL,1)}}\left[\sum_{i=0}^{d+1}\alpha_i\cdot f(\bm{p}+i\bm{q})\neq 0\right]&\leq \rho. && \text{[$j^2(t^2+1)$ vs. $1$ Test.]} \label{eq:test_bound_1_discrete}\\
      \Pr_{\substack{\bm{p}\sim\caG(j\caL,j)\\ \bm{q}\sim\caG(\caL,\sqrt{t^2+1})}}\left[\sum_{i=0}^{d+1}\alpha_i\cdot f(\bm{p}+i\bm{q})\neq 0\right]&\leq \rho. && \text{[$j^2$ vs. $(t^2+1)$ Test.]}\label{eq:test_bound_2_discrete}\\
      \Pr_{\bm p,\bm q\sim\caG(j\caL,j)}\left[\sum_{i=0}^{d+1}\alpha_i\cdot f(\bm{p}+i\bm{q})\neq 0\right]&\leq \rho. && \text{[$j^2$ vs. $j^2$ Test.]}\label{eq:test_bound_3_discrete}
  \end{align}

We first bound $\rho$, by an identical argument, used earlier in \autoref{claim:rho-is-small} to bound $\rho$ in the exact case:
\begin{claim}\label{claim:rho-is-small-discrete}
If \Call{DiscreteCharacterizationTest}{} fails with probability at most $2/3$, then $\rho\leq(4d)^{-2}$.
\end{claim}
Then we bound the probability that $g_{\bm q_1}$ and $g_{\bm q_2}$ differ, in the intersection of $\ball(\bm 0,\srad)$ with  $(d+1)!\caL$:

\begin{lem}\label{lem:well_defined_of_g_discrete}
    Let $\caL$ denote the lattice $\frac1B \Z^n$. 
    Let $\eta_\vartheta \in (0,1/2)$, and $r \leq d \sqrt{\eta_\vartheta(\caL)}$. Then, for every $\bm{p}\in \ball(\bm{0},\srad)\cap (d+1)!\caL$ and every $t\in S\triangleq\{\sqrt{i^2+1}:i\in \{0,\ldots, d+1\}\}$,
   \[
   \Pr_{\substack{\bm{q}_1\sim \caG(\caL, t)\\ \bm{q}_2\sim\caG(\caL, 1)}}\left[g_{\bm{q}_1}(\bm{p})\neq g_{\bm{q}_2}(\bm{p})\right]\leq 2(d+1) \left(\rho + 196d^2 \eta_\vartheta(\caL) \right). 
   \] 
\end{lem}
\begin{proof}
We follow the proof of \autoref{lem:well_defined_of_g}. Fix $\bm p$, $t$ as in the statement of the lemma. For $i \in [d+1]$, will bound the following probability
\begin{align*}
	&\Pr_{\substack{\bm{q}_1\sim\caG(\caL, t)\\ \bm{q}_2\sim\caG(\caL, 1)}} \Big[ f(\underbrace{\bm{p}+i \bm{q}_1}_{\triangleq\bm m}) \neq g_{\bm{q}_2}(\bm{p}+i\bm{q}_1) \Big]=\Pr_{\substack{\bm{m}\sim \caG(\bm p+i\caL, it) \\ \bm{q}_2\sim\caG(\caL, 1)}} \Big[ f(\bm m) \neq g_{\bm{q}_2}(\bm m) \Big], \\
	\leq & \Pr_{\substack{\bm{m}\sim \caG(i\caL, it) \\ \bm{q}_2\sim\caG(\caL, 1)}} \Bigg[ \sum_{j=0}^{d+1}\alpha_i\cdot f(\bm m+j\bm q_2)\neq 0\Bigg] + 2\dtv( \caG(\bm p+i\caL, it), \caG(i\caL, it)),\tag{By definition of $g_{\bm q_2}(\bm m)$} \\
	\leq &\; \rho + 196d^2 \eta_\vartheta(\caL),
\end{align*}
where the bound on the second term follows from \autoref{cor:discrete_gaussian_tvd_bound},
and the first term is bounded by the rejection probability $\rho$, by \eqref{eq:test_bound_1_discrete}.
Note that we can indeed apply \autoref{cor:discrete_gaussian_tvd_bound}, since $\bm p \in (d+1)!\caL$ guarantees that $\bm p \in i\caL$ for every $i \in [d+1]$.

By the same argument as above, with \eqref{eq:test_bound_1_discrete} replaced by \eqref{eq:test_bound_2_discrete}, for any $j \in [d+1]$, we can bound 
\begin{align*}
&\Pr_{\substack{\bm{q}_1\sim\caG(\caL,t)\\ \bm{q}_2\sim\caG(\caL,1)}} \Big[ f(\underbrace{\bm{p}+j\bm{q}_2}_{\triangleq\bm m}) \neq g_{\bm{q}_1}(\bm{p}+j\bm{q}_2) \Big]=\Pr_{\substack{\bm m \sim \caG(\bm p+j\caL, j)\\ \bm{q}_1\sim\caG(\caL,t) }} \Big[ f(\bm m) \neq g_{\bm{q}_1}(\bm m) \Big], \\
\leq &\Pr_{\substack{\bm m \sim \caG(j\caL, j)\\ \bm{q}_1\sim\caG(\caL,t) }} \left[ \sum_{i=0}^{d+1}\alpha_i\cdot f(\bm m+i\bm q_1)\neq 0 \right] + 2\dtv(\caG(\bm p+j\caL, j), \caG(j\caL, j)), \tag{By definition of $g_{\bm q_1}(\bm m)$} \\
\leq &\; \rho + 196d^2 \eta_\vartheta(\caL).
\end{align*}

Taking a union bound over all $i,j \in [d+1]$ gives 
\begin{align*}
	\Pr_{\substack{\bm{q}_1\sim\caG(\caL,t)\\ \bm{q}_2\sim\caG(\caL,1)}} \Bigg[ \underbrace{\sum_{i=1}^{d+1} \alpha_i  \cdot f(\bm{p}+i\bm{q}_1)}_{=g_{\bm{q}_1}(\bm{p})} \neq \sum_{i=1}^{d+1} \sum_{j=1}^{d+1} \alpha_i \alpha_j  \cdot f((\bm{p}+i\bm{q}_1) + j\bm{q}_2) \Bigg] \leq (d+1)\left(\rho + 196d^2 \eta_\vartheta(\caL) \right), \\
	\Pr_{\substack{\bm{q}_1\sim\caG(\caL,t)\\ \bm{q}_2\sim\caG(\caL,1)}} \Bigg[ \underbrace{\sum_{j=1}^{d+1} \alpha_j \cdot f(\bm{p}+j\bm{q}_2)}_{=g_{\bm{q}_2}(\bm{p})} \neq \sum_{j=1}^{d+1} \sum_{i=1}^{d+1} \alpha_i \alpha_j \cdot f((\bm{p}+j\bm{q}_2) + i\bm{q}_1) \Bigg] \leq (d+1) \left(\rho + 196d^2 \eta_\vartheta(\caL) \right).
\end{align*}
Thus, by a final union bound, we can conclude that
\[ \Pr_{\substack{\bm{q}_1\sim\caG(\caL,t)\\ \bm{q}_2\sim\caG(\caL,1)}} [g_{\bm{q}_1}(\bm{p}) \neq g_{\bm{q}_2}(\bm{p})] \leq 2(d+1) \left(\rho + 196d^2 \eta_\vartheta(\caL) \right). \qedhere \]
\end{proof}

An immediate corollary is the following.

\begin{cor}\label{cor:replacement-is-good-discrete}
    If \Call{DiscreteCharacterizationTest}{} fails with probability at most  $2/3$, and $\eta_\vartheta(\caL) \leq (6d)^{-4}$, then for every $\bm p \in \ball(\bm 0,\srad) \cap \frac{(d+1)!}{B}\Z^n$ and every  $t\in \{0,\ldots, d+1\}$,
    \[ \Pr_{\bm{q}\sim\caG\left(\caL,\sqrt{t^2+1}\right)}[g(\bm p)\neq g_{\bm{q}}(\bm{p})] < \frac{1}{4(d+2)}. \]
\end{cor}
\begin{proof}By \autoref{claim:rho-is-small-discrete}, $\rho$ is at most $(4d)^{-2}$. Observe, for any $t \in \{0,\ldots, d+1\}$,
     \begin{align*} \Pr_{\bm{q}\sim\caG\left(\caL,\sqrt{t^2+1}\right)}[g(\bm p)\neq g_{\bm{q}}(\bm{p})] &\leq \Pr_{\bm q_1 \sim \caG(\caL, 1)}[g(\bm p) \neq g_{\bm q_1}(\bm p)] + \Pr_{\substack{\bm{q}\sim\caG\left(\caL,\sqrt{t^2+1}\right)\\ \bm q_1 \sim \caG(\caL, 1)}}[g_{\bm q}(\bm p) \neq g_{\bm q_1}(\bm p)] \\
     &\leq 4(d+1) \left(\rho + 196d^2 \eta_\vartheta(\caL) \right), \tag{By \autoref{lem:well_defined_of_g_discrete}}
     \end{align*}
     which is at most $(4(d+2))^{-1}$, since $\rho\leq(4d)^{-2}$, and by our assumptions on  $\eta_\vartheta(\caL)$.
\end{proof}

Finally, we are ready to prove \autoref{thm:characterization_in_small_ball_discrete}, the discrete analog of \autoref{thm:characterization_in_small_ball}.

\begin{proof}[Proof of \autoref{thm:characterization_in_small_ball_discrete}]By \autoref{claim:rho-is-small-discrete}, $\rho$ is at most $(4d)^{-2}$.
By the same argument as in the proof of \autoref{thm:characterization_in_small_ball}, it is sufficient to show that for $\bm q_1,\bm q_2\sim\caG(\caL,1)$, the following two events hold simultaneously with non-zero probability, for every $\bm{p},\bm{q}\in\ball(\bm 0,\srad)\cap (d+1)!\caL$ such that $\bm p+i\bm q \in \ball(\bm 0, \srad) \cap (d+1)!\caL$ for every $i\in[d+1]$: 
\begin{equation}\label{eq:replacement-is-good-disc}
    \sum_{i=0}^{d+1} \alpha_i \cdot g(\bm{p}+i\bm{q}) =\sum_{i=0}^{d+1} \alpha_i\cdot g_{\bm{q}_1+i\bm{q}_2}(\bm{p}+i\bm{q})
\end{equation}
\begin{equation}\label{eq:post-replacement-is-good-disc}
    \sum_{i=0}^{d+1} \alpha_i \cdot f(\bm{p}+ j\bm{q}_1+ i(\bm{q}+j\bm{q}_2)) =0\text{, for every }j \in [d+1].
\end{equation}

We begin with~\eqref{eq:replacement-is-good-disc}:
{\allowdisplaybreaks
\begin{align}
    &\Pr_{\bm{q}_1,\bm{q}_2 \sim \caG(\caL,1)} \left[\sum_{i=0}^{d+1} \alpha_i \cdot g(\bm{p}+i\bm{q}) =\sum_{i=0}^{d+1}\alpha_i\cdot g_{\bm{q}_1+i\bm{q}_2}(\bm{p}+i\bm{q})\right] \nonumber	\\ 
	\geq & \Pr_{\bm{q}_1,\bm{q}_2 \sim \caG(\caL,1)} \left[g(\bm{p}+i\bm{q}) = g_{\bm{q}_1+i\bm{q}_2}(\bm{p}+i\bm{q}),~~ \forall i \in \{0,\ldots, d+1\} \right] \nonumber \\
	= &1-\Pr_{\bm{q}_1,\bm{q}_2\sim\caG(\caL,1)} \left[\exists i\in\{0,\ldots,d+1\}:g(\bm{p}+i\bm{q}) \neq g_{\bm{q}_1+i\bm{q}_2}(\bm{p}+i\bm{q}) \right] \nonumber \\
	\geq & 1-\sum_{i=0}^{d+1} \Pr_{\bm{q}_1,\bm{q}_2 \sim \caG(\caL,1)}\left[g(\bm{p}+i\bm{q}) \neq g_{\bm{q}_1+i\bm{q}_2}(\bm{p}+i\bm{q}) \right] \nonumber  \\
	= & 1-\sum_{i=0}^{d+1}\bigg(\Pr_{\bm{m}\sim\caG\left(\caL,\sqrt{i^2+1}\right)}\left[g(\bm{p}+i\bm{q})\neq g_{\bm{m}}(\bm{p}+i\bm{q}) \right] +2\dtv\left(\bm q_1 + i \bm q_2,\caG\left(\caL,\sqrt{i^2+1}\right)\right)\bigg). \label{eq:last_g_equals_0} 
\end{align}}
Since $\bm{p}+i\bm{q}\in \ball(\bm{0},\srad)\cap (d+1)!\caL$, the probability that $g(\bm{p}+i\bm{q})\neq g_{\bm{m}}(\bm{p}+i\bm{q})$ is at most  $(4(d+2))^{-1}$, by \autoref{cor:replacement-is-good-discrete}. As well, we can bound the total variation distance by $8 \vartheta$, by applying \autoref{lem:indsum} with parameters $k=2$, $\bm z= (1,i), \bm s = (1,1)$, and noting that $s_i =1  \geq \sqrt{2}i \eta_\vartheta(\caL)$. Thus,~\eqref{eq:last_g_equals_0} is at least 
\begin{align*}
     1-\sum_{i=0}^{d+1} \bigg( \frac{1}{4(d+2)} + 16\vartheta \bigg) \geq 1- \bigg(\frac{1}{4} + 16(d+2)\vartheta \bigg) \geq \frac{1}{2}.
\end{align*}

Next, we bound~\eqref{eq:post-replacement-is-good-disc}. Fix some $j \in [d+1]$, then 
\begin{align*}
    &\Pr_{\bm{q}_1,\bm{q}_2 \sim \caG(\caL,1)} \bigg[ \sum_{i=0}^{d+1} \alpha_i \cdot f(\bm{p}+j\bm{q}_1+i (\bm{q}+j\bm{q}_2 )) \neq 0 \bigg] \tag{Let $\bm z_1\triangleq\bm p+j\bm q_1,\bm z_2\triangleq\bm q+j\bm q_2$}\\
    \leq &\Pr_{\substack{\bm{z}_1 \sim \caG(\bm p+j\caL,j) \\ \bm{z}_2 \sim \caG(\bm q+j\caL,j)}} \bigg[ \sum_{i=0}^{d+1} \alpha_i \cdot f(\bm{z}_1 + i \bm{z}_2) \neq 0\bigg] +2\Big(\underbrace{\dtv (j\bm q_1,\caG (j\caL,j ) )}_{= 0\text{, by \autoref{fact:discrete-int-multiple}}} + \underbrace{\dtv (j\bm q_2,\caG (j\caL,j ) )}_{= 0\text{, by \autoref{fact:discrete-int-multiple} }}\Big) \\
    \leq & \Pr_{\substack{\bm{z}_1 \sim \caG(j\caL,j) \\ \bm{z}_2 \sim \caG(j\caL,j)}} \bigg[ \sum_{i=0}^{d+1} \alpha_i f(\bm{z}_1 + i \bm{z}_2) \neq 0\bigg] + 2 (\dtv(\caG(j\caL,j),\caG(\bm p+j\caL,j)) + \dtv(\caG(j\caL,j),\caG(\bm q+j\caL,j)) )  \\
    \leq & \Pr_{\substack{\bm{z}_1 \sim \caG(j\caL,j) \\ \bm{z}_2 \sim \caG(j\caL,j)}} \bigg[ \sum_{i=0}^{d+1} \alpha_i \cdot f(\bm{z}_1 + i \bm{z}_2) \neq 0\bigg] + 392d^2 \eta_\vartheta(\caL) \tag{By \autoref{cor:discrete_gaussian_tvd_bound}}\\
    \leq &\; \rho +392d^2 \eta_\vartheta(\caL). \tag{By \eqref{eq:test_bound_3_discrete}}
\end{align*}
Finally, by a union bound over all $j \in [d+1]$, the probability of event~\eqref{eq:post-replacement-is-good-disc} can be lower bounded,
\[ \Pr_{\bm{q}_1,\bm{q}_2 \sim \caG(\caL,1)} \Big[ \forall j \in [d+1], \sum_{i=0}^{d+1} \alpha_i \cdot f(\bm{p}+j\bm{q}_1 + i(\bm{q}+j\bm{q}_2)) =0 \Big] \geq 1-(d+1)(\rho +392d^2 \eta_\vartheta(\caL)) \geq 1- o(1),\]
where the last inequality follows by our assumptions on $\rho$ and $\eta_\vartheta(\caL)$.

A final union bound shows that~\eqref{eq:replacement-is-good-disc} and~\eqref{eq:post-replacement-is-good-disc} hold simultaneously with non-zero probability, and the theorem follows. 
\end{proof}

\subsection{Polynomial Within a Hypercube}

Next, we obtain the discrete analog of \autoref{lem:g-on-lines-in-B-is-poly}. We argue that if the conditions of \autoref{lem:g-on-lines-in-B-is-poly-discrete} are met, then $g$ is consistent with a degree-$d$ polynomial on a hypercube $[-r',r']^n \subseteq \ball(\bm 0, \srad)$; we take $r'= r/\sqrt{n}$ as this is a large hypercube which can be inscribed within the cube such that no point in $[-r',r']^n$ is extremal on the cube, however other values of $r'$ work as well. 
This is done in two steps: in \autoref{lem:g-on-lines-in-B-is-poly-discrete} we argue that $g$ is consistent with a polynomial of degree at most $dn$, and in \autoref{lem:radial_lines_consistent_force_low_degree_discrete} we reduce the degree to $d$. As before, let $\bm e_i$ denote the $i$th standard basis vector. 

\begin{lem}\label{lem:local_to_global_discrete} Assume that the assumptions of \autoref{lem:g-on-lines-in-B-is-poly-discrete} hold.
    Let $r'= r/\sqrt{n}$, let $d,B > 0$ satisfy $2rB/\sqrt{n}>d!$ and let $h \colon {[-r',r']}^n \to \mathbb{R}$. Then, the following holds:
     If for every $i \in [n]$ and $\bm a \in {[-r',r']}^n$ such that $\bm a_i = 0$, the restriction of $h$ to the line segment $L_{\bm a, \bm e_i}$, the univariate function $h_{\bm a, \bm e_i}$ is consistent with a degree-$d$ univariate polynomial on every input $x$ for which $\bm a+ x \bm e_i \in {[-r',r']}^n \cap \frac{(d+1)!}{B}\Z^n$, then $h$ agrees with an $n$-variate polynomial of degree at most $dn$
    on ${[-r',r']}^n \cap \frac{(d+1)!}{B}\Z^n$.
\end{lem}
\begin{proof}
Let $\cube\triangleq {[-r',r']}^n \cap \frac{(d+1)!}{B}\Z^n$. We will prove the lemma by induction on the dimension $n$; the case $n=1$ is immediate. 
Now, assume that the statement is true for $n-1$. 
Let $c_1,\dots,c_{d+1}\in  [-r',r']\cap\frac{(d+1)!}{B}\Z$, be $d+1$ distinct values; note that these exist since  $\frac{2r'B}{(d+1)!}>d+1$. For each $i\in \{1,\dots,d+1\}$
 consider the   sub-cubes $\cube_1,\dots, \cube_{d+1}$ of dimension $n-1$, defined as $\cube_i\triangleq {[-r',r']}^{n-1}\cap \frac{(d+1)!}{B}\Z^{n-1} \times c_i$. Note that all the line segments in $\cube_i$ are also contained in $\cube$ and therefore, by assumption, $h$ is consistent with degree-$d$ \emph{univariate} polynomials on them.
 Thus, we can apply the induction hypothesis to argue that $h$ is consistent with degree-$d(n-1)$ polynomials $h_i$ on each of the sub-cubes $\cube_i$. We will combine these polynomials to form an $n$-variate polynomial using the following degree-$d$ polynomials: For every $i\in[d+1],\delta_i$'s are defined as,
\[\delta_i(c_j) \triangleq \begin{cases}1 & i=j,\\0 &i\neq j.\end{cases}\] 
We argue that 
\[h(\bm x) = \sum_{i=1}^{d+1}\underbrace{\delta_i(x_n)}_{\deg=d}\cdot\underbrace{h_i(x_1,\dots,x_{n-1})}_{\deg=d(n-1)}\] on the lattice points inside the cube.
Fix $\bm a\in {[-r',r']}^{n-1}\cap \frac{(d+1)!}{B}\Z^{n-1}$
and let $t$ be a formal variable. We claim that the following two \emph{univariate} polynomials are equal, 
\[h(a_1,\dots,a_{n-1},t)=\sum_{i=1}^{d+1}\delta_i(t)h_i(\bm a).\]
The left polynomial is of degree-$d$ by assumption, while the right polynomial is of degree-$d$ by definition of $\delta_i$ and the fact that $\bm a$ is fixed. Moreover, these two polynomial agree on $d+1$ points ${(a_1,\dots,a_{n-1},i)}_{i\in [d+1]}$, and therefore they are equal.

As equality holds for every such $(\bm a, t) \in [-r',r']^n$, $h$ is a polynomial of degree at most $dn$ within $[-r',r']^n$. 
\end{proof}

Next, we argue that the degree of $h$ is in fact at most $d$.

\begin{lem}\label{lem:radial_lines_consistent_force_low_degree_discrete}
Assume that the assumptions of \autoref{lem:g-on-lines-in-B-is-poly-discrete} hold. Let $r' = r/\sqrt{n}$, let $d,B,m > 0$ satisfy  ${\left(\frac{2r'B}{(d+1)!(m+1)}\right)}^n > n^m$, and let $h:{[-r',r']}^n\to\mathbb{R}$ be a polynomial of finite degree $m$. 
If for every radial line in the cube ${[-r',r']}^n$, the restriction of $h$ to that line agrees with a univariate polynomial of degree at most $d$ on points in $\frac{(d+1)!}{B}\Z^n$, then $m\leq d$.
\end{lem}
\begin{proof}
Consider the coefficient representation of the  polynomial $h(x\bm b)$ in the formal variables $x\in \R$ and $\bm b\in \R^n$. 
This representation is a polynomial of degree $m$ in both $x$ and $\bm b$, as $h$ is degree $m$. 
Consider $\alpha_m$, the  coefficient of the monomial with the highest degree (in $x$) in $h$ as a polynomial in the formal variables $\bm b$.
\[\alpha_m(\bm b) \triangleq \sum_{i_1+\cdots+i_n=m}\alpha_{i_1,\dots,i_n}\prod_{j=1}^{n}b_j^{i_j},\]
and note that $\alpha_m \neq 0$, as otherwise $h$ would have degree less than $m$.

Now, fix $\bm b$ to some value $\bm b^* \in {[-\frac{r'}{m+1},\frac{r'}{m+1}]}^n\cap \frac{(d+1)!}{B}\Z^n $ such that $\alpha_m(\bm b^*) \neq 0$;
 such a point exists since there are at most $n^m$ roots of the polynomial $\alpha_m$ and, by assumption, there are at least \[{\left(\frac{2r'B}{(d+1)!(m+1)}\right)}^n > n^m\]
 many lattice points in the cube ${[-\frac{r'}{m+1},\frac{r'}{m+1}]}^n$, and at least $m+1$ lattice points on the line segment $\pline_{\bm 0,\bm b}\cap {[-r',r']}^n$.
The  univariate polynomial $h(x\bm b^*)$ in the formal variable $x$ is of degree exactly $m$. By assumption, it is consistent with a  univariate polynomial of degree at most $d$ on all points in $\pline_{\bm 0, \bm b^*}\cap \frac{(d+1)!}{B}\Z^n$. Since there are more than $m+1$ points on the line segment, these two polynomials are identical and therefore $m\leq d$. 
\end{proof}

\subsection{Global Polynomial Representation}

Finally, we argue that if the conditions of \autoref{lem:g-on-lines-in-B-is-poly-discrete} are met, then $g$ is a degree-$d$ polynomial.

\begin{thm}\label{thm:g-is-mult-d-poly-discrete} Suppose that the assumptions of \autoref{lem:g-on-lines-in-B-is-poly-discrete} hold, and let $R>r>0$, and let $d,B > 0$ satisfy  $(rB/(d+1)!n^{3/2+d})^n > 1$. If \Call{DiscreteCharacterizationTest}{} fails with probability at most $2/3$, then $g$ is a degree-$d,n$-variate polynomial on the lattice $\caL' \triangleq \frac{(d+1)R}{rB} \mathbb{Z}^n$ within $\ball(\bm 0, R)$.
\end{thm}

\begin{proof}
The proof is identical to the proof of  \autoref{thm:g-is-mult-d-poly}, using \autoref{lem:local_to_global_discrete} and \autoref{lem:g-on-lines-in-B-is-poly-discrete}, noting that our choice of $r,B,d$ satisfies the hypothesis of \autoref{lem:radial_lines_consistent_force_low_degree_discrete}. Choosing $\caL'$ to be a coarser lattice than $\caL$ guarantees that for every $\bm p \in \ball(\bm 0,R)$, there are at least $d+1$ points on the line $\pline_{\bm 0, \bm p} \cap B(\bm 0,R)$ on which to define the value of $g(\bm p)$, and therefore $g$ is well defined on $\caL'$ within $\ball (\bm 0,R)$.
\end{proof}

Finally, we are ready to prove the main lemma, which concludes the proof of correctness for our tester. 

\begin{proof}[Proof of \autoref{lem:discrete_polynomial_lem}]
Suppose that \Call{DiscreteCharacterizationTest}{} fails with probability at most $2/3$. Then, by \autoref{thm:g-is-mult-d-poly-discrete}, $g$ is a degree-$d,n$-variate polynomial. It remains to bound the probability that $g(\bm p) \neq$ \Call{DiscreteQuery-$g$}{$\bm p$} for $\bm p \in \mathbb{R}^n\cap\caL \setminus \ball(0,r)$. If $f$ is itself a degree-$d$ polynomial, then \Call{DiscreteQuery-$g$}{$\bm p$} returns $g(\bm p)$ with probability $1$, so assume otherwise. To query $g$ on a point $\bm{p} \in \mathbb{R}^n\cap\caL$, \Call{DiscreteQuery-$g$}{$\bm p$} obtains $d+1$ points on the line segment $\pline_{\bm 0, \bm p}^{\ball}$ and then interpolate $g$ along this line. For each of these $d+1$ points $\bm s$, \Call{DiscreteQuery-$g$-InBall}{$\bm s$} samples an additional $N'_{\ref{alg:subroutines-discrete}} = O(\log (1/\varepsilon))$ points $\bm q_1,\ldots, \bm q_{N'_{\ref{alg:subroutines-discrete}}}\sim\caG(\caL,1)$, and checks whether  
\[\sum_{i \in [d+1]} \alpha_i \cdot f(\bm s + i \bm q_1) = \sum_{i \in [d+1]} \alpha_i \cdot f(\bm s + i \bm q_j),\]
for all $j \in [N'_{\ref{alg:subroutines-discrete}}]$; it rejects if any of these checks fail. This is equivalent to checking whether $g_{\bm q_1}(\bm s) \neq g_{\bm q_j}(\bm s)$; by \autoref{cor:replacement-is-good-discrete}, this occurs with probability at most $1/(4(d+2))$, since $\bm s\in \ball(\bm 0,\srad)\cap\caL$. The probability that this test returns an incorrect value is the probability that $g(\bm s) \neq g_{\bm q_1}(\bm s) = g_{\bm q_j}(\bm s)$ for every $\bm q_j$, which is at most ${(4(d+2))}^{-N'_{\ref{alg:subroutines-discrete}}} = 4^{-N'_{\ref{alg:subroutines-discrete}}}/{(d+2)}^{N'_{\ref{alg:subroutines-discrete}}} \leq \varepsilon/2(d+1)$, where the final inequality follows by choosing $N'_{\ref{alg:subroutines-discrete}} = O(\log (1/ \varepsilon))$. As \Call{DiscreteQuery-$g$}{$\bm p$} samples $d+1$ points, the probability that these points are all recovered successfully is at least $1-\varepsilon/2$. 

A similar argument holds for points $\bm p \in \ball(\bm 0, r) \cap \frac{1}{B}\mathbb{Z}^n$, and bounds the probability by  $1-\varepsilon/2$ as well. 
\end{proof}

%% file: appendix.tex

\section{A Characterization of Degree-\texorpdfstring{$d$}{d} Polynomials}
\label{sec:appendix}

In this appendix, we show how the \localCharacterization of degree-$d$ polynomials over $\mathbb{R}$ follows from known results. We begin with several definitions.

To connect finite forward differences of $f$ with derivatives of $f$, we consider the \emph{discrete differential operator}. Let $\bm t = (t_1,\ldots, t_{d+1}) \in \mathbb{R}^{d+1}$, where $t_i\neq 0$ for every $i\in[d+1]$, be a vector of length $d+1$, and for any $S \subseteq [d+1]$, denote $t_S \triangleq \sum_{j \in S} t_j$. The discrete differential operator is defined as 
\[D_{\bm{\bm{t}}}[f](x) = D_{(t_1,\hdots,t_{d+1})}[f](x)\triangleq\sum_{S\subseteq[d+1]}(-1)^{|S|}f\left(x+t_S\right).\]
Note that the derivative of $f$ is the limit of the corresponding differential; formally, 
\begin{equation}\label{eq:relating-derivative-and-differential}
    \frac{d^{d+1}}{dx^{d+1}}f(x)=\lim_{\bm{t}\rightarrow\bm{0}}\frac{D_{\bm t }[f](x)}{\prod_{i\in[d+1]}t_i}.
\end{equation}
Thus, we can obtain local information about the derivative by inspecting the discrete differential. We will indirectly evaluate the discrete differential by inspecting the forward finite difference $\Delta_{h}[f](x)$ (defined in \eqref{eq:higher-order-difference}). Indeed, for $\bm 1 = (1,\ldots 1) \in \mathbb{R}^{d+1}$, observe that for any $h \in \mathbb{R}$,
\begin{equation}\label{eq:relating-differential-and-difference}
    D_{h\cdot\bm 1}[f](x) = \sum_{S\subseteq[d+1]} (-1)^{|S|} f\left(x+h|S|\right) = \sum_{i=0}^{d+1} (-1)^{i}\binom{d+1}{i} f(x+hi) = (-1)^{d+1}\Delta_h^{(d+1)}[f](x).
\end{equation}
We will first state a useful result from \cite{Ciesielski1959SomePO}, and sketch a variant of another result from \cite{ALM03}.
\begin{thm}{(Theorem 1 of \cite{Ciesielski1959SomePO})}\label{thm:bounded-to-continuity}
Let $f:(a,b)\to\R$ be \emph{J-convex of the $d$-th order over $(a,b)$}, i.e. $\Delta_h^{(d+1)}[f](x)\geq 0$, for every $x$, and $h>0$ such that $a<x<x+(d+1)h<b$. If $f$ is bounded on a set $E\subset (a,b)$ of positive measure, then $f$ is continuous on the interval $(a,b)$.
\end{thm}
\begin{thm}{(A variant of Theorem 2 of \cite{ALM03}\footnote{See also \cite{Mck67, Ger71, AL07}.})}\label{thm:continuity-to-poly}
Let $f:(a,b)\to\R$ such that $\Delta_h^{(d+1)}[f](x)=0$, for every $x\in(a,b)$ and $h>0$, such that $a<x<x+(d+1)h<b$. If $f$ is continuous on a set $S\subset (a,b)$ of $d+1$ points, then $f$ is a degree-$d$ polynomial over $(a,b)$.
\end{thm}
\begin{proof}[Proof Sketch:]
The proof goes through by showing that for an arbitrary $\alpha\in(a,b)$, and a large enough $M\in\mathbb N, f$ agrees with a unique degree-$d$ polynomial $p$,  on a sequence of sets $\{\{\frac{i\alpha}{2^nM}\}_{i\in\Z}\cap(a,b)\}_{n\in\mathbb{N}}$, and then by the continuity of $f$ on $S$, $f$ and $p$ can be proven to be arbitrarily close on $S$, i.e. $f(x)=p(x)$, for every $x\in S$. With $|S|=d+1$, this proves $f$ is a degree-$d$ polynomial.  
\end{proof}
We are now ready to prove the \localCharacterization, restated next for convenience.
\localchar*
\begin{proof}
Since $\Delta_h^{d+1}[g](x)= 0$ for every $x\in(a,b)$ and sufficiently small $h>0$, such that $a<x<x+(d+1)h<b$, it follows that $g$ is J-convex of the $d$-th order over $(a,b)$. Since $g$ is bounded as well over $(a,b)$, by \autoref{thm:bounded-to-continuity} $g$ must be continuous over $(a,b)$. Now, invoking \autoref{thm:continuity-to-poly}, we thus claim $g$ is a degree-$d$ polynomial over $(a,b)$.
\end{proof}
\ignore{
\begin{proof}
Since $g$ is analytic \enote{bounded?}, for every $ x\in (c,d),g(x)$ converges to the Taylor series expansion of $g$ at some point in the local neighborhood of $x$; i.e., for every $x\in (c,d)$, there exists $\varepsilon\in\mathbb{R}_{\geq 0}$, and $x_0\in(x-\varepsilon,x+\varepsilon)$ such that $$g(x)=\sum_{n=0}^{\infty}\frac{d^{n}}{dx^n}g(x_0) \frac{(x-x_0)^{n}}{n!}.$$
By assumption, $\Delta_h^{(d+1)}[g](x)=0$ for all sufficiently small $h$. Thus, by \eqref{eq:relating-differential-and-difference} we have that   for all sufficiently small $h$ and every $x\in (c,d)$, \[D_{h\cdot\bm{1}}[g](x)=(-1)^{d+1}\Delta_{h}^{(d+1)}[g](x)=0.\]
By \eqref{eq:relating-derivative-and-differential}, for every $ x\in (c,d)$, 
\[\frac{d^{d+1}}{dx^{d+1}}g(x)=\lim_{\bm{t\rightarrow\bm{0}}}\frac{D_{\bm{t}}[g](x)}{\prod_{i\in[d+1]}t_i}=\frac{\lim_{\bm{t\rightarrow\bm{0}}}D_{\bm{t}}[g](x)}{\lim_{\bm{t\rightarrow\bm{0}}}\prod_{i\in[d+1]}t_i}=\frac{\lim_{h\rightarrow 0}D_{h\cdot\bm{1}}[g](x)}{\lim_{h\rightarrow 0}h^{d+1}}=\lim_{h\rightarrow 0}\frac{0}{h^{d+1}}=0.\]
Thus, for every $k\geq d+1$, and $x\in (c,d)$, we have $\frac{d^k}{dx^k} g(x)=0$.

Now, for any point $y\in (c,d)$, consider the Taylor series expansion of $g$, centered at $0$:
$$T(y)=\sum_{n=0}^{\infty}\frac{d^n}{dx^n}g(0) \frac{y^{n}}{n!}=\sum_{n=0}^{d}\frac{d^n}{dx^n}g(0) \frac{y^n}{n!}.$$
This give us a degree-$d$ Taylor polynomial, making the series finite, and hence convergent at every such point; i.e., $T(y)=g(y)$.  Thus,  $g(y)$ is a degree-$d$, univariate polynomial, for every $y\in (c,d)$. 
\end{proof}
}
Next, we prove the  \discreteLocalCharacterization which was used in \autoref{sec:discrete}.
\localchardiscrete*
\begin{proof}
 For each $i \in \mathbb{N}$, let $S_i\triangleq\{\frac{ia}{M},\hdots,\frac{(i+d)a}{M}\}\subset[0,a]$ and let $p_i(t):[0,a]\to\R$ be the unique degree-$d$ polynomial satisfying $p_i(s) =f(s)$ for all $s \in S_i$. We will argue that the polynomials $p_i$ are identical, and thus equal to $f$ on $S \cap [0,a]$. First, we will argue that $p_0 = p_{1}$. Observe that
 \[ 0=\Delta_{\frac{a}{M}}^{(d+1)}[f](0)=\sum_{j=0}^{d+1}\alpha_j\cdot f\left(\frac{ja}{M}\right)=\sum_{j=0}^{d}\alpha_j \cdot f\left(\frac{ia}{M}\right) + \alpha_{d+1} \cdot f\left(\frac{(d+1)a}{M}\right). \] 
 As $p_0$ was defined by interpolating the values of $f$ on $S_0$, we have
 \[0=\sum_{j=0}^{d}\alpha_j \cdot p_0\left(\frac{ja}{M}\right)+ \alpha_{d+1} \cdot  f\left(\frac{(d+1)a}{M}\right)=\underbrace{\sum_{j=0}^{d+1}\alpha_j \cdot p_0\left(\frac{ja}{M}\right)}_{=0}+ \alpha_{d+1} \cdot (f-p_0)\left(\frac{(d+1)a}{M}\right), \]
 which implies that $f((d+1)a/M) = p_0((d+1)a/M)$, and hence $p_0(t) = p_1(t)$ for every $t \in [0,a]$. Repeating this argument for every $i \in \mathbb{N}$, we can conclude that $p_0(t) = p_1(t) \ldots = f(t)$ for every $t \in S \cap [0,a]$, where the equality with $f$ follows because the $p_i$'s are defined by interpolating the values of $f$ on the $S_i$'s. Thus, $f$ agrees with a degree-$d$ polynomial on $S \cap [0,a]$.


\ignore{
Now consider the interval $(0,\frac{a}{M})$. We will prove $g$, on $(0,\frac{a}{M})$, agrees with its degree-$d$ polynomial representation on $S\cap[0,a]$, say $\sum_{i=0}^{d}c_ix^i$.
Since $g$ is analytic, for every $ x\in [0,a],g(x)$ converges to the Taylor series expansion of $g$ at some point in the local neighborhood of $x$; i.e., for every $x\in[0,a]$, there exists $\lambda\in\mathbb{R}_{> 0}$, and $x_0\in(x-\lambda,x+\lambda):$ $$g(x)=\sum_{n=0}^{\infty}\frac{d^{n}}{dx^n}g(x_0) \frac{(x-x_0)^{n}}{n!}.$$
Since $g(0)$, and $g(\frac{a}{M})$ agree with degree-$d$ polynomial representation, it must hold that $\frac{d^k}{dx^k}g(0)=0=\frac{d^k}{dx^k}g(\frac{a}{M})$, for every $k\geq n$.
Now, for any point $y\in (0,\frac{a}{M})$, consider the Taylor series expansion of $g$, centered at $0$:
$$T(y)=\sum_{n=0}^{\infty}\frac{d^n}{dx^n}g(0) \frac{y^{n}}{n!}=\sum_{n=0}^{d}\frac{d^n}{dx^n}g(0) \frac{y^n}{n!}.$$
This give us a degree-$d$ Taylor polynomial, making the series finite, and hence convergent at every such point; i.e., $T(y)=g(y)$. Thus, $g(y)$ agrees with the degree-$d$ polynomial representation, for every $y\in (0,\frac{a}{M})$. Repeating this argument for every interval of the form $(\frac{ia}{M},\frac{(i+1)a}{M})\subset(0,a)$,
we get that $g$ agrees with the degree-$d$ polynomial representation on $[0,a]$.
}
\end{proof}

\section{Proofs from \autoref{sec:poly-on-lines-ball-approx}}\label{sec:pushed-proofs}
In this appendix, we provide proofs of \autoref{lem:well_defined_of_g_approx}, and \autoref{thm:characterization_in_small_ball_approx}. But first we state the result from \cite{Gaj91}, that forms the basis of our \autoref{cor:gajda-cor}:
\begin{thm}{(\cite[Theorem 8]{Gaj91})}\label{thm:approx-characterization}
Let $X$ be a linear space over the rationals, $x_0\in\R,d\in\mathbb N,\phi,a\in(0,\infty)$, and suppose that for all $x\in(x_0-a,x_0+a)$ and  $h\in(-a,a)$, with $x+(d+1)h \in (x_0 -a,x_0+a)$, $f\colon(x_0-a,x_0+a)\to X$  satisfies 
\[|\Delta_h^{(d+1)}[f](x)|\leq\phi.\]
Then, there exists a degree-$d$ polynomial $g\colon\R\to X$, such that for every $x\in (x_0-a,x_0+a),|f(x)-g(x)|\leq l_d''\phi$, where $l_d''\triangleq n_d+2^{d+1}l_d'(2^{d+1}-1)^{n_d},n_d\triangleq\min\{k\in\mathbb{N}:(1+k/d)^k\geq d\},l_d'\triangleq\prod_{i=1}^d(k_i'+1), l_0'\triangleq 1$, and $k_i'\triangleq 3\cdot2^{2i}+(i-1)2^{i+1}-1$. In particular,  $l_d'=\Theta(3^d\cdot 2^{d^2+d})$, $n_d\in\Omega(\log d)\cap o(d)$, and $l_d''=o(2^{8d^2})$. 
\end{thm}
Interestingly, \cite{Gaj91} defines a function $g:(a,b)\to X$ to be a degree-$d$ polynomial, if $\Delta_h^{(d+1)}[g](x)=0$, for all $x\in(a,b)$ and $h>0$, such that $a<x<x+(d+1)h<b$. Note that if $f:(a,b)\to\R$ is bounded on $(a,b)$, then by \autoref{thm:approx-characterization} $g:\R\to\R$ is also bounded, and hence a degree-$d$ polynomial on the same interval, by the \localCharacterization, thus proving \autoref{cor:gajda-cor}. We now resume the proofs:
\gwelldefapprox*
\begin{proof}
    Fix some $t\in \{0,\ldots, d+1\}$ and $\bm{p}\in \ball(\bm{0},\srad)$. We will bound the probability that $g_{\bm{q}_1}(\bm{p})$ and $q_{\bm{q}_2}(\bm{p})$ are far from $\sum_{i=1}^{d+1}\sum_{j=1}^{d+1} \alpha_i \alpha_j \cdot f(\bm{p}+i\bm{q}_1+j\bm{q}_2)$; the lemma will then follow by a union bound.
    
    By definition, we have $g_{\bm{q}_2}(\bm{p}) = \sum_{i=1}^{d+1} \alpha_i \cdot f(\bm{p}+i\bm{q}_2)$. Fixing an $i \in [d+1]$, we get
\begin{align*}
	&\Pr_{\substack{\bm{q}_1\sim\mathcal{N}(\bm{0},(t^2+1)I)\\ \bm{q}_2\sim\mathcal{N}(\bm{0},I)}}\Bigg[ |f(\underbrace{\bm{p}+i \bm{q}_1}_{\triangleq\bm z}) - g_{\bm{q}_2}(\bm{p}+i\bm{q}_1)|>\delta \Bigg]
	= \Pr_{\substack{\bm z \sim \mathcal{N}(\bm p,i^2(t^2+1)I)\\
	\bm{q}_2 \sim \mathcal{N}(\bm{0},I)}} \Bigg[ \Big|f(\bm z) - \sum_{j=1}^{d+1} \alpha_j f(\bm{z} +j \bm{q}_2)\Big|>\delta \Bigg]  \\
	&\leq \Pr_{\substack{\bm{z}\sim\mathcal{N}(\bm{0},i^2(t^2+1)I)\\ \bm{q}_2\sim\mathcal{N}(\bm{0},I)}} \Bigg[ \Bigg|\sum_{j=0}^{d+1} \alpha_j \cdot f(\bm{z} +j \bm{q}_2)\Bigg|>\delta\Bigg] + 2\dtv(\mathcal{N}(\bm{0},i^2(t^2+1)I), \mathcal{N}(\bm{p},i^2(t^2+1)I))   \\
	&\leq \rho + i^2(t^2+1)\srad  \tag{By \eqref{eq:test_bound_1_approx} and \autoref{lem:usefulDTVBound}} \leq \rho+20d^4\srad.
\end{align*}

By a similar calculation, we have for every $j \in [d+1]$ that
\begin{align*}
	&\Pr_{\substack{\bm{q}_1\sim\mathcal{N}(\bm{0},(t^2+1)I)\\ \bm{q}_2\sim\mathcal{N}(\bm{0},I)}} \Bigg[ |f(\underbrace{\bm{p}+j\bm{q}_2}_{\triangleq\bm z}) - g_{\bm{q}_1}(\bm{p}+j\bm{q}_2)|>\delta \Bigg]
	=\Pr_{\substack{ \bm{z}\sim\mathcal{N}(\bm{p},j^2I)\\ \bm{q}_1\sim\mathcal{N}(\bm{0},(t^2+1)I)}}\Bigg[\Big|f(\bm{z})-\sum_{i=1}^{d+1}\alpha_i f(\bm{z}+i\bm{q}_1)\Big|>\delta\Bigg]\\
	&\leq\Pr_{\substack{ \bm{z}\sim\mathcal{N}(\bm{0},j^2I)\\ \bm{q}_1\sim\mathcal{N}(\bm{0},(t^2+1)I)}}\left[\Bigg|\sum_{i=0}^{d+1}\alpha_1\cdot f(\bm{z}+i\bm{q}_1)\Bigg|>\delta\right] + 2\dtv(\mathcal{N}(\bm{0},j^2I),\mathcal{N}(\bm{p},j^2I))\\
	&\leq \rho + j^2\srad \tag{By \eqref{eq:test_bound_2_approx} and \autoref{lem:usefulDTVBound}}\leq \rho+4d^2\srad.
\end{align*}

Taking a union bound over $i \in [d+1]$ and $j \in [d+1]$ respectively, it follows that
\begin{align*}
	\Pr_{\substack{\bm{q}_1\sim\mathcal{N}(\bm{0},(t^2+1)I)\\ \bm{q}_2\sim\mathcal{N}(\bm{0},I)}} \Bigg[\Bigg| \underbrace{\sum_{i=1}^{d+1} \alpha_i  \cdot f(\bm{p}+i\bm{q}_1)}_{=g_{\bm{q}_1}(\bm{p})} - \sum_{i=1}^{d+1} \sum_{j=1}^{d+1} \alpha_i \alpha_j  \cdot f((\bm{p}+i\bm{q}_1) + j\bm{q}_2)\Bigg|>2^{d+1}\delta \Bigg] \leq 2d\rho +40d^5\srad, \\
	\Pr_{\substack{\bm{q}_1\sim\mathcal{N}(\bm{0},(t^2+1)I)\\ \bm{q}_2\sim\mathcal{N}(\bm{0},I)}} \Bigg[\Bigg| \underbrace{\sum_{j=1}^{d+1} \alpha_j \cdot f(\bm{p}+j\bm{q}_2)}_{=g_{\bm{q}_2}(\bm{p})} - \sum_{j=1}^{d+1} \sum_{i=1}^{d+1} \alpha_i \alpha_j \cdot f((\bm{p}+j\bm{q}_2) + i\bm{q}_1)\Bigg|>2^{d+1}\delta \Bigg] \leq 2d\rho+8d^3\srad.
\end{align*}

Thus, by a union bound over the two previous inequalities we can conclude that 
\[ \Pr_{\substack{\bm{q}_1\sim\mathcal{N}(\bm{0},(t^2+1)I)\\ \bm{q}_2\sim\mathcal{N}(\bm{0},I)}} [|g_{\bm{q}_1}(\bm{p}) - g_{\bm{q}_2}(\bm{p})|>2^{d+2}\delta] \leq 4d\rho+48d^5\srad. \qedhere \]
\end{proof}

\smallballcharapprox*
\begin{proof}
Fix $\bm{p},\bm{q} \in \ball(\bm{0},\srad)$ and let  $h\in\mathbb{R}$ be sufficiently small so that $\bm p +ih \bm q \in \ball(\bm 0,\srad)$ for every $i\in[d+1]$; such $h$'s exist, as $\ball(\bm 0,\srad)$ is an open ball containing $\bm p$. We will argue that the following hold simultaneously with non-zero probability over $\bm{q}_1,\bm{q}_2 \sim \mathcal{N}(\bm{0},I)$:
\begin{equation}\label{eq:replacement-is-good_approx}
    \Big|\sum_{i=0}^{d+1} \alpha_i \cdot g(\bm{p}+ih\bm{q}) -\sum_{i=0}^{d+1} \alpha_i\cdot g_{\bm{q}_1+i\bm{q}_2}(\bm{p}+ih\bm{q})\Big|\leq 2^{2d+4}\delta, 
\end{equation}
\begin{equation}\label{eq:post-replacement-is-good_approx}
    \Big|\sum_{i=0}^{d+1} \alpha_i \cdot f(\bm{p}+ j\bm{q}_1+ i(h\bm{q}+j\bm{q}_2))\Big|\leq\delta\text{, for every }j \in [d+1].
\end{equation}
Assuming that these hold, we complete the proof. Fix any $\bm{q}_1,\bm{q}_2$ satisfying both~\eqref{eq:replacement-is-good_approx}, and~\eqref{eq:post-replacement-is-good_approx}. Then,
\begin{align*}
    \Big|\sum_{i=0}^{d+1} \alpha_i \cdot g(\bm{p}+ih\bm{q})\Big| &\leq \Big|\sum_{i=0}^{d+1} \alpha_i\cdot g_{\bm{q}_1+i\bm{q}_2}(\bm{p}+ ih\bm{q})\Big|+2^{2d+4}\delta \tag{By~\eqref{eq:replacement-is-good_approx}} \\
    &=\Big|\sum_{i=0}^{d+1} \alpha_i\Big( \sum_{j=1}^{d+1} \alpha_j \cdot f(\bm{p}+ih\bm{q}+j(\bm{q}_1+i\bm{q}_2))\Big)\Big|+2^{2d+4}\delta\tag{By definition of $g_{\bm q}(\bm p)$} \\
    &=\Big|\sum_{j=1}^{d+1} \alpha_j\Big( \sum_{i=0}^{d+1} \alpha_i \cdot f(\bm{p}+j\bm{q}_1+i(h\bm{q}+j\bm{q}_2))\Big)\Big|+2^{2d+4}\delta \\
    &\leq \Big|\sum_{j=1}^{d+1}\alpha_j \cdot \delta\Big|+2^{2d+4}\delta \tag{By~\eqref{eq:post-replacement-is-good_approx}}\\
    &\leq 2^{d+1}\delta+2^{2d+4}\delta=2^{d+1}(2^{d+3}+1)\delta\leq 2^{2d+5}\delta.
\end{align*}

 Next, we prove~\eqref{eq:replacement-is-good_approx} and~\eqref{eq:post-replacement-is-good_approx}, by arguing that each holds with positive probability and then taking a union bound. For~\eqref{eq:replacement-is-good_approx}, we observe
 {\allowdisplaybreaks
\begin{align*}
    &\Pr_{\bm{q}_1,\bm{q}_2 \sim \mathcal{N}(\bm{0},I)} \bigg[\Big|\sum_{i=0}^{d+1} \alpha_i \cdot g(\bm{p}+ih\bm{q}) -\sum_{i=0}^{d+1}\alpha_i\cdot g_{\bm{q}_1+i\bm{q}_2}(\bm{p}+ih\bm{q})\Big|\leq 2^{2d+4}\delta\bigg]	\\ 
	&\geq \Pr_{\bm{q}_1,\bm{q}_2 \sim \mathcal{N}(\bm{0},I)} \left[|g(\bm{p}+ih\bm{q}) - g_{\bm{q}_1+i\bm{q}_2}(\bm{p}+ih\bm{q})|\leq 2^{d+3}\delta,~~ \forall i \in \{0,\ldots, d+1\} \right] \\
	&=1-\Pr_{\bm{q}_1,\bm{q}_2\sim\mathcal{N}(\bm{0},I)} \left[\exists i\in\{0,\ldots,d+1\}:|g(\bm{p}+ih\bm{q}) - g_{\bm{q}_1+i\bm{q}_2}(\bm{p}+ih\bm{q})|>2^{d+3}\delta \right] \\
	&\geq 1-\sum_{i=0}^{d+1} \Pr_{\bm{q}_1,\bm{q}_2 \sim \mathcal{N}(\bm{0},I)}\left[|g(\bm{p}+ih\bm{q}) - g_{\bm{q}_1+i\bm{q}_2}(\bm{p}+ih\bm{q})|>2^{d+3}\delta \right] \tag{By union bound} \\
	&= 1-\sum_{i=0}^{d+1}\Pr_{\bm{m}\sim\mathcal{N}(\bm{0},(i^2+1)I)}\left[|g(\bm{p}+ih\bm{q})- g_{\bm{m}}(\bm{p}+ih\bm{q})|>2^{d+3}\delta \right] \tag{Letting $\bm{m}\triangleq\bm{q}_1+i\bm{q}_2$} \\
	&> 1-\sum_{i=0}^{d+1}\frac{1}{7d}=1-\frac{d+2}{7d}>\frac{1}{2}. \tag{Applying \autoref{cor:replacement-is-good_approx}, as $\bm{p}+ih\bm{q}\in \ball(\bm{0},\srad)$}
\end{align*}
}

For~\eqref{eq:post-replacement-is-good}, consider some $j \in [d+1]$, then
{\allowdisplaybreaks
\begin{align*}
    &\Pr_{\bm{q}_1,\bm{q}_2 \sim \mathcal{N}(\bm{0},I)} \bigg[ \Big|\sum_{i=0}^{d+1} \alpha_i \cdot f(\underbrace{\bm{p}+j\bm{q}_1}_{\triangleq\bm{z}_1}+i (\underbrace{h\bm{q}+j\bm{q}_2}_{\triangleq\bm{z}_2} ))\Big|>\delta \bigg] 
    = \Pr_{\substack{\bm{z}_1 \sim \mathcal{N}(\bm{p},j^2I) \\ \bm{z}_2 \sim \mathcal{N}(h\bm{q},j^2I)}} \bigg[ \Big|\sum_{i=0}^{d+1} \alpha_i\cdot f(\bm{z}_1 + i \bm{z}_2)\Big|>\delta\bigg] \\
    \leq & \Pr_{\substack{\bm{z}_1 \sim \mathcal{N}(\bm{0},j^2I) \\ \bm{z}_2 \sim \mathcal{N}(\bm{0},j^2I)}} \bigg[ \Big|\sum_{i=0}^{d+1} \alpha_i\cdot f(\bm{z}_1 + i \bm{z}_2)\Big|>\delta\bigg]\\
    +& 2 \left(\dtv\left(\mathcal{N}(\bm{0},j^2I), \mathcal{N}(\bm{p},j^2I)\right) + \dtv\left(\mathcal{N}(\bm{0},j^2I), \mathcal{N}(h\bm{q},j^2I) \right) \right) \\
    \leq & \Pr_{\bm{z}_1,\bm z_2 \sim \mathcal{N}(\bm{0},j^2I)} \bigg[ \Big|\sum_{i=0}^{d+1} \alpha_i \cdot f(\bm{z}_1 + i \bm{z}_2)\Big|>\delta\bigg] +j^2\srad + j^2 h\srad \tag{By \autoref{lem:usefulDTVBound}}\\
    \leq &\; \rho + j^2\srad +j^2h\srad < \rho + 8d^2\srad. \tag{By~\eqref{eq:test_bound_3_approx}}
\end{align*}}
By a union bound over all $j \in [d+1]$, 
\[ \Pr_{\bm{q}_1,\bm{q}_2 \sim \mathcal{N}(\bm{0},I)} \bigg[ \forall j \in [d+1], \Big|\sum_{i=0}^{d+1} \alpha_i \cdot f(\bm{p}+j\bm{q}_1 + i(h\bm{q}+j\bm{q}_2))\Big|\leq\delta \bigg] \geq 1-(2d\rho+16d^3\srad), \]
which is at least $ 2/3$, as $\rho\leq (30d)^{-2}$ by \autoref{claim:rho-is-small-approx}, and we set $r=(4d)^{-6}$ in \autoref{cor:replacement-is-good_approx}.
A final union bound concludes that both~\eqref{eq:replacement-is-good_approx} and~\eqref{eq:post-replacement-is-good_approx} hold simultaneously with non-zero probability.
\end{proof}

\section{Proof of \autoref{lem:indsum}}
\label{sec:AppendixC}

In this appendix we prove \autoref{lem:indsum}, which is an immediate consequence of the following lemma. This lemma follows in a straightforward manner from \cite[Theorem 3.3]{MP13} . 
 Denote by $\oplus$, the standard \emph{direct sum} of lattices and by $\otimes$, the standard \emph{tensor product}, and let $I_n$ denote the $n \times n$ identity matrix.

\begin{lem}
    Let $\caL\subseteq\R^n$ be a full rank lattice, and fix $k\in\Z_{>0}$, $\vartheta$, $s_1,\dots,s_k\in\R_{>0}$, and $\bm z\triangleq(z_1,\hdots,z_k)\in\Z^k$ such that $ s_i\geq\|\bm z\|_{\infty}\sqrt{2}\eta_\vartheta(\caL)$ for every $i\in[k]$. Let $\bm y_1,\ldots, \bm y_k$ be sampled independently from $\caG(\caL,s_i)$. Then, for $s\triangleq\sqrt{\sum_{i=1}^k(z_is_i)^2}$, the total variation distance between $\bm y \triangleq \sum_{i=1}^k z_i\bm y_i$ and $\caG\left(\gcd{(\bm z)}\caL,s \right)$ is given by 
    \[ \dtv \left(\bm y,\caG \left({\gcd{(\bm z)}\caL,s} \right)\right)=\frac{1}{2}\sum_{\bm y\in\gcd(\bm z)\caL}\left|\frac{\tau_{s}(\bm y)}{\tau_{s}(\gcd(\bm z)\caL)}-\frac{\tau_{s}(\bm y)\tau(L+\bm x)}{\tau(\caL')}\right|, \]
where the $nk$-dimensional lattice $\caL'\triangleq\bigoplus_{i=1}^k s_i^{-1}\caL=(\bm S \otimes I_n)^{-1}\caL^{\oplus k}$, for $\bm S\triangleq\mathrm{diag}(s_1,\hdots,s_k)$, and $L$ is the sublattice of $\caL'$ containing the elements which fall in the kernel of 
$\bm Z \triangleq (\bm z^\top \bm S) \otimes I_n$; that is,
$L \triangleq  \caL' \cap \ker(\bm Z)$. As well, $\bm x$ is the orthogonal projection of $\bm x'$ onto $\ker(\bm Z)$, where $\bm x'\in\caL':\bm Z\bm x'=\bm y$. Furthermore, $$\dtv \left(\bm y,\caG \left({\gcd{(\bm z)}\caL,s} \right)\right)\leq \frac{2k\vartheta} {1-2k\vartheta}.$$
    

\end{lem}
The proof of the first part follows by recording the parameters obtained in \cite{MP13}, while the proof for the second part is provided below:

\begin{proof}Let $\Re\colon\mathbb{C}\to\R$ denote the real part of a complex number.
From the proof of Lemma 4.1 in \cite{MR07}, we have that for every $\bm x\in\R^n$, lattice $\Lambda$, and $\vartheta >0$, such that $\eta_{\vartheta}(\Lambda)\leq 1$, 
\begin{align*}
    \underbrace{\tau(\Lambda+\bm x)}_{\geq 0}&=\underbrace{\det(\Lambda^*)}_{\geq 0}\Big(1+\sum_{\bm w\in \Lambda^*\setminus\{\bm 0\}}e^{2\pi i\langle \bm x,\bm w\rangle}\underbrace{\tau(\bm w)}_{\geq 0}\Big)\\
    &=\det(\Lambda^*)\Big(1+\sum_{\bm w\in \Lambda^*\setminus\{\bm 0\}}\underbrace{\Re(e^{2\pi i\langle \bm x,\bm w\rangle})}_{\in[-1,1]}\tau(\bm w)\Big)
    \in\det(\Lambda^*)\Big(1\pm\sum_{\bm w\in\caL^*\setminus\{\bm 0\}}\tau(\bm w)\Big) \\ 
    &=\det(\Lambda^*)\Big(1\pm\underbrace{\tau(\Lambda^*\setminus\{\bm 0\})}_{\leq\vartheta,\because \eta_{\vartheta}(\caL)\leq 1}\Big)=\det(\Lambda^*)(1\pm\vartheta).
\end{align*}
From the proof of Theorem 3.3 in \cite{MP13}, we have $\eta_{\vartheta'}(L)\leq\sqrt{2} \eta_{\vartheta}(\caL)/\min s_i\leq 1$, where $\vartheta'\triangleq(1+\vartheta)^{k-1}-1$. And for all $\vartheta,k$ such that $k \vartheta \in [0,1]$ we have $\vartheta'=(1+\vartheta)^{k-1}-1\leq 2k\vartheta$.
Also, note that for every $0\leq\vartheta_1\leq\vartheta_2$,  $\eta_{\vartheta_1}(\Lambda)\geq\eta_{\vartheta_2}(\Lambda)$ holds, for every lattice $\Lambda$. So, we can claim $\eta_{2k\vartheta}(L)\leq\eta_{\vartheta'}(L)\leq 1$, and hence we have that for every $\bm x\in\R^n,\tau(L+\bm x)\in\det(L^*)(1\pm 2k\vartheta)=\det(L^*)[1-2k\vartheta,1+2k\vartheta]$. Now observe that

\[\sum_{\bm y\in\gcd(\bm z)\caL}\underbrace{\frac{\tau_{s}(\bm y)}{\tau_{s}(\gcd(\bm z)\caL)}}_{\triangleq p(\bm y)}=1=\sum_{\bm y\in \gcd(\bm z)\caL}\underbrace{\frac{\tau_{s}(\bm y)\tau(L+\bm x)}{\tau(\caL')}}_{\triangleq q(\bm y)}=\sum_{\bm y\in \gcd(\bm z)\caL}\underbrace{p(\bm y)c\cdot(1+\vartheta(\bm y))}_{=q(\bm y)},\] 
where $c\triangleq \tau_{s}(\gcd(\bm z)\caL)\det(L^*)/ \tau(\caL')=\tau_{s}(\gcd(\bm z)\caL)/\tau(\caL')\det(L)$, and $\vartheta(\bm y)\in[-2k\vartheta,2k\vartheta]$ for every $\bm y\in \gcd(\bm z)\caL$. Thus, we have,
\begin{align*}
    \frac{1}{c}&=\sum_{\bm y\in\gcd(\bm z)\caL}p(\bm y)(1+\vartheta(\bm y))=\underbrace{\sum_{\bm y\in\gcd(\bm z)\caL}p(\bm y)}_{=1}+\underbrace{\sum_{\bm y\in\gcd(\bm z)\caL}p(\bm y)\vartheta(\bm y)}_{p(\bm y)\text{ is a PMF}}\\
    &=1+\underbrace{\mathbb{E}_{\bm y\sim_{p(\bm y)}\gcd(\bm z)\caL}[\vartheta(\bm y)]}_{\in[\min_{\bm y}\vartheta(\bm y),\max_{\bm y}\vartheta(\bm y)]}\in[1-2k\vartheta,1+2k\vartheta].
\end{align*}
Thus, $c \in [1/(1+2k \vartheta), 1/(1-2k \vartheta)]$ and $1-c \in [-2k \vartheta/(1-2k\vartheta), 2k \vartheta/(1+2k \vartheta)]$. It follows that 

\begin{align*}
    \dtv\left(\bm y,\caG\left({\gcd(\bm z)\caL,s}\right)\right)&=\frac{1}{2}\sum_{\bm y\in\gcd(\bm z)\caL}\left|\frac{\tau_{s}(\bm y)}{\tau_{s}(\gcd(\bm z)\caL)}-\frac{\tau_{s}(\bm y)\tau(L+\bm x)}{\tau(\caL')}\right|\leq\frac{1}{2}\max_{\bm y\in\gcd(\bm z)\caL}\left|1-\frac{q(\bm y)}{p(\bm y)}\right|\\
    &\leq\frac{1}{2}\max_{\bm y\in\gcd(\bm z)\caL}\max_c\left|1-c(1+\vartheta(\bm y))\right|\\
    &\leq\frac{1}{2}\left(\max_{c}|1-c|+\max_c c\cdot\max_{\bm y\in\gcd(\bm z)\caL}|\vartheta(\bm y)|\right)\\
    &=\frac{1}{2}\left(\frac{2k\vartheta}{1-2k\vartheta}+\frac{1}{1-2k\vartheta}2k\vartheta\right)=\frac{2k\vartheta}{1-2k\vartheta}. \qedhere
\end{align*}
\end{proof}

Finally, setting $k$ and $\vartheta$ such that $k\vartheta< 1/4$ proves \autoref{lem:indsum}.

%% file: tolerant_addativity.tex
\section{Distribution-Free Approximate Tester for  Additivity}\label{sec:additivity}

Recall that a function $f\colon\mathbb{R}^n \to \mathbb{R}$ is \emph{additive} if for every $\bm x,\bm y \in \mathbb{R}^n$, $f(\bm x+ \bm y) = f(\bm x) + f(\bm y)$; a function is \emph{linear} if it is both additive and for every $\alpha \in \mathbb{R}$ and $\bm x \in \mathbb{R}^n$, $f(\alpha \bm x) = \alpha f(\bm x)$. 
In this appendix we modify the additivity tester of \cite{FlemingY20} to be robust against noise. This gives us a tester for additivity with better error parameters than the approximate degree-$1$ tester obtained from \autoref{thm:approximate_low_degree_tester}. Formally, given query access to the input function $f: \R^n \to \R$, sampling access to unknown $(\varepsilon/4,R)$-concentrated distribution $\mathcal{D}$, and constants $0<\alpha\in \R$ and $0<\varepsilon \in \R$,  
 a \emph{distribution-free approximate tester} for additivity distinguishes between the following two cases with probability at least $2/3$:
\begin{itemize}
\item \textbf{\textsc{Yes Case:}} There exists an additive function $h: \R^n \to \R$ such that for all $\bm{p} \in \R^n$:
\[
|f(\bm{p})-h(\bm{p})| \leq \alpha;
\]
\item \textbf{\textsc{No Case:}}  For any additive function  $h: \R^n \to \R$:
\[
\Pr_{\bm{p} \sim \mathcal{D}}[|f(\bm{p})-h(\bm{p})|>21015 \cdot \lrad n^{1.5}\alpha] > \eps.
\]
\end{itemize}
We say that the tester has one-sided error if, for every $f$ satisfying the \textsc{Yes Case} , the tester always accepts, with probability $1$.

The main theorem of this section is the following.

\begin{restatable}{thm}{noisy}\label{thm:central-gausssian}
Let $\alpha, \varepsilon > 0$ and $\caD$ be an unknown $(R,\varepsilon/4)$-concentrated distribution. There exists a one-sided error, $O (\frac{1}{\varepsilon} \log\frac{1}{\varepsilon} )$-query for distinguishing between the case when $f$ is pointwise $\alpha$-close to some additive function and the case when, for every additive function $h$, $\Pr_{\bm p \sim \caD}[|f(\bm p)-h(\bm p)| > O(\lrad n^{1.5}\alpha) ] > \eps$.
\end{restatable}

The remainder of this section is organized as follows: In \autoref{sec:noisy_proof_idea}, first we describe several properties of additive functions which we will require for our tester, and  give an overview of the proof of \autoref{thm:central-gausssian}. Then, we present our tester under some constraints on the unknown $\mathcal{D}$, and give informal description of the proof technique.
In \autoref{subsec:proof_tester_exists} we prove  the main \autoref{thm:central-gausssian}, relying on our main \autoref{lem:g-is-well-defined}. \autoref{subsec:proof_of_g_additive_in_small_ball} is devoted to prove the main \autoref{lem:g-is-well-defined}.
Finally, in \autoref{subsec:distribution-free-additivity}, we  show that our tester is actually a multiplicative error distribution-free tester, without any assumption on the unknown distribution $\mathcal{D}$.   

\subsection{Proof Overview and \texorpdfstring{$\delta$}{d}-Additive Functions}
\label{sec:noisy_proof_idea}

We say that a function $f:\R^n\to\R$ is \emph{$\delta$-additive}, if for every $\bm x,\bm y\in \R^n$ it holds that 
\[|f(\bm x+\bm y)-f(\bm x)-f(\bm y)|\leq \delta.\]
Satisfying $\delta$-additivity implies that the following inequalities hold, which will be the basis for our tester. For every $\bm x,\bm y,\bm z\in \R^n$, assuming $f$ is $\alpha$-close to some additive function $h$, we have: 
\begin{align}
    |f(\bm x-\bm y)-f(\bm x)+f(\bm y)|\leq |\underbrace{h(\bm x-\bm y)-h(\bm x)-h(-\bm y)}_{=0}|+3\alpha &\leq \delta,\label{eq:diff-is-delta-additive}\\
    |f(\bm x)+f(-\bm x)|\leq |h(\bm x)+h(-\bm x)|+2\alpha=|h(\bm x)-h(\bm x)|+2\alpha &\leq \delta\\
    |f(\bm x-\bm y)-f(\bm x-\bm z)-f(\bm z-\bm y)|\leq |h(\bm x-\bm y)-(\underbrace{h(\bm x-\bm z)+h(\bm z-\bm y)}_{=h(\bm x-\bm y)})| +3\alpha &\leq\delta
\end{align}

Our tester (given in \autoref{alg:zero-mean-additivity} and \autoref{alg:subroutines_noisy}) follows the general outline given in the introduction for testing linearity. First, it tests whether $f$ satisfies $\delta$-additivity over a set of samples drawn from the distribution $\cN(\bm 0,I)$. If this test passes with sufficiently high probability then we able to show that $g$ --- a self-corrected function of $f$ on $\ball(\bm 0,\srad)$ ---  is $O(n^{1.5}\delta$)-close to an additive function $h:\R^n\to\R$, and furthermore, if $f$ is $\delta$-additive, then $f$ and $g$ (and therefore $f$ and $h$) are close. 
To do so, we crucially rely on the  following stability theorem for additive functions which follows from \cite[Theorem 2]{Kominek89}. 
\begin{thm}\label{thm:closeness_to_additivity_on_ball}
Let $\srad > 0$ and  $g:\ball(\bm 0,\srad)\to \R$. If $g$ is $\delta$-additive, then there exists an additive function $h:\R^n\to\R$, 
such that for every $\bm x\in \ball(\bm 0,\srad)$ 
  \[|g(\bm x)-h(\bm x)|\leq 5n^{1.5}\delta.\]
\end{thm}

Second, we show that by the way we have constructed $g$, we are able to approximate its value on points within $\ball(\bm 0, \srad)$ with high probability. Thus, for any point $\ball(\bm 0,\srad)$, we can estimate the distance between $f$ and $g$, and therefore between $f$ and $h$, the additive function which is close to $g$, given by \autoref{thm:closeness_to_additivity_on_ball}. 

For points $\bm p \not \in \ball(\bm 0, \srad)$, we map them to a point within $\ball(\bm 0, \srad)$ by dividing by a contraction factor $\kappa_{\bm p}$, defined as
\[ \kappa_{\bm p} \triangleq \begin{cases} 1 &\text{ if } \|\bm p \|_2 \leq \srad, \\ \left \lceil \frac{ \|\bm p\|_2 }{\srad}\right \rceil   &\text{ if } \|\bm p \|_2 > \srad. \end{cases}\]
 Then, we approximate $h$ on the corresponding point $\bm p/\kappa_{\bm p}$ inside the ball and map $h(\bm p/\kappa_{\bm p})$ back to $h(\bm p)$.


We are now ready to formally define $g$.

\paragraph{The Self-Corrected Function.} Let $r$ be a sufficiently small rational; $r \triangleq 1/50$ suffices. Define the value of the self-corrected function $g$ at a point $\bm{p}\in\ball(\bm 0,\srad)$ as the (weighted) median value of  $g_{\bm{x}}(\bm{p})\triangleq f(\bm{p}-\bm{x}) + f(\bm{x})$,  
each weighted according to its probability mass under $\bm{x} \sim \mathcal{N}(\bm 0,I)$. For points $\bm p$ outside of the ball, we project them into the ball by diving by a sufficiently large contraction factor that depends on the magnitude of $\bm p$. 

Concretely, $g:\mathbb{R}^nb \rightarrow \mathbb{R}$ is defined as follows
\[ g(\bm p) \triangleq \kappa_{\bm p} \cdot \mathop{\mathsf{med}}_{\bm x\sim\mathcal{N}(\bm 0,I)} \left[  g_{\bm x} \left( \frac{\bm p}{\kappa_{\bm p}} \right)  \right] = \kappa_{\bm p} \cdot \mathop{\mathsf{med}}_{\bm x\sim\mathcal{N}(\bm 0,I)} \left[  f \left( \frac{\bm p}{\kappa_{\bm p}}- \bm x \right) +f \left (\bm x \right)  \right]. \]

The intuition for using median is that it, in the case when $f$ is approximately additive, the median value should allow us to approximately correct the errors in $f$, and thus $g$ should be close to additive. We use the median here, rather than the majority, because the majority is more affected by outliers.

\subsection{Approximate Additivity Tester }\label{subsec:proof_tester_exists}
Our tester is given in \autoref{alg:zero-mean-additivity}, which uses subroutines given in \autoref{alg:subroutines_noisy}.\\
\begin{algorithm}[ht]
  \caption{Approximate Additivity Tester}\label{alg:zero-mean-additivity}
  \Procedure{\emph{\Call{ApproxAdditivityTester}{$f,\mathcal{D},\alpha,\eps,\lrad$}}}{
    \Given{Query access to $f\colon \mathbb{R}^n \to \mathbb{R}$, sampling access to an unknown $(\eps/4,\lrad)$-concentrated distribution $\mathcal{D}$, a noise parameter $\alpha>0$, and a farness parameter $\eps> 0$;}
    $\delta\gets 3\alpha$, $r\gets 1/50$\;
    \textbf{Reject} if \Call{TestAdditivity}{$f,\delta$} returns \textbf{Reject}\;
    \For{$N_{\ref{alg:zero-mean-additivity}} \gets O(1/\varepsilon)$ times}{
      Sample $\bm p \sim \mathcal{D}$\;
    \If{$\bm p\in\ball(\bm 0,\lrad)$}{
      \textbf{Reject} if $|f(\bm p) - $ \Call{Approximate-$g$}{$\bm p,f,\delta$}$|>5\delta n^{1.5}\kappa_{\bm p}$, or if \Call{Approximate-$g$}{$\bm p,f,\delta$} returns \textbf{Reject}.
    }}
    \textbf{Accept}.}
\end{algorithm}
\begin{algorithm}
  \caption{Additivity Subroutines}\label{alg:subroutines_noisy}
  \Procedure{\emph{\Call{TestAdditivity}{$f,\delta$}}}{
    \Given{Query access to $f\colon \mathbb{R}^n \to \mathbb{R}$, threshold parameter $\delta>0$;}
    \For{$N_{\ref{alg:subroutines_noisy}} \gets O(1)$ times}{
      Sample $\bm{x},\bm{y},\bm{z} \sim \cN(\bm 0, I)$\;
      \textbf{Reject} if $|f(-\bm x)+f(\bm x)|>\delta$\;
      \textbf{Reject} if $|f(\bm x - \bm y) - \left(f(\bm x) - f(\bm y)\right)|>\delta$\;
      \textbf{Reject} if $\left|f\left(\frac{\bm x - \bm y}{\sqrt{2}} \right) - \left( f \left(\frac{\bm x -\bm z}{\sqrt{2}} \right) + f \left(\frac{\bm z - \bm y}{\sqrt 2} \right)\right)\right|>\delta$\; 
    }
    \textbf{Accept}.
  }
  \Procedure{\emph{\Call{Approximate-$g$}{$\bm p, f,\delta$}}}{
    \Given{$\bm p \in \mathbb{R}^n$, query access to $f\colon \mathbb{R}^n \to \mathbb{R}$, threshold parameter $\delta>0$;}
    $N'_{\ref{alg:subroutines_noisy}} \gets O\left(\log \frac{1}{\varepsilon}\right)$\;
    Sample $\bm x_{1}, \ldots, \bm x_{N'_{\ref{alg:subroutines_noisy}}} \sim \mathcal{N}(\bm 0,I)$\;
    \textbf{Reject} if there exists $j \in [ N'_{\ref{alg:subroutines_noisy}}]$ such that $\left|\left(f(\bm p/\kappa_{\bm p} - \bm x_1) + f(\bm x_1)\right) -\left(f(\bm p/\kappa_{\bm p} - \bm x_j) + f(\bm x_j)\right)\right|>2\delta$\;
    \Return $\kappa_{\bm p} \left( f( \bm p/\kappa_{\bm p} - \bm x_1) + f(\bm x_1) \right)$.
  }
\end{algorithm}


The following lemma records the properties of $g$ that will be guaranteed by our tester. 

\begin{lem}\label{lem:g-is-additive}
  With $\srad\triangleq 1/50$, if \Call{TestAdditivity}{$f,\delta$} accepts with probability at least $1/3$, then $g$ is a $14\delta$-additive function inside the small ball $\ball(\bm 0,\srad)$, and furthermore, for every $\bm p\in\ball(\bm 0,\srad)$  it holds that 
  \[\Pr_{\bm x \sim \mathcal{N}(\bm 0,I)}[\left|g(\bm p) - \left( f(\bm p - \bm x) + f(\bm x)\right)\right|\geq 4\delta ] < 1/2.\]
\end{lem}

We prove \autoref{thm:central-gausssian} assuming that \autoref{lem:g-is-additive} holds.


\begin{proof}[Proof of \autoref{thm:central-gausssian}]
  First, observe that if $f$ is a noisy version of an additive function with noise bounded by $\alpha$, then $f$ is a $\delta$-additive function for $\delta=3\alpha$, and we claim \autoref{alg:zero-mean-additivity} always accepts. It is immediate that \Call{TestAdditivity}{$f$} always accepts. To see that $f$ also passes the remaining tests, observe that, since $f$ is $\alpha$-close to an additive function $h$, point-wise, we can claim:
  \begin{align*}
      \left|\kappa_{\bm p}  g_{\bm x}(\bm p/\kappa_{\bm p})  - f(\bm p)\right| &\leq \kappa_{\bm p}\left| g_{\bm x}(\bm p/\kappa_{\bm p}) - f(\bm p/\kappa_{\bm p})\right| + \left| \kappa_{\bm p}f(\bm p/\kappa_{\bm p}) - f(\bm p) \right|\\
      &= \kappa_{\bm p}\underbrace{\left| f(\bm p/\kappa_{\bm p} - \bm x) + f(\bm x) - f(\bm p/\kappa_{\bm p}) \right|}_{\leq\delta\text{, by }\eqref{eq:diff-is-delta-additive}} + \left| \kappa_{\bm p}f(\bm p/\kappa_{\bm p}) -h(\bm p) + h(\bm p)- f(\bm p) \right|\\
      &\leq \kappa_{\bm p}\delta + \left| \kappa_{\bm p}f(\bm p/\kappa_{\bm p}) -\kappa_{\bm p}h(\bm p/\kappa_{\bm p}) + h(\bm p) - f(\bm p) \right| \tag{as $h(\bm p)=\kappa_{\bm p}h(\bm p/\kappa_{\bm p})$}\\
      &\leq\delta\kappa_{\bm p} +\kappa_{\bm p}\underbrace{\left|f(\bm p/\kappa_{\bm p}) - h(\bm p/\kappa_{\bm p})\right|}_{<\alpha} + \underbrace{\left| h(\bm p)-f(\bm p)\right|}_{<\alpha}\\
      &\leq\delta\kappa_{\bm p} +\alpha\kappa_{\bm p} +\alpha<2\delta\kappa_{\bm p}.
  \end{align*}
  
  
  Note that \Call{Approximate-$g$}{$\bm p, f$} never rejects, because we have $\left|  g_{\bm x}(\bm p/\kappa_{\bm p}) - f(\bm p/\kappa_{\bm p})\right|\leq \delta$.
  Then, by the triangle inequality,
  \begin{align*}
      |g_{\bm x_i}(\bm p/\kappa_{\bm p})-g_{\bm x_j}(\bm p/\kappa_{\bm p})|\leq 2| g_{\bm x}(\bm p/\kappa_{\bm p})  - f(\bm p/\kappa_{\bm p})|\leq 2\delta.
  \end{align*}

  We now show that if $f$ is $\varepsilon$-far from all additive functions, then Algorithm~\ref{alg:zero-mean-additivity} rejects with probability at least $2/3$.
  If \Call{TestAdditivity}{$f$} accepts with probability at most $1/3$, we can reject $f$ with probability at least $2/3$.
  Hence, we assume that \Call{TestAdditivity}{$f$} accepts with probability at least $1/3$.
  Then by \autoref{lem:g-is-additive}, the function $g$ is $14\delta$-additive, inside $\ball(\bm 0,\srad)$. 
  Using \autoref{thm:closeness_to_additivity_on_ball}, there is an additive function $h:\R^n\to\R$,  which is 
  $70\delta n^{1.5}$-close to $g$, on the small ball, i.e. for every $\bm x\in\ball(\bm 0,\srad),|g(\bm x)-h(\bm x)|\leq 70\delta n^{1.5}$. Since $f$ is $\varepsilon$-far from any additive function, we have $f$ is $\varepsilon$-far from $h$.

  Now, we want to bound the probability that Step~2 of Algorithm~\ref{alg:zero-mean-additivity} passes.
  First, we bound the probability that \Call{Approximate-$g$}{$\bm p, f$} fails to recover the value of $g(\bm p)$ within an error of $4\delta$. That is, we bound the probability that $|g_{\bm x_1}(\bm p/\kappa_{\bm p})- g_{\bm x_j}(\bm p/\kappa_{\bm p})|\leq 2\delta$, for all $j \in [ N'_{\ref{alg:subroutines_noisy}}]$ (so that it doesn't reject), but $\left|g(\bm p/\kappa_{\bm p}) - g_{\bm x_1}(\bm p/\kappa_{\bm p}) \right|>4\delta$, by the probability that for all sampled vectors $\bm x_i,i\in  \left[ N'_{\ref{alg:subroutines_noisy}} \right]$, $\left|g(\bm p/\kappa_{\bm p})-g_{\bm{x}_i}(\bm p/\kappa_{\bm p})\right|\geq 4\delta$.
  By Lemma~\ref{lem:g-is-additive}, the probability that we draw $N'_{\ref{alg:subroutines_noisy}}$ points which satisfy this, is less than $2^{-N'_{\ref{alg:subroutines_noisy}}}$, which can be made $ \leq \varepsilon /4$ by choosing the hidden constant in $N'_{\ref{alg:subroutines_noisy}}$ to be large enough. 

  Now that we have established that --- in the case we  obtained query access to  approximate $g$ inside the small ball --- we get the correct approximation within $4\delta$ with high probability, it remains to show that we can test whether $f$ and $h$ are close.
  After arguing that $g$ is $14\delta$-additive in $\ball(\bm 0,\srad)$, it will follow using \autoref{thm:closeness_to_additivity_on_ball}, that $g$ is close to an additive function $h:\R^n\to\R$ on all points inside $\ball(\bm 0,\srad)$. 
We argue that in the YES case, if $f$ is close to some additive function then, for every $\bm p\in \ball(\bm 0,\lrad)$ (which contains the majority of the mass of the unknown distribution $\mathcal{D})$, we have $|f(\bm p)-h(\bm p)|\leq O(n^{1.5}\delta)$.

While in the NO case, since $f$ is far from any additive function, it is also far from $h$, and therefore 
\[\Pr_{\bm p\sim\mathcal{D}}[|f(\bm p)-h(\bm p)|>4515n^2\alpha]\geq \varepsilon.\]
And also 
\[\Pr_{\bm p\sim\mathcal{D}}[|f(\bm p)-h(\bm p)|>4515n^2\alpha : \bm p\in\ball(\bm 0,\lrad)]\geq \frac{3}{4}\varepsilon.\]

  
  If \Call{TestAdditivity}{$f$} passes with probability at least $1/3$, then by \autoref{lem:g-is-additive}, $g$ will be $14\delta$-additive inside $\ball(\bm 0,\srad)$,
  and for every $\bm p\in\ball(\bm 0,\srad),\Pr_{\bm p\sim\cN(\bm 0,I)}[|g(\bm p)-g_{\bm x}(\bm p)|\geq 4\delta]<1/2$.
  Consequently, by \autoref{thm:closeness_to_additivity_on_ball}, there would exist an additive function $h:\R^n\to\R$, such that for every $\bm x\in\ball(\bm 0,\srad),|g(\bm x)-h(\bm x)|\leq 70n^{1.5}\delta$. This gives us, for every $\bm p\in\ball(\bm 0,\lrad)$,
  \[|g(\bm p)-h(\bm p)|=\left| \kappa_{\bm p}g\left(\frac{\bm p}{\kappa_{\bm p}}\right)-\kappa_{\bm p}h\left(\frac{\bm p}{\kappa_{\bm p}}\right)\right|\leq 70n^{1.5}\delta\kappa_{\bm p}\leq 3500n^{1.5}\delta \lrad\leq 7000n^2\delta.\]
  
  Note that since $f$ is $\varepsilon$-far from $h$ we have 
  \begin{align*}
      &\Pr_{\bm p \sim \mathcal{D}} [\left|f(\bm p)- g(\bm p)\right|>5\delta n^{1.5}\kappa_{\bm p} : \bm p\in\ball(\bm 0,\lrad)]\\
      &\geq\Pr_{\bm p\sim\mathcal{D}}[|f(\bm p)-h(\bm p)|>21015n^2\alpha : \bm p\in\ball(\bm 0,\lrad)]\\
      &\geq \frac{3\varepsilon}{4}
  \end{align*}
  
  
  
  Indeed, the probability that Step~2 of Algorithm~\ref{alg:zero-mean-additivity} fails to reject is at most
  \begin{align*}
    & {\left( \Pr_{\bm p \sim \mathcal{D}} \left[ \left|f(\bm p)- g(\bm p)\right|\leq 5\delta n^{1.5}\kappa_{\bm p} \vee \text{\Call{Approximate-$g$}{$\bm p, f$} fails to correctly recover $g(\bm p)$} : \bm p\in\ball(\bm 0,\lrad)\right]  \right)}^{N_{\ref{alg:zero-mean-additivity}}} \\
    &\leq  \Big( 1- \Pr_{\bm p \sim \mathcal{D}} [\left|f(\bm p)- g(\bm p)\right|>5\delta n^{1.5}\kappa_{\bm p}:\bm p\in\ball(\bm 0,\lrad)]\\
    &+ \Pr_{\bm p \sim \mathcal{D}} \left[ \text{\Call{Approximate-$g$}{$\bm p, f$} fails to correctly recover $g(\bm p)$}:\bm p\in\ball(\bm 0,\lrad) \right]  \Big)^{N_{\ref{alg:zero-mean-additivity}}} \\
    &< {\left(1- \frac{3\varepsilon}{4} +\frac{\varepsilon}{4} \right)}^{N_{\ref{alg:zero-mean-additivity}}} < \frac{1}{3},
  \end{align*}
  by choosing the hidden constant in $N_{\ref{alg:zero-mean-additivity}}$ to be large enough.
  Therefore, Algorithm~\ref{alg:zero-mean-additivity} rejects with probability at least $2/3$.
\end{proof}

It now remains to prove \autoref{lem:g-is-additive}, showing that if  Algorithm~\ref{alg:zero-mean-additivity} succeeds, then $g$ is $14\delta$-additive inside $\ball(\bm 0,\srad)$, and can be well approximated in $\ball(\bm 0,\srad)$ with  high probability by querying $f$ on correlated points.

\subsection{\texorpdfstring{$O(\delta)$}{O(d)}-Additivity of \texorpdfstring{$g$}{g} Inside \texorpdfstring{$\ball(\bf 0,\it\srad)$}{Ball(0, r)}} \label{subsec:proof_of_g_additive_in_small_ball}
First, we record the basic, but useful observation that if the \Call{TestAdditivity}{} subroutine passes, then each of its tests hold with high probability over $\mathcal{N}(\bm 0,I)$.

\begin{lem}\label{lem:f-additive-whp}
  If \Call{TestAdditivity}{$f$} accepts with probability at least $1/3$, then 
  \begin{align} 
    \Pr_{\bm x \sim \mathcal{N}(\bm 0, I)}[\left|f(-\bm x) + f(\bm x)\right|\leq \delta] \geq \frac{999}{1000}, \label{eq:lem3-1} \\
  \Pr_{\bm x,y \sim \mathcal{N}(\bm 0,I)} \left [\left|f(\bm x - \bm y) - f(\bm x) + f(\bm y)\right|\leq \delta \right] \geq \frac{999}{1000}, \label{eq:lem3-3} \\
   \Pr_{\bm x,\bm y,\bm z \sim \mathcal{N}(\bm 0,I)} \left [\left|f\left(\frac{\bm x - \bm y}{\sqrt{2}} \right) - f \left(\frac{\bm x - \bm z}{\sqrt{2}} \right) - f \left(\frac{\bm z - \bm y}{\sqrt{2}} \right)\right|\leq \delta \right] \geq \frac{999}{1000}.  \label{eq:lem3-5}
  \end{align}
\end{lem}
\begin{proof}
  Suppose for contradiction that at least one of~\eqref{eq:lem3-1},~\eqref{eq:lem3-3}, and~\eqref{eq:lem3-5}  does not hold.
  We here assume that~\eqref{eq:lem3-1} does not hold as other cases are similar.

  We accept only when all the sampled points $\bm x$ satisfy $|f(-\bm x)+f(\bm x)|\leq \delta$.
  By setting the hidden constant in $N_{\ref{alg:subroutines_noisy}}$ to be large enough, this happens with probability at most
  \[ {\Big(\Pr_{\bm x \sim \mathcal{N}(0,I)} [|f(-\bm x ) + f(\bm x) |\leq \delta] \Big)}^{N_{\ref{alg:subroutines_noisy}}} < {\left( \frac{999}{1000}\right)}^{N_{\ref{alg:subroutines_noisy}}} < \frac{1}{3}, \]
  which is a contradiction.
 \end{proof}

In order to argue that $g$ is $O(\delta)$-additive on points within $\ball(\bm 0,\srad)$, we will rely on the fact that $\bm p + \bm x$ is approximately distributed as $\bm x \sim \mathcal{N}(\bm 0,I)$, if $\|\bm p\|_2$ is small.
By \autoref{lem:usefulDTVBound}, we have a bound on the total variation distance between $\bm x$ and $ \bm p+\bm x$.
Next, we will show that $g$ is $O(\delta)$-additive within $\ball(\bm 0,\srad)$.
\begin{lem}\label{lem:g-additive-in-ball}
  Suppose that $\eqref{eq:lem3-1}-\eqref{eq:lem3-5}$ of \autoref{lem:f-additive-whp} hold. Then for every $\bm p,\bm q \in \mathbb{R}^n$ with $\|\bm p\|_2, \|\bm q\|_2, \|\bm p + \bm q\|_2 \leq \srad$, it holds that \[\left|g(\bm p + \bm q) -g(\bm p) - g(\bm q)\right|\leq 14\delta.\]
\end{lem}

The proof of this lemma will crucially rely on the following two lemmas, which say that the conclusions of \autoref{lem:f-additive-whp} hold with high probability, even when one of the points are fixed to some $\bm p \in \ball(\bm 0,\srad)$. A consequence of this is that we will be able to query $g$ within a small error, with high probability.

\begin{lem}\label{lem:g-is-well-defined}
  Suppose that $\eqref{eq:lem3-1}-\eqref{eq:lem3-5}$  of \autoref{lem:f-additive-whp} hold. Then, for every $\bm p \in \mathbb{R}^n$ with $\|\bm p\|_2 \leq \srad$,
  \begin{equation}\label{eq:g_is_approximated_majority_gaussian}
      \Pr_{\bm x \sim \mathcal{N}(\bm 0,I)}[|g(\bm p) - ( f(\bm p-\bm x) + f(\bm x))|<4\delta] \geq \frac{113}{125}.
  \end{equation}
\end{lem}

The proof of this lemma will rely on an earlier stated theorem which provides a relationship between the majority and the median:

\medIsClose*

We provide a proof of this theorem  in \autoref{sec:appendixB}.
With this result in hand, we are ready to prove \autoref{lem:g-is-well-defined}.

\begin{proof}[Proof of \autoref{lem:g-is-well-defined}]
  Fix a point $\bm p \in \mathbb{R}^n$ with $\|\bm p\|_2 \leq \srad$. We will bound the following probability, which can be thought as the  approximate-collision probability. 
  \[ A:= \Pr_{\bm x,\bm y \sim \mathcal{N}(\bm 0,I)} [|(f(\bm p-\bm x) + f(\bm x))-( f(\bm p -\bm y) +f(\bm y))|\leq 4\delta].\]
  Observe that
  \begin{align*}
    1-A & = \Pr_{\bm x,\bm y \sim \mathcal{N}(\bm 0,I)}[|f(\bm x) - f(\bm y) -f(\bm p - \bm y) + f(\bm p-\bm x)|\geq 4\delta ]\\
    & \leq \Pr_{\bm x,\bm y \sim \mathcal{N}(\bm 0,I)}[|f(\bm x-\bm y) - f(\bm p-\bm y) + f(\bm p-\bm x)|>3\delta]  \\
    & + \Pr_{\bm x,\bm y \sim \mathcal{N}(\bm 0,I)}[|f(\bm x) - f(\bm y) - f(\bm x-\bm y)|>\delta ] \tag{By Triangle Inequality} \\
    & < \Pr_{\bm x,\bm y \sim \mathcal{N}(\bm 0,I)}[|f(\bm x-\bm y) - f(\bm p-\bm y) + f(\bm p-\bm x)|>3\delta] + \frac{1}{1000} \tag{By \autoref{lem:f-additive-whp} \eqref{eq:lem3-3}}
  \end{align*}
  To bound the first term, we observe, by the fact that $\bm x-\bm p,\bm y-\bm p\sim\cN(-\bm p, I)$ and  $\bm p \approx 0$, the random variables $\bm x-\bm p$ and $\bm y-\bm p$ should be distributed similarly to $\bm x$ and $\bm y$. Indeed,
  {\allowdisplaybreaks
  \begin{align*}
    &\Pr_{\bm x,\bm y \sim \mathcal{N}(\bm 0,I)}[|f(\bm x-\bm y) -( f(\bm p-\bm y) - f(\bm p-\bm x))|>3\delta] \\
    =& \Pr_{\bm x,\bm y \sim \mathcal{N}(\bm 0,I)}[|f(\underbrace{\bm x-\bm p}_{\triangleq \bm{\tilde{x}}} -(\underbrace{\bm y-\bm p}_{\triangleq \bm{\tilde{y}}})) -( f(\bm p-\bm y) - f(\bm p-\bm x))|>3\delta] \\
    =&\Pr_{\bm {\tilde{x}},\bm{\tilde{y}} \sim \mathcal{N}(-\bm p,I)}[|f(\bm{\tilde{x}} - \bm{\tilde{y}}) - (f(-\bm{\tilde{y}}) - f(-\bm{\tilde{x}}))|>3\delta] \\
    \leq & \Pr_{\bm{\tilde{x}},\bm{\tilde{y}} \sim \mathcal{N}(\bm 0,I)}[|f(\bm{\tilde{x}}-\bm{\tilde{y}}) -( f(-\bm{\tilde{y}}) - f(-\bm{\tilde{x}}))|>3\delta] + 2\dtv \Big(\mathcal{N}(\bm 0, I), \mathcal{N}(-\bm p, I) \Big)  \\
     \leq & \Pr_{\bm{\tilde{x}},\bm{\tilde{y}} \sim \mathcal{N}(\bm 0,I)}[|f(\bm{\tilde{x}}-\bm{\tilde{y}}) - f(\bm{\tilde{x}}) + f(\bm{\tilde{y}})|>\delta]  +  \Pr_{\bm{\tilde{x}} \sim \mathcal{N}(\bm 0, I)}[ |f(-\bm{\tilde{x}}) - f(\bm{\tilde{x}})|>\delta]\\
     & +\Pr_{\bm{\tilde{y}} \sim \mathcal{N}(\bm 0, I)}[ |f(-\bm{\tilde{y}}) - f(\bm{\tilde{y}})|>\delta] + \frac{1}{50} \tag{By Triangle Inequality, and \autoref{lem:usefulDTVBound}}\\
    \leq &\; \frac{3}{1000} +  \frac{1}{50} \leq \frac{23}{1000}. \tag{By  \autoref{lem:f-additive-whp}$~\eqref{eq:lem3-1}-\eqref{eq:lem3-5}$}
  \end{align*}}
  Plugging this into our previous bound on $A$, we have
 
   \[ \Pr_{\bm x,\bm y\sim\cN(\bm 0,I)}[|g_{\bm x}(\bm p)-g_{\bm y}(\bm p)|\leq 4\delta] \geq 1- \left(\frac{1}{1000} + \frac{23}{1000} \right) = 1 - \frac{3}{125},  \]
Applying \autoref{thm:median-is-close}, we conclude that for every $\bm p\in\ball(\bm 0,\srad)$ 

  \[\Pr_{\bm q \sim \mathcal{N}(\bm 0, I)} [|g(\bm p) - (f(\bm p -\bm q) +f(\bm q))|\leq 4\delta] \geq 1-4 \cdot \frac{3}{125} = \frac{113}{125}.\]
\end{proof}

The following lemma is essentially condition~\eqref{eq:lem3-3} of ~\autoref{lem:f-additive-whp} with two fixed points.

\begin{lem}\label{lem:whp-g-additive-two-points}
  Suppose that $\eqref{eq:lem3-1}-\eqref{eq:lem3-5}$ of \autoref{lem:f-additive-whp} hold then, for every $\bm p,\bm q \in \mathbb{R}^n$ with $\|\bm p\|_2, \|\bm q\|_2, \|\bm p+\bm q\| \leq \srad$, it holds that
  \[ \Pr_{\bm{x,y,z} \sim \mathcal{N}(\bm 0, I)} \left[ \left|g(\bm p + \bm q) -\left( f \left(\bm p - \frac{\bm x-\bm z}{\sqrt{2}} \right) + f \left( \bm q - \frac{\bm z-\bm y}{\sqrt{2}} \right) + f \left( \frac{\bm x-\bm y}{\sqrt{2}} \right)\right)\right|>5\delta \right] <  \frac{177}{1000}. \]
\end{lem}

\begin{proof}
  Fix a pair of points $\bm p,\bm q \in \mathbb{R}^n$ with $\|\bm p\|_2, \|\bm q\|_2 \leq \srad$. We can bound the probability
  \begin{align*}
    &\Pr_{\bm x,\bm y,\bm z \sim \mathcal{N}(\bm 0, I)} \left[\left| g(\bm p + \bm q) -\left( f \left(\bm p - \frac{\bm x-\bm z}{\sqrt{2}} \right) + f \left( \bm q - \frac{\bm z-\bm y}{\sqrt{2}} \right) + f \left( \frac{\bm x-\bm y}{\sqrt{2}} \right)\right)\right|>5\delta\right] \\
    \leq & \Pr_{\bm x,\bm y \sim \mathcal{N}(\bm 0, I)} \left[ \left|g(\bm p+\bm q) -\left( f \left(\bm  p+\bm q - \frac{\bm x-\bm y}{\sqrt{2}} \right) + f \left( \frac{\bm x-\bm y}{\sqrt{2}} \right)\right)\right|>4\delta\right] \\
    +& \Pr_{\bm x,\bm y,\bm z \sim \mathcal{N}(\bm 0, I)} \left[ \left|f \left(\bm  p +\bm  q - \frac{\bm x-\bm y}{\sqrt{2}} \right) -\left( f \left(\bm  p - \frac{\bm x-\bm z}{\sqrt{2}} \right) + f \left(\bm  q -\frac{\bm z-\bm y}{\sqrt{2}} \right)\right)\right|>\delta \right]
  \end{align*}
  To bound the first term, observe that if $\bm x,\bm y \sim \mathcal{N}(\bm 0, I)$, then the random variable $\bm m\triangleq(\bm x-\bm y)/\sqrt{2}\sim\mathcal{N}(\bm 0,I)$. Furthermore, because $\|\bm p+\bm q \|_2 \leq \srad$,  we can apply ~\autoref{lem:g-is-well-defined} \eqref{eq:g_is_approximated_majority_gaussian} and conclude that
  \begin{align*}
    &\Pr_{\bm x,\bm y \sim \mathcal{N}(\bm 0, I)} \left[ \left|g(\bm p +\bm  q) -\left( f \left(\bm p +\bm  q -  \frac{\bm x-\bm y}{\sqrt{2}} \right) + f \left(\frac{\bm x-\bm y}{\sqrt{2}} \right)\right)\right|>4\delta \right] \\
    =& \Pr_{\bm m\sim \mathcal{N}(\bm 0, I)} \left[ \left|g(\bm p +\bm  q) -\left( f ((\bm p +\bm  q) -  \bm m ) + f (\bm m )\right)\right|>4\delta \right]\leq \frac{12}{125}.
  \end{align*}
  To bound the second term, observe that
  \begin{align*}
    &\Pr_{\bm x,\bm y,\bm z \sim \mathcal{N}(\bm 0, I)} \left[\left| f \left(\bm  p +\bm  q - \frac{\bm x-\bm y}{\sqrt{2}} \right) -\left( f\left(\bm p- \frac{\bm x-\bm z}{\sqrt{2}} \right) + f\left(\bm q- \frac{\bm z-\bm y}{\sqrt{2}} \right)\right)\right|>\delta \right] \\
    =&\Pr_{\bm x,\bm y,\bm z \sim \mathcal{N}(\bm 0, I)} \left[ \left|f \Bigg(\frac{(\sqrt{2}\bm q +\bm y)-(\bm x-\sqrt{2}\bm p)}{\sqrt{2}} \Bigg) -\Bigg(f \Bigg( \frac{(\sqrt{2}\bm q +\bm y)-\bm z}{\sqrt{2}} \Bigg) + f \Bigg( \frac{\bm z-(\bm x-\sqrt{2}\bm p)}{\sqrt{2}} \Bigg)\Bigg)\right|>\delta \right] \\
    =&\Pr_{\substack{\tilde{\bm x}\triangleq \bm x-\sqrt{2}\bm p \sim \mathcal{N}(-\sqrt{2}\bm p,I) \\\tilde{\bm  y}\triangleq\bm y +\sqrt{2}\bm q \sim \mathcal{N}(\sqrt{2}\bm q,I) \\ \bm z \sim \mathcal{N}(\bm 0,I)}}\left[ \left|f \left(\frac{\tilde{\bm y}-\tilde{\bm x}}{\sqrt{2}} \right) -\left( f \left( \frac{\tilde{\bm y}-\bm z}{\sqrt{2}} \right) + f\left( \frac{\bm z-\tilde{\bm x}}{\sqrt{2}} \right)\right)\right|>\delta \right] \\
    \leq  &\Pr_{\tilde{\bm x},\tilde{\bm y},\bm z \sim \mathcal{N}(\bm 0, I)}\left[\left| f \left(\frac{\tilde{\bm y}-\tilde{\bm x}}{\sqrt{2}} \right) -\left( f \left( \frac{\tilde{\bm y}-\bm z}{\sqrt{2}} \right) + f \left( \frac{\bm z-\tilde{\bm x}}{\sqrt{2}} \right)\right)\right|>\delta \right] \\
    & + 2\left(\dtv \left( \mathcal{N}(\bm 0, I), \mathcal{N}(-\sqrt{2}\bm p, I )\right) + \dtv \left( \mathcal{N}(\bm 0,I), \mathcal{N}(\sqrt{2}\bm q, I )\right)\right) \\
    \leq & \frac{1}{1000}  + \frac{\sqrt{2}}{25} < \frac{81}{1000}. \tag{By \autoref{lem:f-additive-whp} \eqref{eq:lem3-5} and ~\autoref{lem:usefulDTVBound}}
  \end{align*}
  Combining both of these bounds, we have
  \[\Pr_{\bm x,\bm y,\bm z \sim \mathcal{N}(\bm 0,I)} \left[ \left|g\left(\bm p +\bm  q\right) -\left( f \left(\bm p - \frac{\bm x-\bm z}{\sqrt{2}} \right) + f \left(\bm  q - \frac{\bm z-\bm y}{\sqrt{2}} \right) + f \left(\frac{\bm x-\bm y}{\sqrt{2}} \right)\right)\right|>5\delta \right]
  < \frac{177}{1000}.\]
\end{proof}

$O(\delta)$-additivity of $g$ within $\ball(\bm 0,\srad)$ is an immediate consequence of these two lemmas.

\begin{proof}[Proof of Lemma~\ref{lem:g-additive-in-ball}]
  Let $\bm p,\bm q \in \mathbb{R}^n$ be any pair of points satisfying $\|\bm p\|_2, \|\bm q\|_2, \|\bm p + \bm q\|_2 \leq \srad$. Our aim is to show that $\left|g(\bm p + \bm q) - g(\bm p) - g(\bm q)\right|\leq 14\delta$. By a union bound we show that the probability that $\bm x,\bm y,\bm z \sim \mathcal{N}(0,I)$ simultaneously satisfy:
  {\allowdisplaybreaks
  \begin{align}
     \left|g(\bm p + \bm q) -\left( f\left(\bm p - \frac{\bm x - \bm z}{\sqrt{2}}\right) + f\left(\bm  q - \frac{\bm z - \bm y}{\sqrt{2}}\right)+ f\left(\frac{\bm x - \bm y}{\sqrt{2}}\right)\right)\right|&<5\delta,\label{itm:itm-1}\\
     \left|g(\bm p) -\left( f\left(\bm p - \frac{\bm x - \bm z}{\sqrt{2}}\right) + f\left(\frac{\bm x - \bm z}{\sqrt{2}}\right)\right)\right|&<4\delta,\label{itm:itm-2}\\
    \left| g(\bm q) -\left( f\left(\bm q - \frac{\bm z - \bm y}{\sqrt{2}}\right) + f\left( \frac{\bm z - \bm y}{\sqrt{2}}\right)\right)\right|&<4\delta,\label{itm:itm-3}\\
    \left|f\left(\frac{\bm x - \bm y}{\sqrt{2}}\right) -\left(f\left(\frac{\bm x - \bm z}{\sqrt{2}}\right) - f \left(\frac{\bm z - \bm y}{\sqrt{2}}\right)\right)\right|&<\delta\label{itm:itm-4}
  \end{align}}
  is at least $1- (177/1000+ 2\cdot 12/125  + 1/1000) =63/100  > 0$. Probabilities for \eqref{itm:itm-1}, and \eqref{itm:itm-4} follow by \autoref{lem:whp-g-additive-two-points}, and \autoref{lem:f-additive-whp} \eqref{eq:lem3-3}, respectively. For \eqref{itm:itm-2} and \eqref{itm:itm-3} we are using the fact that $(\bm x - \bm z)/\sqrt{2}\text{, and }(\bm z - \bm y)/\sqrt{2}$ are distributed as $\mathcal{N}(0,I)$ and apply \autoref{lem:g-is-well-defined}~\eqref{eq:g_is_approximated_majority_gaussian}.
  Fixing such a triple $(\bm x,\bm y,\bm z)$, we conclude that
  {\allowdisplaybreaks
  \begin{align*}
    |g(\bm p + \bm q)-g(\bm p)-g(\bm q)| \leq & 
    \left|g ( \bm p + \bm q) - \left(f \left(\bm p - \frac{\bm x - \bm z}{\sqrt{2}} \right) + f \left( \bm q - \frac{\bm z - \bm y}{\sqrt{2}} \right) + f \left (\frac{\bm x - \bm y}{\sqrt{2}} \right)\right)\right| \\
    &+\left| f \left(\bm p - \frac{\bm x - \bm z}{\sqrt{2}} \right) +f \left( \bm q - \frac{\bm z - \bm y}{\sqrt{2}} \right) +f \left (\frac{\bm x - \bm y}{\sqrt{2}} \right)-g(\bm p) - g(\bm q) \right| \\
    \leq & 5\delta+\left| f \left(\bm p - \frac{\bm x - \bm z}{\sqrt{2}} \right)+ f \left(\frac{\bm x - \bm z}{\sqrt{2}}\right) -g(\bm p)\right|\\
    &+\left|f \left( \bm q - \frac{\bm z - \bm y}{\sqrt{2}} \right)+  f \left( \frac{\bm z - \bm y}{\sqrt{2}} \right)- g(\bm q)\right|\\
    &+\left|f \left (\frac{\bm x - \bm y}{\sqrt{2}} \right)- f \left(\frac{\bm x - \bm z}{\sqrt{2}} \right) -  f \left( \frac{\bm z - \bm y}{\sqrt{2}} \right)\right|
    \leq  14\delta.
  \end{align*}}
  Therefore, $g$ is $14\delta$-additive within $\ball(\bm 0,\srad)$.
\end{proof}

With this we are ready to prove \autoref{lem:g-is-additive}. 

\begin{proof}
  [Proof of \autoref{lem:g-is-additive}]
  $g$ is $14\delta$-additive by \autoref{lem:g-additive-in-ball}. And, 
  $g_{\bm x}(\bm p)=f(\bm p-\bm x)+f(\bm x)$ is a good estimation (up to $4\delta$) for $g(\bm p)$ with high probability $\left(\frac{113}{125}>\frac{1}{2}\right)$ for $x\sim\cN(\bm 0,I)$ by \autoref{lem:g-is-well-defined}\eqref{eq:g_is_approximated_majority_gaussian}.
\end{proof}

\subsection{Multiplicatively-Approximate Distribution-Free Additivity Tester}\label{subsec:distribution-free-additivity}

In this section, we show that a small adaption of our tester give us a distribution-free tester for multiplicatively approximate additivity, without any precondition on the unknown distribution $\mathcal D$ (such as assuming that it is concentrated).
After removing the condition of sampled points being inside  $\ball(\bm 0,\lrad)$, the adapted tester is represented in \autoref{alg:distribution-free-additivity}. 
\begin{algorithm}
  \caption{Distribution-Free Approximate Additivity Tester With Multiplicative Error}\label{alg:distribution-free-additivity}
  \Procedure{\emph{\Call{MultApproxAdditivityTester}{$f,\mathcal{D},\alpha,\eps,\lrad$}}}{
    \Given{Query access to $f\colon \mathbb{R}^n \to \mathbb{R}$, sampling access to an unknown $(\eps/4,\lrad)$-concentrated distribution $\mathcal{D}$, noise parameter $\alpha>0$, farness parameter $\delta>0$;}
    $\delta\gets 3\alpha,r\gets 1/50$\;
    \textbf{Reject} if \Call{TestAdditivity}{$f,\delta$} returns \textbf{Reject}\;
    \For{$N_{\ref{alg:distribution-free-additivity}} \gets O(1/\varepsilon)$ times}{
      Sample $\bm p \sim \mathcal{D}$\;
      \textbf{Reject} if $|f(\bm p) -$ \Call{Approximate-$g$}{$\bm p,f,\delta$}$|>5\delta n^{1.5}\kappa_{\bm p}$  or if \Call{Approximate-$g$}{$\bm p,f,\delta$} returns \textbf{Reject}.
    }
    \textbf{Accept}.}
\end{algorithm}

We note that the subroutines in \autoref{alg:subroutines_noisy} remain the same and still sample points from $\cN(\bm 0,I)$, in order to check that $f$ satisfies the characterization properties, and to approximate $g$ inside $\ball(\bm 0,\srad)$.

\paragraph{Distribution-Free Multiplicatively-Approximate Tester for Additivity.} 
Given query access to the input function $f: \R^n \to \R$, sampling access to unknown distribution $\mathcal{D}$, as well to $\cN(\bm 0, I)$, a parameter $0<\alpha\in \R$ and  a constant $0<\varepsilon \in \R$, a distribution-free, multiplicative-approximate tester for additivity distinguishes between the following two cases with probability at least $2/3$:
\begin{itemize}
\item
\textbf{\textsc{Yes Case:}} There exists an additive function $h: \R^n \to \R$ such that for all $\bm{p} \in \R^n$:
\[
|f(\bm{p})-h(\bm{p})| \leq \alpha;
\]
\item
\textbf{\textsc{No Case:}}  For any additive function  $h: \R^n \to \R$:
\[
\Pr_{\bm{p} \sim \mathcal{D}}[|f(\bm{p})-h(\bm{p})|>600 n^{1.5}\alpha\kappa_{\bm p}] > \eps.
\]
\end{itemize}

 Correctness of our tester, given in \autoref{alg:distribution-free-additivity}, follows from this theorem.

\begin{proof}
  The proof follows the same path as for \autoref{thm:central-gausssian}. We only adapt the \autoref{alg:distribution-free-additivity} to now test all points sampled by $\mathcal D$. 
  In the \textbf{\textsc{Yes case}}, the tests always accept. Indeed the \Call{TestAdditivity}{$f$} subroutine passes with probability $1$, and we claim \Call{Approximate-$g$}{$\bm p,f$} never rejects and returns an approximate value $\kappa_{\bm p} g_{\bm x_1 }\left(\frac{\bm p}{\kappa_{\bm p}}\right)$, when queried $g(\bm p)= \kappa_{\bm p} g\left(\frac{\bm p}{\kappa_{\bm p}}\right)$, where $\bm x_1\sim\cN(\bm 0,I)$. Recall that in the \textbf{\textsc{Yes case}}, $\left| g_{\bm x_1}\left(\frac{\bm p}{\kappa_{\bm p}}\right)-f\left(\frac{\bm p}{\kappa_{\bm p}}\right)\right|\leq \delta=3\alpha$. 
  Therefore we have, by triangle inequality, for every $\bm p,\bm x_1\in \R^n$, 
  \[\left| f(\bm p)- \kappa_{\bm p} g_{\bm x_1}\left(\frac{\bm p}{\kappa_{\bm p}}\right)\right| \leq \left|f(\bm p)-\kappa_{\bm p} f\left(\frac{\bm p}{\kappa_{\bm p}}\right)\right| + \kappa_{\bm p}\left| f\left(\frac{\bm p}{\kappa_{\bm p}}\right) - g_{\bm x_1}\left(\frac{\bm p}{\kappa_{\bm p}}\right)\right| \leq \alpha + \alpha\kappa_{\bm p} + \delta\kappa_p\leq 2\delta\kappa_p.  \]
  Last inequality by $f$ being point-wise close to an additive function. Thus Step 4 in \autoref{alg:distribution-free-additivity} always passes.

  In the \textbf{\textsc{No Case}}, we reject with probability at least $2/3$.
  Indeed, if \Call{TestAdditivity}{$f$} rejects with probability $>2/3$ we are done. So, assume that it accepts with probability at least $1/3$, then we see that the premise of  \autoref{lem:g-is-additive} holds. 
  
  We first bound the probability of Step 4 of \autoref{alg:distribution-free-additivity} to pass. For this we use the fact that the probability that $\tilde{g}(\bm p)\triangleq$\Call{Approximate-$g$}{$\bm p,f$} fails to approximate $g$ withing $6\delta$ error is at most $\frac{\varepsilon}{2}$ as we proved for \autoref{thm:central-gausssian}.
  {\allowdisplaybreaks
  \begin{align*}
      \Pr_{\bm p\sim\mathcal{D}}[\text{Step 6 passes}]
      &\leq \Pr_{\bm p\sim\mathcal{D}}\left[\left|f(\bm p)-\tilde{g}(\bm p)\right|<5\delta n^{1.5}\kappa_{\bm p}\right]\\
      &\leq \Pr_{\bm p\sim\mathcal{D}}\left[\left|f(\bm p)-g(\bm p)\right|<20\delta n^{1.5}\kappa_p \vee \left|\tilde{g}(\bm p)-g(\bm p)\right|>6\delta\right]\\
     & \leq 1-\Pr_{\bm p\sim\mathcal{D}}\left[\left|f(\bm p)-g(\bm p)\right|\geq 20\delta n^{1.5}\kappa_p \right]+\Pr_{\bm p\sim\mathcal{D}}\left[ \left|\tilde{g}(\bm p)-g(\bm p)\right|>6\delta\right]\\
     &\leq 1-\Pr_{\bm p\sim\mathcal{D}}\left[\left|f(\bm p)-g(\bm p)\right|\geq 20\delta n^{1.5}\kappa_p \right]+\frac{\varepsilon}{2}\\
     &\leq 1-\frac{\varepsilon}{2}.
  \end{align*}
  }
 For the last inequality we have to bound the probability that $f$ and $g$ are far, say $$\Pr_{\bm p\sim\mathcal{D}}\left[\left|f(\bm p)-g(\bm p)\right|\geq 20\delta n^{1.5}\kappa_p \right]\geq \varepsilon,$$for that we use  \autoref{thm:closeness_to_additivity_on_ball} to show that there  exist an additive function $h:\R^n\to\R$, such that for every $\bm x\in\ball(\bm 0,\srad),|g(\bm x)-h(\bm x)|\leq 150n^{1.5}\delta$. This gives us, for every $\bm p\in\R^n$,
  \[|g(\bm p)-h(\bm p)|=\left| \kappa_{\bm{p}}g\left(\frac{\bm p}{\kappa_{\bm{p}}}\right)-\kappa_{\bm{p}}h\left(\frac{\bm p}{\kappa_{\bm{p}}}\right)\right|\leq 150n^{1.5}\delta\kappa_{\bm{p}}.\]
  
  Note that since $f$ is $\varepsilon$-far from any additive function, it is also $\varepsilon$-far from $h$ and with probability $\varepsilon$ we draw $\bm p\sim\mathcal {D}$ that satisfies $|f(\bm p)-h(\bm p)|>200\delta n^{1.5}\kappa_{\bm p}$. For these $\bm p$, it holds that 
  \[200\delta n^{1.5}\kappa_{\bm p}<|f(\bm p)-h(\bm p)|<|f(\bm p)-g(\bm p)|+|g(\bm p)-h(\bm p)|\leq |f(\bm p)-g(\bm p)|+150\delta n^{1.5}\kappa_{\bm p},\]
  implying that $|f(\bm p)-g(\bm p)|>50\delta n^{1.5}\kappa_{\bm p}$.
\end{proof}

\section{Proof of \autoref{thm:median-is-close}}
\label{sec:appendixB}

In this appendix we prove the following lemma which gives a sufficient condition for the median of any distribution to be close to a random element sampled from that distribution. 
\medIsClose*

\begin{proof}
Define $S_{\leq}\triangleq\{\bm q\in\Omega:g(\bm q)\leq g_{\med}\}$, and $S_{\geq}\triangleq\{\bm q\in\Omega:g(\bm q)\geq g_{\med}\}$. Since, $g_{\med}$ is the median of the set $\{g(\bm q):\bm{q}\in\Omega\}$ over $\bm q\sim\caD$, 
\begin{equation}\label{eq:equiprob}
    \Pr_{\bm q\sim\caD}\left[\bm q\in S_{\leq}\right]=\Pr_{\bm q\sim\caD}\left[\bm q\in S_{\geq}\right]=\frac{1}{2}.
\end{equation}
Suppose for contradiction that the following hold:
\begin{equation}\label{eq:premise}
\begin{split}
    & \Pr_{\bm{q}_1,\bm{q}_2\sim\caD}[|g(\bm q_1)-g(\bm q_2)|<\delta]  >1-\eta, 
\end{split}
\end{equation}
\begin{equation}\label{eq:conclusion}
\begin{split}
    &\Pr_{\bm{q}_1\sim\caD}[|g_{\med}-g(\bm q_1)|<\delta]\leq 1-4\eta,
\end{split}
\end{equation}
By \eqref{eq:conclusion}, for $\bm q_1\sim\caD$, $g(\bm q_1)$ will be at least $\delta$-far from $g_{\med}$ with probability more than $4\eta$. We will argue that for $\bm q_2\sim\caD$, $g(\bm q_2)$ will be at least $\delta$-far from $g(\bm q_1)$, with probability more than $\eta$, contradicting \eqref{eq:premise}.

Suppose that $\bm q_1\in S_{\leq}$. Then, for any $\bm q_2\in S_{\geq}$, we have 
\[ |g(\bm q_2)-g(\bm q_1)|\geq |g(\bm q_2)-g_{\med}| +  |g_{\med}-g(\bm q_1)|\geq 0 + \delta = \delta.\]

Similarly, if $\bm q_1\in S_{\geq}$, then for any $\bm q_2\in S_{\leq}$, 
\[ |g(\bm q_2)-g(\bm q_1)|\geq |g(\bm q_2)-g_{\med}| +  |g_{\med}-g(\bm q_1)|\geq 0 + \delta = \delta.\]

Therefore, 
\begin{align*}
    &\Pr_{\bm q_1,\bm q_2\sim\caD}[|g(\bm q_1)-g(\bm q_2)|\geq\delta]\\
    &\geq\Pr_{\bm q_1\sim\caD}[\bm q_1\in S_{\leq}]\cdot\Pr_{\bm{q}_1\sim\caD}[|g_{\med}-g(\bm q_1)|\geq\delta\mid\bm q_1\in S_{\leq}]\cdot\Pr_{\bm q_2\sim\caD}[\bm q_2\in S_{\geq}]\\
    &+\Pr_{\bm q_1\sim\caD}[\bm q_1\in S_{\geq}]\cdot\Pr_{\bm{q}_1\sim\caD}[|g_{\med}-g(\bm q_1)|\geq\delta\mid\bm q_1\in S_{\geq}]\cdot\Pr_{\bm q_2\sim\caD}[\bm q_2\in S_{\leq}]\\
    &=\frac{1}{4}\left(\Pr_{\bm{q}_1\sim\caD}[|g_{\med}-g(\bm q_1)|\geq\delta\mid\bm q_1\in S_{\leq}] + \Pr_{\bm{q}_1\sim\caD}[|g_{\med}-g(\bm q_1)|\geq\delta\mid\bm q_1\in S_{\geq}]\right)  \tag{By \eqref{eq:equiprob}}\\
    &=\frac{1}{4}\left(\Pr_{\bm{q}_1\sim\caD}[|g_{\med}-g(\bm q_1)|\geq\delta]\right)\\
    &>\frac{4\eta}{4}=\eta. \tag{By \eqref{eq:conclusion}} 
\end{align*} \qedhere
\end{proof}

\ignore{
An immediate corollary of \autoref{thm:median-is-close} is the following. 
\begin{restatable}{cor}{medIsCloseGaussian}\label{cor:Median-of-scaled-gaussian-is-also-close}
Let $\Omega$ be a sample space, $g:\Omega\to\R$, and $\caD,\caD'$ be any distributions over $\Omega$. For every $\eta \in [0, 1/5]$, and any $\delta \in \mathbb{R}$, if both \begin{align*}
    \Pr_{\bm{q}_1,\bm q_2\sim\caD}[|g(\bm q_1)-g(\bm q_2)|<\delta] &>1-\eta, \mbox{ and} \\
    \Pr_{\substack{\bm{q}_1\sim\caD'\\ \bm q_2\sim\caD}}[|g(\bm q_1)-g(\bm q_2)|<\delta] &>1-\eta,
\end{align*}
then $ \Pr_{\bm{q}\sim\caD'}[|g_{\med}-g(\bm q)|<2\delta]>1-5\eta$, where $g_{\med}\triangleq\med_{\bm q\sim\caD}\{g(\bm q)\}$. 
\end{restatable}
\begin{proof}
Define the set $S\triangleq\{\bm q\in\Omega:|g_{\med}-g(\bm q)|\geq\delta\}$. Since we have
\begin{equation*}
\begin{split}
    &\Pr_{\bm{q}_1,\bm q_2\sim\caD}[|g(\bm q_1)-g(\bm q_2)|<\delta]>1-\eta,\\
    \equiv &\Pr_{\bm q_1,\bm q_2\sim\caD}[|g(\bm q_1)-g(\bm q_2)|\geq\delta]\leq\eta,
\end{split}
\end{equation*}
by \autoref{thm:median-is-close}, we claim
\begin{equation}
\begin{split}
    \Pr_{\bm q\sim\caD}[\bm q\in S]\leq 4\eta.
\end{split}
\end{equation}
If we randomly sample vectors $\bm {q_1}\sim\caD',\bm {q_2}\sim\caD$, by our premise, $g(\bm q_1)$ is at least $\delta$-far from $g(\bm q_2)$, with probability at most $\eta$, and $g(\bm q_2)\in S$, i.e., $g(\bm q_2)$ itself is at least $\delta$-far from $g_{\med}$, with probability at most $4\eta$. By the triangle inequality, and a union bound, we have
\begin{equation*}
\begin{split}
    \Pr_{\bm q_1\sim\caD'}[|g_{\med}-g(\bm q_1)|\geq 2\delta]& \leq\Pr_{\substack{\bm q_1\sim\caD'\\ \bm q_2\sim\caD}}[|g(\bm q_1)-g(\bm q_2)|\geq\delta] + \Pr_{\bm q_2\sim\caD}[\bm q_2\in S]\\
    &\leq \eta + 4\eta = 5\eta.
\end{split}
\end{equation*}
It thus follows that $\Pr_{\bm q\sim\caD'}[|g_{\med}-g(\bm q)|< 2\delta] >1-5\eta$, completing the proof.
\end{proof}
}

